\documentclass[10pt,conference]{IEEEtran}
\pdfoutput=1

\usepackage{booktabs} 
\usepackage{amsmath}
\usepackage{numprint}
\usepackage{amsthm}
\usepackage{amssymb}
\usepackage[binary-units=true]{siunitx}
\usepackage[utf8]{inputenc}
\usepackage[english]{babel}
\usepackage{color}
\usepackage{dsfont}
\usepackage[export]{adjustbox}
\usepackage{framed}
\usepackage{graphicx}
\usepackage{xspace}
\usepackage{latexsym}
\usepackage{listings}
\usepackage{multirow}
\usepackage{blindtext}
\usepackage{rotating}
\usepackage{placeins}
\usepackage{subcaption}
\usepackage{tabto}
\usepackage{tabularx}
\usepackage{rotating}
\usepackage{url}
\usepackage[vlined, linesnumbered, ruled, resetcount]{algorithm2e}
\usepackage{enumitem}
\usepackage{widetable}
\usepackage{tagging}

\let\oldnl\nl
\newtheorem{theorem}{Theorem}[section] 
\newtheorem{definition}[theorem]{Definition} 
 
\SetKwComment{Comment}{//}{}

\makeatletter
\newcommand{\linebreakand}{%
  \end{@IEEEauthorhalign}
  \hfill\mbox{}\par
  \mbox{}\hfill\begin{@IEEEauthorhalign}
}
\makeatother

\begin{document}
\title{Efficient Process-to-Node Mapping Algorithms for Stencil Computations}
\author{
	\IEEEauthorblockN{Sascha Hunold}
	\IEEEauthorblockA{TU Wien, Faculty of Informatics\\
	  Vienna, Austria\\
	  \url{hunold@par.tuwien.ac.at}
	}
	\and
	\IEEEauthorblockN{Konrad von Kirchbach}
	\IEEEauthorblockA{TU Wien, Faculty of Informatics\\
	  Vienna, Austria\\
	  \url{kirchbach@par.tuwien.ac.at}
	}
	\and
	\IEEEauthorblockN{Markus Lehr}
	\IEEEauthorblockA{TU Wien, Faculty of Informatics\\
	  Vienna, Austria\\
	  \url{lehr@par.tuwien.ac.at}
	}
	\linebreakand
	\IEEEauthorblockN{Christian Schulz}
	\IEEEauthorblockA{
	  University of Vienna, Faculty of Computer Science\\
	  Vienna, Austria\\
	  \url{christian.schulz@univie.ac.at}
	}
	\and
	\IEEEauthorblockN{Jesper Larsson Tr\"aff}
	\IEEEauthorblockA{TU Wien, Faculty of Informatics\\
	  Vienna, Austria\\
	  \url{traff@par.tuwien.ac.at}
	}
}

\maketitle
\begin{abstract}
	Good process-to-compute-node mappings can be
	decisive for well performing HPC applications. A special, important
	class of process-to-node mapping problems is the problem of mapping
	processes that communicate in a sparse stencil pattern to Cartesian grids.  
	By thoroughly exploiting the inherently present structure in this type
	of problem, we devise three novel distributed algorithms that are
	able to handle arbitrary stencil communication patterns effectively. 
	We analyze the expected performance of our algorithms based on an
	abstract model of inter- and intra-node communication. 
        An
	extensive experimental evaluation on several HPC machines shows that   
	our algorithms are up to two orders of magnitude faster in
	running time than a (sequential) high-quality general graph
	mapping tool, while obtaining similar results in communication
	performance. Furthermore,
	our algorithms also achieve significantly better mapping quality
	compared to previous state-of-the-art Cartesian grid mapping
	algorithms.
	This results in up to a threefold performance improvement of an \texttt{MPI\_\-Neighbor\_\-alltoall}\xspace
	exchange operation.
	Our new algorithms can be used to implement the \texttt{MPI\_\-Cart\_\-create}\xspace functionality.
\end{abstract}

  
\section{Introduction}

The communication performance of applications running on High-Performance Computing (HPC) systems
depends on a variety of factors like the
capability and topology of the underlying communication system, the
required communication (patterns, frequencies, volumes, and
dependencies) between processes, and the
software and algorithms used to realize the communication.  If the
communication pattern is known, and if a hardware topology description
is given, it is natural to attempt to find a good mapping of the
application processes onto the hardware processors such that pairs of
processes that frequently communicate large amounts of data
become~located~closely.

Many important scientific computing applications involve stencil
computations.  For example, stencil computations are used for climate
and ocean modeling \cite{haralick1992computer}, in computational
electromagnetic codes
\cite{taflove2005computational,DBLP:journals/jcphy/UstyugovPKN09}, for
image-processing \cite{sawdey1995design}, in Jacobi or multigrid
solvers~\cite{DBLP:conf/ipps/RenganarayananHDR07}, for earthquake
simulations \cite{DBLP:conf/sc/ChristenSC12} or in general in
simulations systems such as
OpenLB~\cite{DBLP:conf/europar/FietzKSSH12}.  In most cases, elements
of a $d$-dimensional matrix are repeatedly updated using the values of
fixed \emph{stencil} pattern of neighboring elements. When run on a
parallel computer, this yields communication patterns that are very
regular and more or less symmetric depending on the organization of
the processors. More precisely, each processing element exchanges data
repeatedly with a small set of neighboring processing elements, and
all processing element neighborhoods have the same structure. In this
situation, in each exchange step, all processes communicate with
other processes and all follow the same pattern determined by the
computational stencil and the organization of the processes.
\begin{figure}[t]
	\center
	\includegraphics[trim = 0 0 0 0, clip, height=4.5cm]{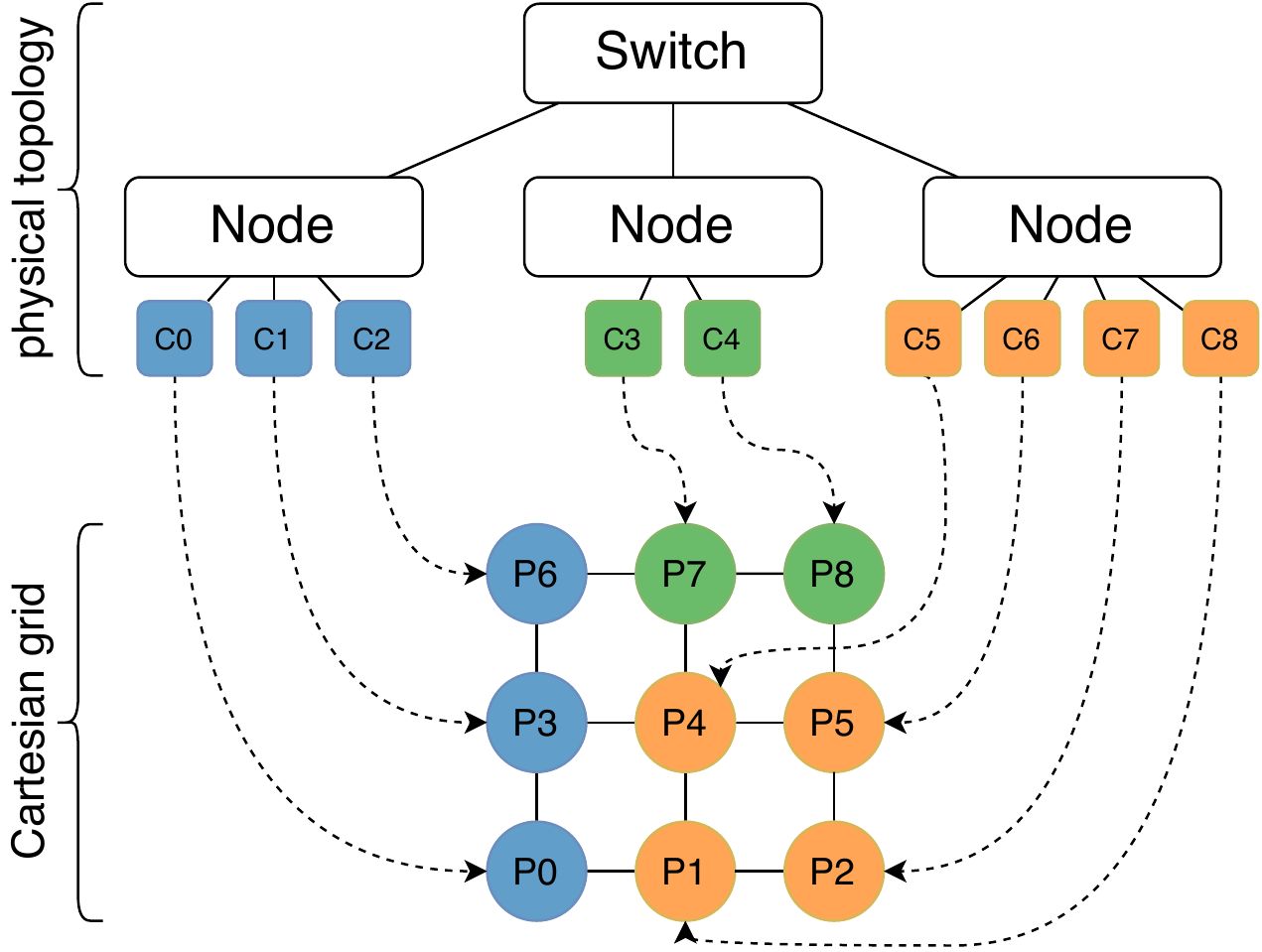}
	\caption{Motivational example: Given a set of compute nodes
          with possibly different numbers of processes per node, and a
          computational grid, find a mapping of the processes on the nodes
          to the grid, \textrm{s.t.}\xspace the number of communication edges between
          compute nodes is minimized.}
	\label{fig:motivation}
\end{figure}

The Message Passing Interface (MPI)~\cite{MPI-3.1} supports complex
communication patterns by providing functions to specify virtual
process topologies by process neighborhoods.  Using for instance
Cartesian topologies the user can refer to processes by rank or by
coordinate vectors. Moreover, MPI supports neighborhood collective
operations such as \texttt{MPI\_\-Neighbor\_\-alltoall}\xspace which make it possible for the
MPI library to exploit (regular) communication patterns to provide
more efficient data exchange operations.  MPI also
defines functionality to reorder processes in order to optimize the
communication performance, however, at the moment most MPI libraries
do not actually perform such remapping (in the general case). \\

\noindent \textbf{Contribution.} We make the following contributions:
\begin{itemize}
\item We show that the general Cartesian mapping problem under stencil
  patterns is NP-hard.
  This result motivates
  our work on heuristic algorithms for the problem.
\item We present new algorithms for the process-mapping problem for
  stencil patterns which in contrast to previous solutions are also
  applicable to cases where 1) the number of \texttt{MPI}\xspace processes per node is
  different and 2) where the number of processes is not factorizable
  or divisible by the number of processes per node, and 3) consider
  the case of arbitrary stencil patterns, not only the nearest-neighbor stencils
  implied by the \texttt{MPI}\xspace specification.
\item We perform an extensive experimental evaluation and benchmark
  the time needed for an \texttt{MPI\_\-Neighbor\_\-alltoall}\xspace operation. The results
  show that our algorithms significantly outperform previous solutions in terms of
  communication performance as well as initialization time.
 For example, our algorithms are up to two orders of magnitude faster in running time than the (sequential) 
	high-quality general graph mapping tool Vienna Mapping (\texttt{VieM}\xspace), while obtaining similar results in
	communication performance. Moreover, our algorithms are up to three times faster than other 
	Cartesian grid mapping algorithms, and achieving significantly better mapping quality.  \\
\end{itemize} 

\noindent \textbf{Organization.} The rest of the paper is organized as follows.  We start by
introducing the process-to-node mapping problem in
Section~\ref{sec:preliminaries} and discuss related work in
Section~\ref{sec:relatedwork}.  In Section~\ref{sec:hardness}, we look
at the mapping problem for Cartesian graphs, for which we show that
mapping problem for specific graph types is NP-hard.  For that reason,
we introduce three different, efficient algorithms to solve the
process-to-node mapping problem for Cartesian graphs in
Section~\ref{sec:neighborhoodaware}.  We present the results of an
extensive experimental evaluation of our novel algorithms in
Section~\ref{sec:experiments} before we conclude in Section~\ref{sec:conclusion}.

\section{Notation and Problem Formulation}
\label{sec:preliminaries}

We are considering the traditional setup of HPC system architectures,
where several compute nodes are interconnected via a high-speed
network.  The compute nodes usually comprise multiple processor-cores,
often on two or more CPU sockets. In order to solve a computational
problem on such systems, data needs to be exchanged among the
distributed processes.  We call communication between processes
residing on different compute nodes \emph{inter-node communication}
and communication between processes residing on the same compute node
\emph{intra-node communication}.  We follow the common assumption that
intra-node communication (on a compute node) is (much) faster than
inter-node communication with higher cumulated bandwidth. We
assume homogeneous communication performance between the computation nodes,
and also within the nodes~\cite{Gropp2019,Niethammer2019}.

We denote by $N$ the \emph{number of
  (compute) nodes} allocated for the application to run.  Let $p$ be
the total \emph{number of processes} of an application and $n_i$ with
$i \in \{0, 1, \dots, N-1\}$ the \emph{number of processes per node
  $i$}, \textrm{i.e.}\xspace, $\sum_{i=0}^{N-1} n_i = p$. If all the nodes have the
same number of processes (homogeneous node sizes), we denote by $n$
the number of processes on each node, \textrm{i.e.}\xspace, $p = N n$.
     
We assume that the processes are organized in a $d$-dimensional
Cartesian grid with dimension sizes $\mathcal{D}=[d_0, \dots, d_{d-1}]$, and thus,
the \emph{size of a grid} is the number of processes it comprises,
$p=\prod_{i=0}^{d-1}d_i$.  Each process with rank~$r$,
$0 \leq r < p$, is associated with a vector
$\vec{r}=[r_0, \dots, r_{d-1}]$, where $0 \leq r_i < d_i$ for
$i \in \{0, 1, \dots, d-1\}$, uniquely determining the position of
the process in the grid. W.l.o.g., processes are assigned in row-major
order to the grid.

\paragraph*{Target Stencils}

We now define three different stencils which will be used in the
remainder of the article. A $2$-d example of the considered stencils is depicted in
Figure~\ref{fig:stencils}.
To that end, we consider a
$k$-\emph{neighborhood} of a process to be a set of communication
targets which can be described as a list of relative coordinates
$\mathcal{S}=\{\vec{R_0}, \vec{R_1}, \dots, \vec{R_{k-1}}\}$.
Every~$\vec{R_i}=[R_{i,0}, R_{i,1}, \dots, R_{i,d-1}]$ with
$0 \leq i < k$ describes the relative offset along the dimensions
to the target process.  Let $\mathds{1}_i$ be a vector with the only
non-zero component being one at index $i$.
Then, we define the following stencils:
\begin{enumerate}[label=(\alph*)]
\item \emph{nearest neighbor stencil}:
  $\mathcal{S}=\{\mathds{1}_i, -\mathds{1}_i \mid  0 \leq i <
  d\}$,
\item \emph{component stencil}: $\mathcal{S}=\{\mathds{1}_i, -\mathds{1}_i \mid  0 \leq i <
  d-1\}$, and
\item \emph{nearest neighbor with hops}: $\mathcal{S}=\{\mathds{1}_i, - \mathds{1}_i \mid 0~\leq~i~<~d\} \cup \{a\mathds{1}_0, -a\mathds{1}_0 \mid \forall a \in \{2, 3\}\}$.
\end{enumerate}
\begin{figure}[t]
	\begin{subfigure}{0.20\textwidth}
		\center
	\includegraphics[trim=0 110 250 60, clip, scale=0.55]{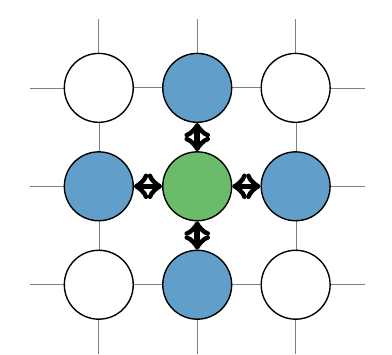}
	\caption{nearest neighbor stencil}
	\label{fig:nearestneighbor}
	\vfill
	\center
	\includegraphics[trim=0 110 250 60, clip, scale=0.55]{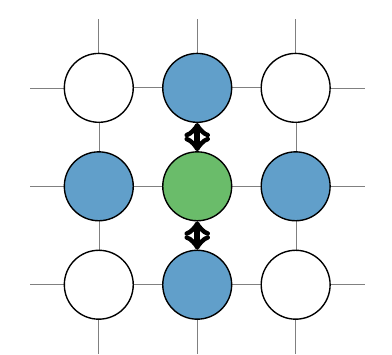}
	\caption{component stencil}
	\label{fig:component}
	\end{subfigure}
	\hfill
	\begin{subfigure}{0.30\textwidth}
		\center
                \vspace*{.5cm}
	\includegraphics[trim=0 35 250 20, clip, scale=0.55]{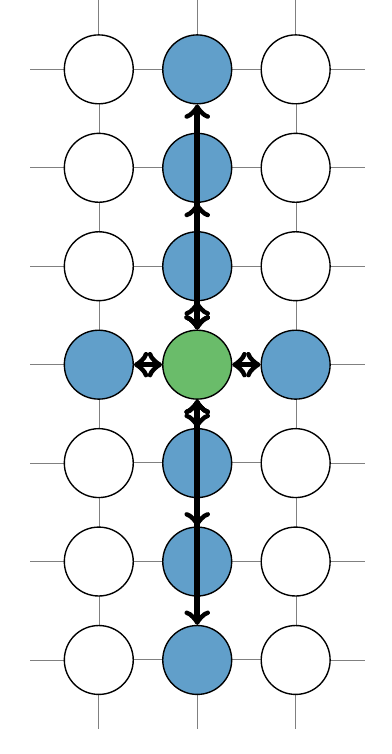}
	\caption{nearest neighbor with hops stencil}
	\label{fig:nearestneighborhops}
	\end{subfigure}
	\caption{Examples of the three $2$-d stencils.}
	\label{fig:stencils}
\end{figure}

\paragraph*{Optimization Problem}
By defining the $k$-neighborhood communication neighbors of each process in the
Cartesian process grid with dimension sizes $\mathcal{D}$,
we induce a Cartesian communication graph $\mathnormal{C}=(V, E)$ (Cartesian graph)
where $V$ is the vertex set representing the processes, \textrm{i.e.}\xspace, $|V|=p$
and $E$ is the set of communication edges between the processes. We
assume unit edge-weight and sparse communication, \textrm{i.e.}\xspace, the number of
communication neighbors is much smaller than the total number of
processes~($k \ll p$).

Let $\sigma: V \times V \rightarrow \{0,1\}$ be a cost function that
determines whether the communication between two processes $v_i$ and
$v_j$ involves two different compute nodes, \textrm{i.e.}\xspace, inter-node
communication is required.
Let $\mathcal{N}$ be the set of compute nodes and let
$M : V \rightarrow \mathcal{N}$ be a function that maps a process
$u \in V$ to exactly one compute node $\eta \in \mathcal{N}$.
For all $u, v \in E$, let~$\sigma(u,v)=0$ if $M(u) = M(v)$ and $\sigma(u,v)=1$ otherwise.
The total cost (amount) of inter-node communication operations is defined as
$J_\text{sum} := \sum_{ (u, v) \in E} \sigma(u, v)$.
We define the bottleneck node as the node with the largest number of
outgoing communication edges, \textrm{i.e.}\xspace,
$N_b := \operatorname*{argmax}_{\eta \in \mathcal{N}}\{\sum_{ (u,v) \in E}
\sigma(u, v) \mid M(u)=\eta\}$.  Let~$\mathcal{B} := \{u \mid \forall u \in V :
M(u)=N_b\}$ be the set of vertices assigned to the bottleneck node $N_b$.
Then the cost of this bottleneck node is~$J_\text{max} :=\sum_{ u\in
\mathcal{B}} \sum_{ (u,v) \in E} \sigma(u,v)$.

Our objective is to find a mapping function $M$ of processes to nodes
that minimizes $J_\text{sum}$. We use $J_\text{max}$ to distinguish
cases with similar values of $J_\text{sum}$, especially in the
experimental evaluation. Note that the original allocation
($N \times n$) given by the scheduler needs to be respected, \textrm{i.e.}\xspace, for
each node $N_i$ it must hold that $|\{u \in V \mid M(u)=N_i\}| = n_i$.

\section{Related Work}
\label{sec:relatedwork}

There has been an immense amount of research on partitioning and process
mapping -- we refer
to~\cite{GPOverviewBook,SPPGPOverviewPaper,DBLP:reference/bdt/0003S19} for
extensive material.  The problem of process reordering for different topologies
has been an
active field of research since the beginning of \texttt{MPI}\xspace
\cite{Brandfass2013,Traeff2002,Mercier2011,Hatazaki1998,YuChungMoreira06}.
Many reordering algorithms take an arbitrary unstructured graph as
input topology, making it difficult to perform efficient, scalable
mappings on a structured grid, where communication is implicitly implied
through the grid structure. In this paper we aim to exploit both stencil and
grid structure, \textrm{i.e.}\xspace, aim for specialized algorithms.

Gropp~\cite{Gropp2019} pointed out that many \texttt{MPI}\xspace implementations have not
implemented the \texttt{MPI\_\-Cart\_\-create}\xspace reordering method. As a response, he
proposed an algorithm (\texttt{Nodecart}\xspace) for homogeneous node sizes
$n$, based on the prime factorization of $n$. \texttt{Nodecart}\xspace decomposes
the dimensions into a grid spanning the nodes and a grid describing
the layout of the processes within a node. From this decomposition,
every process can calculate its new coordinate $\vec{r}$ from which it
can obtain its new rank. As a result, Gropp was able to show
significant improvements in the time needed for a nearest neighbor
message exchange in comparison to a blocked mapping of processes to nodes.
\texttt{Nodecart}\xspace was specifically designed for the implied nearest neighbor stencil of
Cartesian communicators in \texttt{MPI}\xspace.

Niethammer and Rabenseifner \cite{Niethammer2019} used a different approach to
assign processes to nodes. They point out that the \texttt{MPI\_\-Dims\_\-create}\xspace routine only
considers the total number of processes $p$ in an application from which it
finds a grid decomposition where the dimension sizes are as close as possible.
This algorithm can lead to bad domain decompositions in terms of inter-domain
communication, if the underlying data mesh is not shaped cubically.  Thus, they
propose to solve the task of grid dimension creation and process mapping
simultaneously. This is done by finding a factorization of the number nodes $N$, with the
aim to minimize the weighted communication over the domain boundaries (the
weights represent the expansion of the application mesh and a user defined
communication cost factor). Their algorithm can be extended to handle
hierarchical systems, where the weighted inter-domain communication is
minimized at each level, but it requires symmetric hierarchies.  With this
approach, they can achieve significant performance gains for a
nearest neighbor message exchange in comparison to the blocked assignment of 
processes to nodes.

Schulz et al. \cite{Schulz2017, Schulz2017a}
developed an algorithm (VieM, Vienna Mapping) for general process mapping with the objective
of minimizing the total weighted communication.
Their algorithm takes as input an unstructured communication graph
and maps it onto a hierarchical hardware graph. This is done in a 
recursive manner with perfectly balanced graph partitioning techniques
and randomized local search for improving found solutions. Even though
the approach is costly in terms of runtime and memory,  the communication cost in comparison
to the state-of-the-art has been significantly~reduced.

\section{NP-Hardness of Cartesian Mapping Problem}
\label{sec:hardness}
In general, graph embedding problems are NP-hard, as shown by
\cite{Bokhari1981}. However, the structure of the Cartesian graph
mapping problem induced by a $k$-neighborhood pattern could make the
problem easier to solve, in terms of NP-hardness or complexity of
approximation algorithms. We propose the following formal definition of the
Cartesian partitioning problem.

\begin{definition}
	Let $\mathcal{C}=(V, E)$ be a Cartesian graph with dimension sizes $\mathcal{D}$ and 
	$k$-neighborhood $\mathcal{S}$, as defined in
	Section~\ref{sec:preliminaries} and let $\mathcal{N}$ be a set of
	partition sizes (number of cores per compute node) $\mathcal{N} = n_0,
	\dots, n_{N-1}$, \textrm{s.t.}\xspace, $\sum_{i=0}^{N-1} n_i = |V|$.  Let $M: V
	\rightarrow \mathcal{N}$ be a mapping function that assigns each vertex
	$v \in V$ of the Cartesian graph to a distinct partition~$n \in
	\mathcal{N}$.
  The \texttt{GRID-\-PARTITION}\xspace problem answers the question whether there exists
	a mapping $M$ such that $J_\text{sum} \le Q$.
  
\end{definition}

\begin{definition}
  The \texttt{3-WAY-\-PARTITION}\xspace problem consists of dividing a
  multi-set of integers into three subsets, such that the sum of each
  subset is equal.  Formally, given a multi-set~$\mathcal{I}$ of
  integers, we ask whether~$\mathcal{I}$ can be partitioned into $3$
  disjoint sub-sets $I_{1, 2, 3}$, where
  $\sum_{x \in I_{1}} x = \sum_{y \in I_2} y = \sum_{z \in I_3} z$
	and~$I_1 \cup I_2 \cup I_3 = \mathcal{I}$ holds.
\end{definition}
It is well-known that the \texttt{3-WAY-\-PARTITION}\xspace problem is
NP-complete \cite{Korf2009}.
Now, we show that the \texttt{GRID-\-PARTITION}\xspace problem is already
NP-hard for two dimensions and a simple, one-dimensional
component stencil.
To that end, we reduce \texttt{3-WAY-\-PARTITION}\xspace to
\texttt{GRID-\-PARTITION}\xspace which leads to the following theorem.

\begin{theorem}
The \texttt{GRID-\-PARTITION}\xspace problem is NP-hard, when restricted to two
dimensions and a one-dimensional component stencil
	$\mathcal{S}=\{-\mathds{1}_1, \mathds{1}_1\}$.
\end{theorem}

\begin{figure}[t]
	\center
	\includegraphics[trim=0 0 0 0, clip, height=2.5cm]{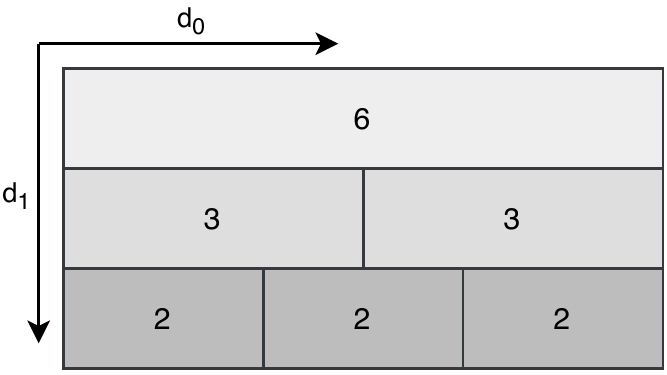}
	\caption{Example of \texttt{3-WAY-\-PARTITION}\xspace to \texttt{GRID-\-PARTITION}\xspace transformation. Given 
	$\mathcal{I'}=\{6,3,3,2,2,2\}$, create a grid with $\mathcal{D}=[6,2]$ for the Cartesian graph
	and find a mapping \textrm{s.t.}\xspace $J_\text{sum} \leq 2|\mathcal{I'}|-6$.}
	\label{fig:hardness}
\end{figure}

\begin{proof}
	Given an arbitrary instance $\mathcal{I'}$ of \texttt{3-WAY-\-PARTITION}\xspace,
        we construct an instance of
        \texttt{GRID-\-PARTITION}\xspace with a Cartesian graph, composed of $\mathcal{S}$ and~$\mathcal{D}$,
	and $\mathcal{N}, Q$
        as follows: $\mathcal{S} = \{-\mathds{1}_1, \mathds{1}_1\}$, $\mathcal{D}=[3,
          \frac{\sum_{x \in \mathcal{I'}} x}{3}]$, $\mathcal{N} =
	\{x \mid x \in \mathcal{I'}\}$, $Q = 2 |\mathcal{I'}| - 6$.

An optimal mapping of a two-dimensional
\texttt{GRID-\-PARTITION}\xspace problem with the component stencil always
traverses the vertices in the grid along the communicating dimension
given by the stencil, assigning them to a partition until it is
full. Thus, for $|\mathcal{N}| \geq 3$, each partition has at most two
outgoing communication edges (the first vertex assigned to the
partition and the last), except partitions at the border of the grid
which have one outgoing communication edge, \textrm{i.e.}\xspace, we can say \textrm{w.l.o.g}\xspace
that $Q= 2|\mathcal{N}| - 6$.

A \emph{yes} instance of \texttt{3-WAY-\-PARTITION}\xspace now corresponds to a
\emph{yes} instance of \texttt{GRID-\-PARTITION}\xspace, since we can assign every
vertex in the first, second and third column to the partitions that correlate
to the values in of $I_1$, $I_2$, $I_3$, respectively.
\end{proof}

An example for an instance with $\mathcal{I'}=\{6,3,3,2,2,2\}$ is shown in
Figure~\ref{fig:hardness}. By encoding solutions with the first and last vertex of each
partition, one can easily show that two-dimensional \texttt{GRID-\-PARTITION}\xspace
with the component stencil is NP-complete (that is, is also in
NP). With some adjustments, one can also show that the problem is
NP-complete if we allow periodicity along the dimension of
communication. The more interesting case is for which fixed stencils
the problem remains NP-complete.

\section{Rank Reordering Algorithms for $k$-Neighborhoods}
\label{sec:neighborhoodaware}

In this section, we propose three algorithms for a $k$-neighborhood
aware process reordering for Cartesian grids. The goal is to find
reordering schemes that are a) fully distributed, that is, each process
can compute its new rank independently of the other processes based on
the input alone (grid and stencil), and b) efficient, that is, not
polynomially dependent on $p$, preferably dependent only on the size
of the (sparse, compared to $p$) stencil and the number of dimensions
(we will tolerate polylog $p$).

\subsection{Hyperplane Algorithm}
The \texttt{Hyperplane}\xspace algorithm is a variation of recursive bisection. 
The main idea consists in recursively finding a split of a suitable dimension
$d_i$ of the Cartesian grid $g$ into $d_i', d_i''$ \textrm{s.t.}\xspace $d_i= d_i' + d_i''$.
This induces two grids $g'$ and $g''$, where the $i$th dimension size of $g'$
is $d_i'$ and of $g''$ is $d_i''$. The split is chosen \textrm{s.t.}\xspace the sizes of the two
new grids is a multiple of $n$, \textrm{i.e.}\xspace, $n\mid |g'|$ and $n\mid|g''|$.
Note that if we have heterogeneous nodes, one can use the mean, minimum
or maximum of the node sizes as an input for the algorithm.
This produces $N$ grids each of
size $n$ which can be mapped to the nodes. The cuts should be chosen \textrm{s.t.}\xspace the
minimal possible amount of communication between the grids is induced, since
those correspond to
inter-node communication. For that purpose, we calculate how parallel each
vector~$\vec{R} \in \mathcal{S}$ in the stencil is to a grid dimension $j$
using the cosine.
\begin{equation}\label{angle}
	\cos(\alpha_{\vec{R}, \vec{e_j}}) = \frac{\vec{R}\vec{e_j}}{||\vec{R}||||\vec{e_j}||} \in [-1, 1].
\end{equation}
Here, $\vec{e_j}$ is the unit vector along dimension~$j$ with~$0 \leq j < d$,~$\alpha_{\vec{R_i}, \vec{e_j}}$ 
is the angle between the relative coordinate vector~$\vec{R_i}$ and dimension $j$'s unit vector $\vec{e_j}$.
In order to have a monotonic increasing function, we 
square each of the values in Equation~\eqref{angle} and sum them over all relative coordinate
vectors $\vec{R} \in \mathcal{S}$, giving us the following list
\begin{equation}\label{cosines}
	\left[\sum_{i=0}^{k-1} \cos^2(\alpha_{\vec{R_i}, \vec{e_0}}),\ 
	\dots,\
	\sum_{i=0}^{k-1} \cos^2(\alpha_{\vec{R_i}, \vec{e_{d-1}}})\right].
\end{equation}
The dimension with the minimal value in Equation~\eqref{cosines} is the most orthogonal
to all $\vec{R} \in \mathcal{S}$, thus, we try to partition the grid alongside of it. 
Ties are broken by size, \textrm{i.e.}\xspace, we want to partition along the bigger dimension.
By sorting the dimensions according to their value in Equation~\eqref{cosines},
we define a preferred dimension order along which we try~to~perform~the~cuts.

The pseudo-code can be found in Algorithm~\ref{hyperplanealgo}. The input consists
of the dimension sizes~$\mathcal{D}$, the~$k$-neighborhood $\mathcal{S}$,
the number of processes per node~$n$, the rank~$r$ of the calling process and
it outputs the new position~$\vec{r}_\text{new}$ of the calling rank on the
Cartesian grid.

In each recursive step, we check if the grid is smaller than~$2 n$. 
We do this, for it is not needed to find an explicit cut for this grid size,
rather, we can directly calculate the new coordinates with the preferred
dimension order. This avoids bad splitting along very skewed grids, \textrm{e.g.}\xspace, a
nearest neighbor stencil on a two-dimensional grid with dimensions~$[2, n]$
where~$n$ is large and odd. Instead of being forced to cut along the first
dimension of size~$2$ to obtain two~$[1, n]$ partitions, we obtain two
partitions, each with~$3$~outgoing~communication~edges.

Otherwise, the algorithm finds the best possible split of the current
grid into two new grids. For that purpose, it traverses the
current dimensions~$\mathcal{D}$ sorted in increasing order of the values
defined in Equation~\eqref{cosines} (sorting in each recursive step is necessary,
because of changing dimension
sizes) and tries to position the splitting hyperplane in the current dimension~$d_i$.  The
hyperplane is initially placed at the center~$\frac{d_i}{2}$ of the candidate
dimension~$d_i$. If the initial split is not
suitable, the position of the hyperplane is incremented/decremented, respectively, \textrm{s.t.}\xspace
the position is as close as possible to the original grid's border in
an effort to reduce~$J_\text{max}$. If it cannot find a suitable split along the 
candidate dimension, it will proceed to the next, until it finds a split.
With the following proof, we show that it is always possible to find such a split.
\begin{theorem}\label{hyperplaneproof}
	Let $C, n \in \mathbb{N}$ and $C \geq 2$, let the dimension sizes
	be given by $\mathcal{D}=[d_0, \dots, d_{d-1}]$ and $\forall d
	\in \mathcal{D} : d \in \mathbb{N}$, \textrm{s.t.}\xspace $C  n =
	\prod_{i=0}^{d-1} d_i$. Then, it is always possible to split a dimension
	$d'$ \textrm{s.t.}\xspace the two induced grids are of size which is a multiple of
	$n$.
\end{theorem}

\begin{proof}
	Let $F(x) =
	\{ f_{1}, \dots, f_{l} \} $ be the multi-set of all prime factors of $x$.
	Then,
	\begin{equation}\label{prime_factor_product}
		F(\prod_{i=0}^{d-1} d_i) = \prod_{i=0}^{d-1} F(d_i) = F(C) F(p) = F(C p).
	\end{equation}
	Since $C \geq 2$ it must hold that, $\exists d' \in \mathcal{D}$, \textrm{s.t.}\xspace
	$F(d') \cap F(C) \neq \emptyset$
\end{proof}
In Line~$5$ of Algorithm~\ref{hyperplanealgo}, the subroutine \texttt{find\_split}
returns the index $i$ of the dimension to be split, and split sizes $d'$ and $d''$.
When a suitable split is found, two new grids $g'$, $g''$ are created where the $i$th dimension
size is replaced with $d'$ and $d''$ for $g'$ and $g''$, respectively. If the calling
rank is located on the left-hand side of the split, it will call \texttt{Hyperplane}\xspace with
$g'$ ($LHS$) as input and else, it calls \texttt{Hyperplane}\xspace with $g''$ ($RHS$) as new~input.

The number
of recursions executed by the \texttt{Hyperplane}\xspace algorithm is logarithmic
in the number of compute nodes~$N$, although the two grids per
recursion step can be very imbalanced in terms of size.
\begin{theorem}
  \label{proof:hyperplane_runtime}
	Let $g$ be a grid with dimension sizes $\mathcal{D}=[d_0, \dots,\
	d_{d-1}]$ and $\forall d \in \mathcal{D} : d \in \mathbb{N}$,
	and $\prod_{i=0}^{d-1} d_i = C p$ for some $C \in \mathbb{N}$.
	Then, the \texttt{Hyperplane}\xspace algorithm will always partition $g$ into two
	grids $g'$ and $g''$ \textrm{s.t.}\xspace $\frac{1}{2} \leq \frac{|g'|}{|g''|} \leq
	1$.
\end{theorem}

\begin{proof}
	Let $F(x)$ be as defined in the proof of Theorem~\ref{hyperplaneproof}.
	Let $d'$ be the candidate dimension of the algorithm with $F(d') \cap
	F(C) \neq \emptyset$. Let $F(d') \cap F(C) = \{d'_1, \dots,
	d'_m\}$ be ordered in ascending order.
	Then, the algorithm will surely find a suitable split at $d'_1$
	with $\left\lfloor \frac{d'_{1}}{2} \right\rfloor$ and $\left\lceil \frac{d'_{1}}{2} \right\rceil$. If $d'_{1}=2$, the resulting split will yield two partitions of exactly the same size.
	
	If $d'_{1} \geq 3$, then a split yields two partitions with a bigger or
	equal difference than a split of any other prime factor $d'_k \in F(d') \cap F(C)$.
	With $\left\lfloor
	\frac{d'_{1}}
	{2} \right\rfloor = \frac{d'_{1}-1}{2}$ and $\left\lceil \frac{d'_{1}}{2} \right\rceil = \frac{d'_{1}+1}{2}$,
	\begin{equation}
		\begin{aligned}
			\frac{|g'|}{|g''|} =
			\frac{\frac{d'_{1}-1}{2} \prod_{j=2}^m d'_j \prod_{d_i \neq d'} d_i}
			{\frac{d'_{1}+1}{2} \prod_{j=2}^m d'_j \prod_{d_i \neq d'} d_i} &\leq
			\frac{\frac{d'_{k}-1}{2} \prod_{j \neq k }^m d'_j \prod_{d_i \neq d'} d_i}
			{\frac{d'_{k}+1}{2} \prod_{j\neq k}^m d'_j \prod_{d_i \neq d'} d_i} \\
			\implies 
			\frac{d'_{1}-1}{d'_{1}+1} &\leq \frac{d'_{k}-1}{d'_{k}+1}.\\
		\end{aligned}
	\end{equation}	
	This always holds, since this is a strictly monotonic increasing function that converges to $1$ for growing $d'_k$ values.
		To see this suppose $x,y \geq 0$
		\begin{equation}
			\begin{aligned}
				\frac{x-1}{x+1} &\leq \frac{y-1}{y+1}\\
				(x-1)(y+1) &\leq (y-1)(x+1)\\
				xy +x -y -1 &\leq xy + y -x -1\\
				x &\leq y.
			\end{aligned}
		\end{equation}
	Note that if the Algorithm~\ref{hyperplanealgo} first positions the hyperplane at $\left\lfloor \frac{d'}{2} \right\rfloor$ it will eventually find a
	suitable split, latest at $\left\lfloor \frac{d'_{1}}{2} \right\rfloor \prod_{j=2}^{m} d'_j$.
	\begin{equation}
		\begin{aligned}
			\left\lfloor \frac{d'_1}{2} \right\rfloor \prod_{j=1}^m d'_j &\leq \left\lfloor \frac{d'}{2} \right\rfloor\\
			\frac{d'_1 - 1}{2} \prod_{j=1}^m d'_j &\leq \frac{d' - 1}{2}\\
			d' - \prod_{j=2}^m d'_j &\leq d' -1\\
			1 &\leq  \prod_{j=2}^m d'_j
		\end{aligned}
	\end{equation}
	
	We can bound the ratio of the two grid sizes $|g'|$ and $|g''|$ from below, since $d'_1 \geq 3$
	\begin{equation}
		\begin{aligned}
			\frac{|g'|}{|g''|} =  
			\frac{\frac{d'_{1}-1}{2} \prod_{j=2}^m d'_j \prod_{d_i \neq d'} d_i}
			{\frac{d'_{1}+1}{2} \prod_{j=2}^m d'_j \prod_{d_i \neq d'} d_i} &\geq \frac{d'_1 - 1}{d'_1 +1}
			\geq \frac{3-1}{3+1} &\geq \frac{1}{2}.
		\end{aligned}
	\end{equation}
\end{proof}

It follows that the running time 
of \texttt{Hyperplane}\xspace is bounded by~$\mathcal{O}\left( \log N  \sum_{i=0}^{d-1} d_i  \right)$.

\begin{algorithm}[t!]
	\SetKw{break}{break}
	\SetKw{return}{return}
	\SetAlgoLined
	\SetKwFunction{hyperplane}{hyperplane}
	\SetKwFunction{newcoordinate}{new\_coordinate}
	\SetKwFunction{sort}{sorted}
	\SetKwFunction{adapt}{push\_to\_boundary}
	\SetKwFunction{findsplit}{find\_split}
	\SetKwFunction{inducedgrids}{induced\_grids}
	\SetKwFunction{calcremain}{n\_remain\_parts}
	\KwIn{Grid with dimension sizes~$\mathcal{D}=[d_0, \dots, d_{d-1}]$,
	set of relative vectors~$\mathcal{S}$ describing the $k$-neighborhood,
	number of processes per node~$n$ and rank $r$ of calling process.}
	\KwResult{New coordinate $\vec{r}_\text{new}$ of calling rank.}
	\Indm\renewcommand{\nl}{\let\nl\oldnl}\hyperplane{$\mathcal{D}, \mathcal{S}, n, r, \vec{r}_\text{new}$}\\
	\Indp
	\eIf{$\prod_{d \in \mathcal{D}} d \leq 2 n$}{
		$\vec{r}_\text{new} \gets \newcoordinate{$\mathcal{D}, \mathcal{S}, n, r$}$\\
		\return
	}{
		$i, d', d'' \gets \findsplit{$\mathcal{D}, \mathcal{S}, n$}$\\
		$LHS, RHS~\gets~\inducedgrids{$\mathcal{D}, i, d',d''$}$\\
		\eIf{$r \in LHS$}{
			$\hyperplane{$LHS, \mathcal{S}, n, r_\text{new}$}$\\
		}{
			$\hyperplane{$RHS, \mathcal{S}, n, r_\text{new}$}$\\
		}
	}
	\caption{\texttt{Hyperplane}\xspace Algorithm}
	\label{hyperplanealgo}
\end{algorithm}
\begin{figure}[t]
	\begin{subfigure}{0.15\textwidth}
	\includegraphics[trim = 100 70 100 45, clip, scale=0.45]{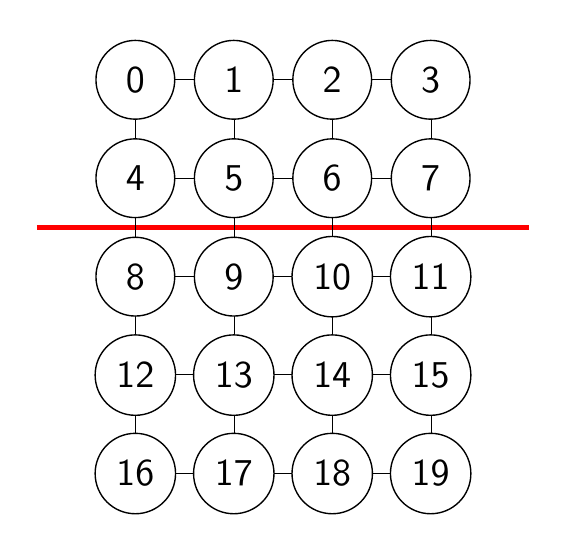}
	\caption{}
	\end{subfigure}
	\begin{subfigure}{0.15\textwidth}
	\includegraphics[trim = 100 70 100 45, clip, scale=0.45]{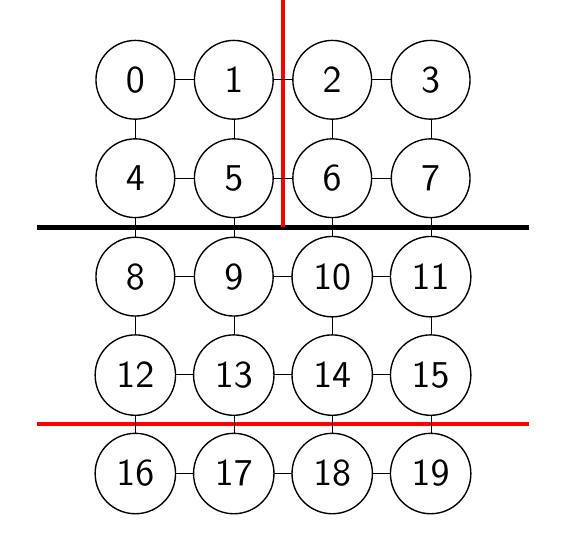}
	\caption{}
	\end{subfigure}
	\begin{subfigure}{0.15\textwidth}
	\includegraphics[trim = 100 70 100 45, clip, scale=0.45]{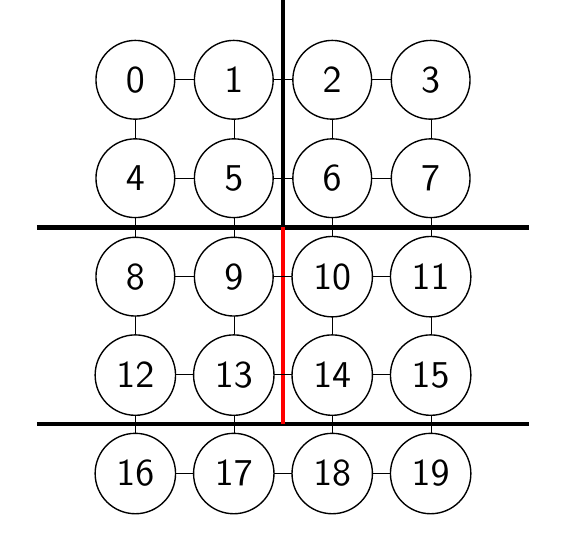}
	\caption{}
	\end{subfigure}
	\caption{\texttt{Hyperplane}\xspace algorithm example on $5 \times 4$ grid with a 
	nearest neighbor stencil and $N=5$, $n=4$:
	(a) The first split is along the largest dimension;
	(b) Two new splits found on both grids; (c) Final split of the grid.}
	\label{hyperplanefig}
\end{figure}

\subsection{$k$-$d$ Tree Algorithm}
\label{sec:kdtree}

Similar to the \texttt{Hyperplane}\xspace algorithm, but inspired by the $k$-$d$ tree data
structure, the $k$-$d$~\texttt{tree}\xspace algorithm works also by recursively splitting
the grid along the dimensions. The main difference to the \texttt{Hyperplane}\xspace algorithm
is that the recursion continues until there is only one vertex left in the Cartesian
grid. The advantage of this approach is that the algorithm is no
longer bound to the number of processes per node $n$. In fact, it is oblivious
to it and only tries to find \emph{dense} mappings, \textrm{i.e.}\xspace, localize communicating
vertices.

Instead of partitioning the grid dimensions in a round-robin manner, like in
the $k$-$d$ tree data-structure, we choose the largest dimension, weighted by
the inverse amount of communication that is performed across it, so that we avoid
splitting a dimension along which is communicated intensively. 
To be more precise, we define the amount of communication across a dimension
$j$ to be 
$
	f_j := |\{\vec{R} \mid \forall \vec{R} \in \mathcal{S} : R_{j} \neq
0\}|,
$
where $R_{j}$ is the $j$th component of $\vec{R}$, then we want to find the
dimension index $i := \operatorname*{argmax}_{0 \leq i < |\mathcal{D}|} \frac{d_i}{f_i}$
over the current grid dimensions sizes $\mathcal{D}$ and split equally along $d_i$.
All ranks $r$ with $r \leq \frac{1}{2}\prod_{j=0}^{d-1} d_j$ are assigned to
the left-hand side, and the others to the right-hand side of the split.

The pseudo-code
is given in Algorithm~\ref{algo:kdtree}. If there is only one vertex left in
the calling grid, we enter
the base-case of the recursive function, in which the vector $\vec{r}_\text{new}$ already
holds the coordinates of the calling process in the grid. Otherwise, 
we split the best suited dimension equally and recurse accordingly. 
\begin{algorithm}[t!]
	\SetKw{break}{break}
	\SetKw{return}{return}
	\SetAlgoLined
	\SetKwFunction{findsplit}{find\_split\_index}
	\SetKwFunction{kdtreecoord}{kd\_tree}
	\SetKwFunction{newcoordinate}{new\_coordinate}
	\KwIn{Dimension sizes $\mathcal{D}$, set of relative vectors
	describing the~$k$-neighborhood $\mathcal{S}$ and 
	rank~$r$ of calling process.}
	\KwResult{New coordinate $\vec{r}_\text{new}$ of calling rank}
	\Indm\renewcommand{\nl}{\let\nl\oldnl}\kdtreecoord{$\mathcal{D}, \mathcal{S}, r, \vec{r}_\text{new}$}\\
	\Indp
	\eIf{$\prod_{d \in \mathcal{D}} d = 1$}{ \return\\ }{
		$k \gets \findsplit{$\mathcal{D}, \mathcal{S}$}$\\
		\eIf{$r \leq \frac{\prod_{d \in \mathcal{D}}d}{2}$}{
			$d_k \gets \lfloor \frac{d}{2} \rfloor$\\
			$\kdtreecoord{$\mathcal{D}, \mathcal{S}, r, \vec{r}_\text{new}$}$\\
		}{
			$d_k \gets \lceil \frac{d}{2} \rceil$\\
			$\vec{r}_\text{new}[k] \gets \vec{r}_\text{new} + d_k$\\
			$\kdtreecoord{$\mathcal{D}, \mathcal{S}, r, \vec{r}_\text{new}$}$\\
		}
	}
	\caption{$k$-$d$~\texttt{tree}\xspace Algorithm}
	\label{algo:kdtree}
\end{algorithm}

For the run-time analysis, it is not difficult to see that depth of the 
recursion tree is $\log_2 p$, where $p$ is the total number of processes
in the Cartesian grid. At each step of the recursion, we have to find the
best dimension for the split. Using a priority queue, this can be achieved
in $\mathcal{O}\left(\log d\right)$ steps. Since the remaining computations
are constant, the runtime for calculating the new reordering is
$\mathcal{O}\left(\log p \log d \right)$.

\subsection{Stencil Strips Algorithm}
\label{sec:stencilstrips}
\begin{algorithm}[b!]
	\SetKw{break}{break}
	\SetAlgoLined
	\SetKwFunction{stencilstrips}{stencil\_strips}
	\SetKwFunction{distortionfactors}{get\_distortion\_factors}
	\SetKwFunction{striplengths}{get\_strip\_lengths}
	\SetKwFunction{stripcoord}{get\_strip\_id}
	\SetKwFunction{newcoordinate}{new\_coordinate}
	\SetKwFunction{setassigndirection}{set\_assignment\_direction}
	\KwIn{Dimension sizes $\mathcal{D}$, set of relative vectors 
	$\mathcal{S}$ describing the~$k$-neighborhood,
	number of processes per compute node $n$
	and rank $r$ of calling process.}
	\KwResult{New coordinate $\vec{r}_\text{new}$ of calling rank.}
	\Indm\renewcommand{\nl}{\let\nl\oldnl}\stencilstrips{$\mathcal{D}, \mathcal{S}, n, r, \vec{r}_\text{new}$}\\
	\Indp
	$\alpha \gets \distortionfactors{$\mathcal{S}, |\mathcal{D}|, n$}$\\
	$s \gets \striplengths{$|\mathcal{D}|, n, \alpha$}$\\
	$v \gets \prod_{i \gets 0}^{|\mathcal{D}|-1} s[i]$\\ 
	\For{$i \gets 0$ \KwTo $|\mathcal{D}|-1$}{
		$v \gets \frac{v}{s[i]}$\\
		$s_\text{coord}[i] \gets \lfloor \frac{r}{v} \rfloor$\\
		$r \gets r - (s_\text{coord}[i]-1)v$\\
	}
	$\vec{r}_\text{new} \gets \newcoordinate{$\mathcal{D}, s, s_\text{coord}, r$}$\\
	\caption{\texttt{Stencil Strips}\xspace Algorithm}
	\label{stencilstripscoord}
\end{algorithm}

During early experimentation, we noticed that a consecutive assignment of processes to 
compute nodes for grids where the dimension sizes were close to the $d$th root
of $n$ (\textrm{i.e.}\xspace, optimal side length) and the nearest neighbor stencil
resulted in lower values of $J_\text{sum}$ and
$J_\text{max}$ in comparison to recursive bisection.
This observation was the
motivation for this algorithm. The main idea is to partition the grid into
strips, where strip lengths are chosen such that they are close to the scaled
length of an optimal bounding rectangle of $\mathcal{S}$ (\textrm{e.g.}\xspace, $\sqrt[\leftroot{-3}\uproot{3}d]{n}$ for
the nearest neighbor stencil). Processes in
coherent strips are~assigned~to~partitions/nodes.

For that purpose, let $R_{i}$ be the $i$th component of $\vec{R}$, then we define
$n_{i, \text{max}} := \max \{R_{i} \mid \forall \vec{R} \in \mathcal{S}\}$ and
$n_{i, \text{min}} := \min \{R_{i}\mid \forall \vec{R} \in \mathcal{S}\}$
to be the maximal and minimal value along dimension $i$ in the $k$-neighborhood
$\mathcal{S}$, respectively. The \emph{extension} $e_i$ of $\mathcal{S}$ along dimension $i$ is
$e_i := n_{i, \text{max}} - n_{i, \text{min}}$.
We use the extensions to find a bounding $n$-dimensional rectangle of
$\mathcal{S}$, with dimension sizes $\mathcal{E} = [e_0, \dots,
e_{d-1}]$.  We define the volume $V_b$ of the bounding $n$-dimensional
rectangle to be \begin{equation*} 
	V_b := \prod_{i=0}^{d-1} \epsilon_i, \quad \epsilon_i=
	\begin{cases}
		1 & \text{if } e_i = 0\\
		e_i & \text{else}.
	\end{cases}
\end{equation*}
With $d_b:= |\{e_i \mid \forall e_i \in \mathcal{E} \land e_i \neq 0\}|$ being
the number of non-zero expansions, we define the \emph{distortion factor} $\alpha_i$ to 
a~$d_b$-dimensional cube along dimension $i$ as
\begin{equation*} 
	\alpha_i := \frac{e_i}{\sqrt[\leftroot{-3}\uproot{3}d_b]{V_b}}. 
\end{equation*} 
As input for the algorithm serves the dimension sizes~$\mathcal{D}$ of the
grid, the~$k$-neighborhood $\mathcal{S}$, the number of processes per node~$n$
and the rank $r$ of the calling process. The pseudo-code is presented in
Algorithm~\ref{stencilstripscoord}. The algorithm outputs the new position
vector $\vec{r}_\text{new}$ of the calling rank. The algorithm itself works as
follows: we calculate a list of distortion factors $\alpha$ for every dimension
(Line~$1$). With the distortion factors, we can calculate a list $s$, containing
the optimal strip length~$s_i$ along dimension~$i$ (Line~$2$), with
consideration to already computed strip lengths, the distortion
factors~$\alpha_i$ and the node size~$n$:
\begin{equation*}
	s_i := \sqrt[\leftroot{-3}\uproot{3}d-i]{\frac{\alpha_i n}{\prod_{j=0}^{i-1} s_j}}.
\end{equation*}
This is done for every dimension except the largest one, for we iteratively position
the strips along the largest dimension. To be more precise, we assume that
along the largest dimension the strip length is one. In every other dimension
$d_i$, we fit~$\lfloor \frac{d_i}{s_i} \rfloor$ strips (the last strips is of
size~$s_i + \frac{d_i}{s_i}$). 
Every process can locally compute the number of small and large
strips along the dimensions, how many processes fit into each strip
and, with the calling process' rank~$r$ the coordinates $s_\text{coord}$ of
the strip in which it is contained (Line~$3$-$5$),
from which it can calculate its new position in the strip and thus, in the grid
$\vec{r}_\text{new}$~(Line~$9$). In order to obtain cohesive partitions, we flip the
strip assignment direction accordingly as depicted in
Figure~\ref{fig:stencilstrips}.

The strip widths, strip coordinates and the ranks position in the strip can
each be calculated in~$\mathcal{O}(d)$. The most expensive part of the 
algorithm is calculating the distortion factors, for which we need the maximal
and minimal expansion along the dimensions. This amounts in $\mathcal{O}(k d)$
steps. Thus, the run-time of the algorithm is bounded by $\mathcal{O}(k d)$.
\begin{figure}[t]
	\begin{subfigure}{0.2\textwidth}
		\center
	\includegraphics[trim = 0 0 0 0, clip, height=2cm]{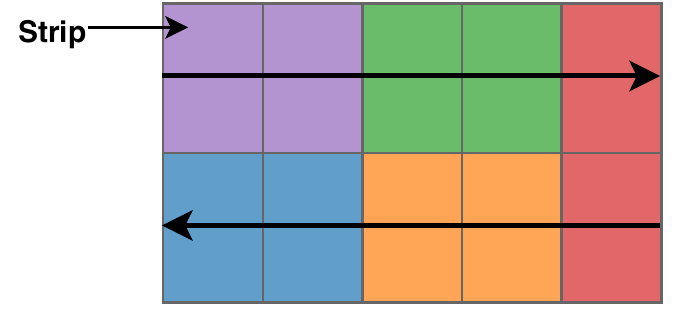} \\
	\caption{}
	\end{subfigure}    
	\hfill	
	\begin{subfigure}{0.2\textwidth}
		\center
	\includegraphics[trim = 0 0 0 0, clip, height=2cm]{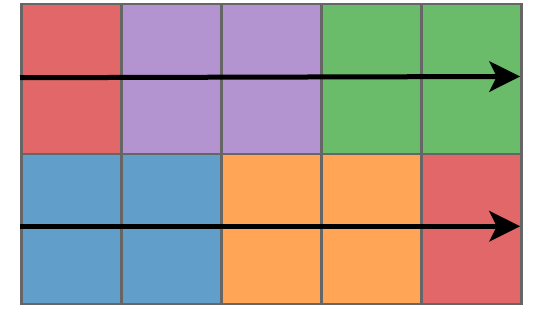} \\
	\caption{}
	\end{subfigure}
	\caption{Assignment direction of the \texttt{Stencil Strips}\xspace algorithm.
	(a): Alternating strips assignment direction accordingly guarantees 
	coherent partitions.
	(b): Imprudent assignment direction can lead to incoherent partitions.
	}
	\label{fig:stencilstrips}
\end{figure}

\section{Experimental evaluation}
\label{sec:experiments}
We have implemented the presented algorithms for Cartesian rank
reordering and the algorithm presented by Gropp~\cite{Gropp2019}, 
in accordance with the detailed pseudo-code of his paper.
We aim to show the advantage of approaches that do not rely on factorization of
the number of processes per node $n$. For that purpose, we compare the presented 
algorithms to a blocked assignment of ranks to nodes (henceforth referred as
to blocked), \texttt{Nodecart}\xspace, and \texttt{VieM}\xspace, which are described in
Section~\ref{sec:relatedwork}.  We start by describing the
systems used to run the experiments in Section~\ref{sec:systems}, implementation and
benchmarking details are found in Section~\ref{sec:methodology}. In
Section~\ref{sec:reduction}, we visualize the distribution of the reduction in
inter-node communication achieved by the different algorithms over the blocked
mapping for a wide set of instances. We proceed in
Section~\ref{sec:throughput}, where we present the throughput gained by our
algorithms for a neighbor all-to-all exchange and different message sizes on
different machines. Finally, we compare the algorithmic running time in
Section~\ref{sec:instantiation}.  
\begin{figure*}[t]
	\captionsetup[sub]{labelformat=empty, skip=0pt}
	\subcaptionbox{Nearest neighbor}[\textwidth]['u']{
		\begin{subfigure}{0.24\textwidth}
			\center
			\caption{}
			\includegraphics[trim=0 20 0 10, clip, height=3.3cm]{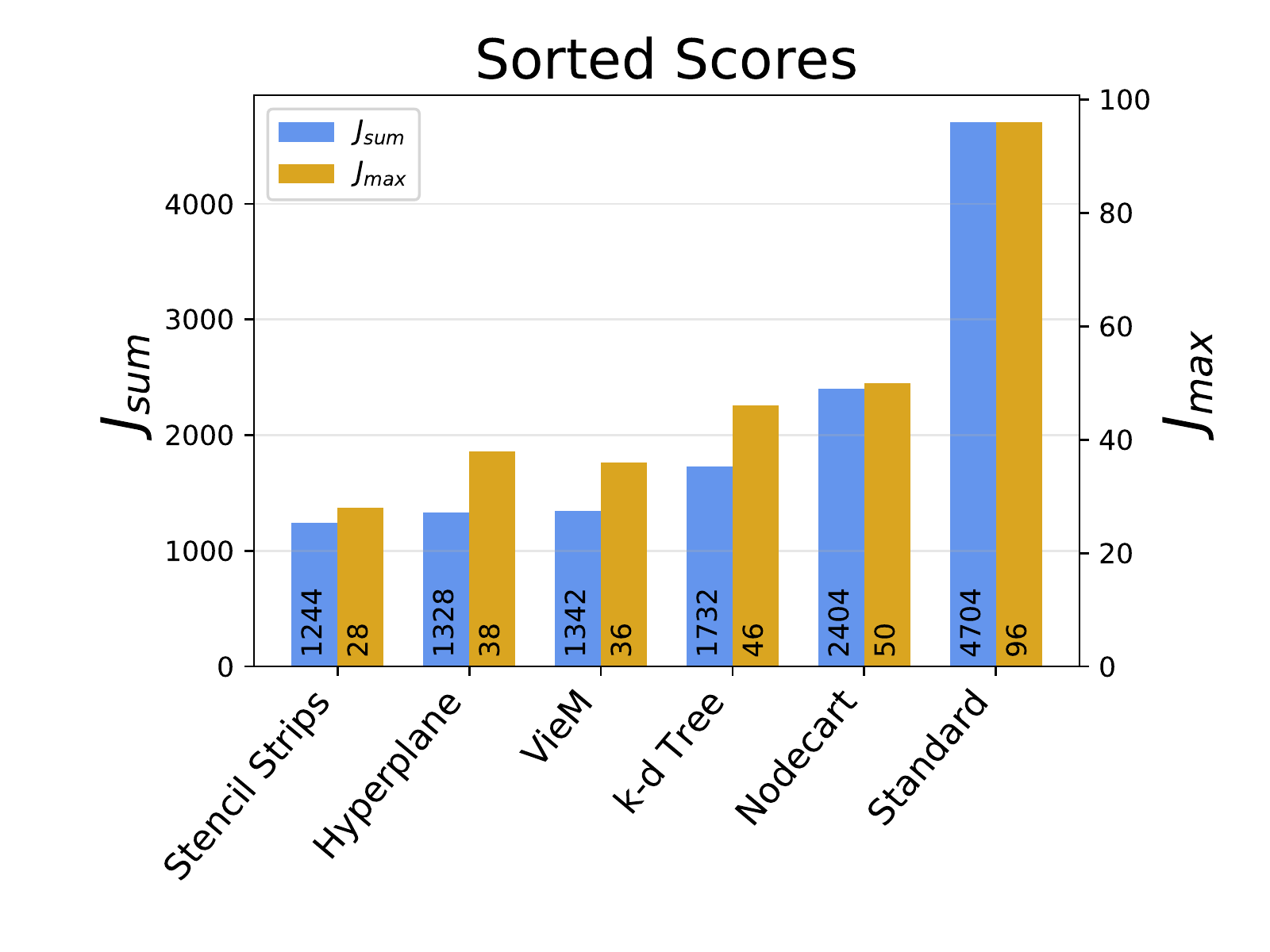}
		\end{subfigure}
		\hfill
		\begin{subfigure}{0.24\textwidth}
			\centering
			\caption{\emph{VSC}$4$\xspace}
			\includegraphics[trim = 0 20 0 30, clip, height=3.1cm]{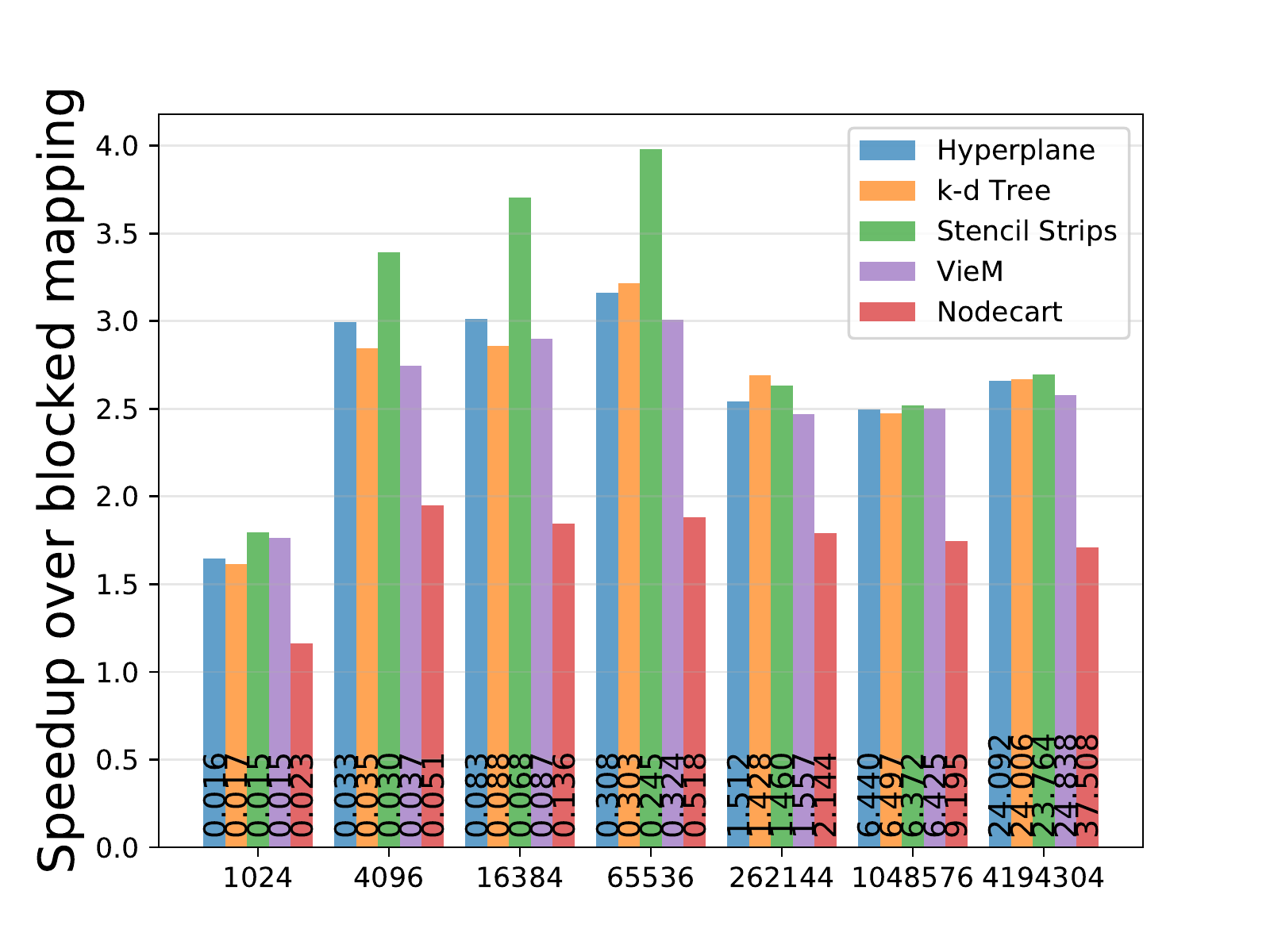}
			\end{subfigure}
		\begin{subfigure}{0.24\textwidth}
			\centering
			\caption{\emph{SuperMUC-NG}\xspace}
			\includegraphics[trim = 0 20 0 30, clip, height=3.1cm]{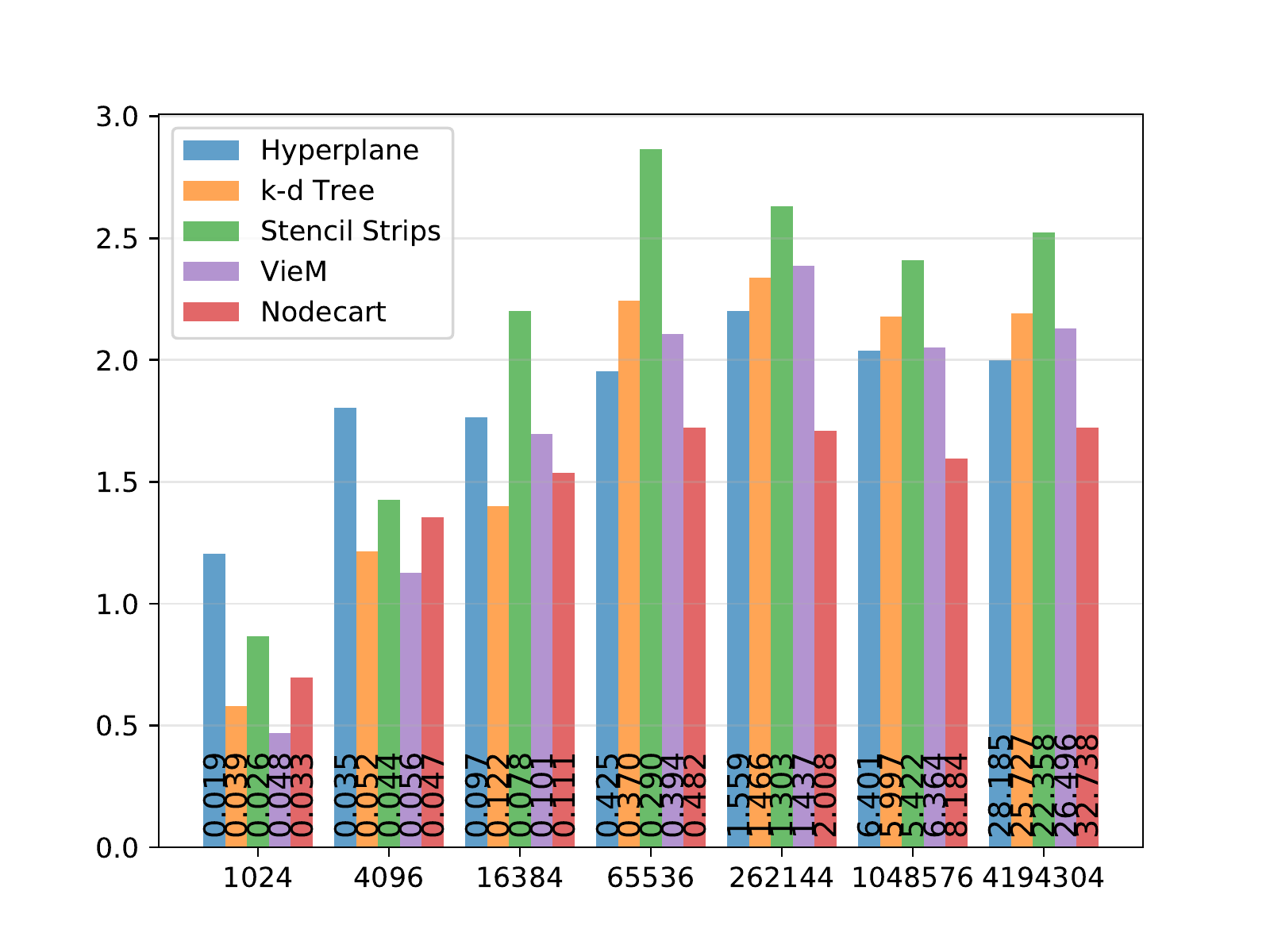}
		\end{subfigure}
		\begin{subfigure}{0.24\textwidth}
			\centering
			\caption{\emph{JUWELS}\xspace}
                         \includegraphics[trim = 0 20 0 30, clip, height=3.1cm]{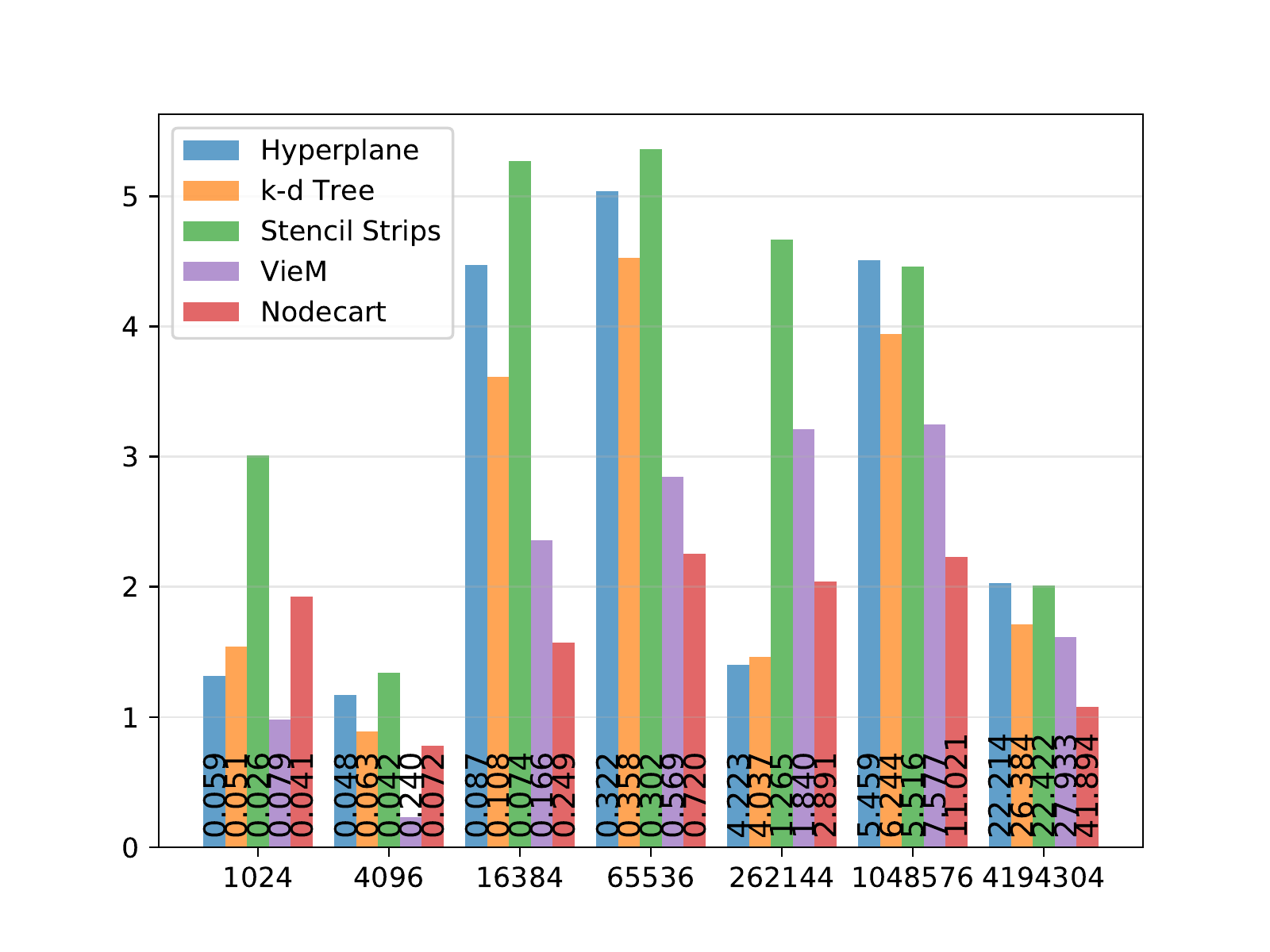}
		\end{subfigure}
	}
	\captionsetup[sub]{labelformat=empty, skip=0pt}
	\subcaptionbox{Nearest neighbor with hops}[\textwidth]['u']{
		\begin{subfigure}{0.24\textwidth}
			\center
			\caption{}
			\includegraphics[trim=0 20 0 10, clip, height=3.3cm]{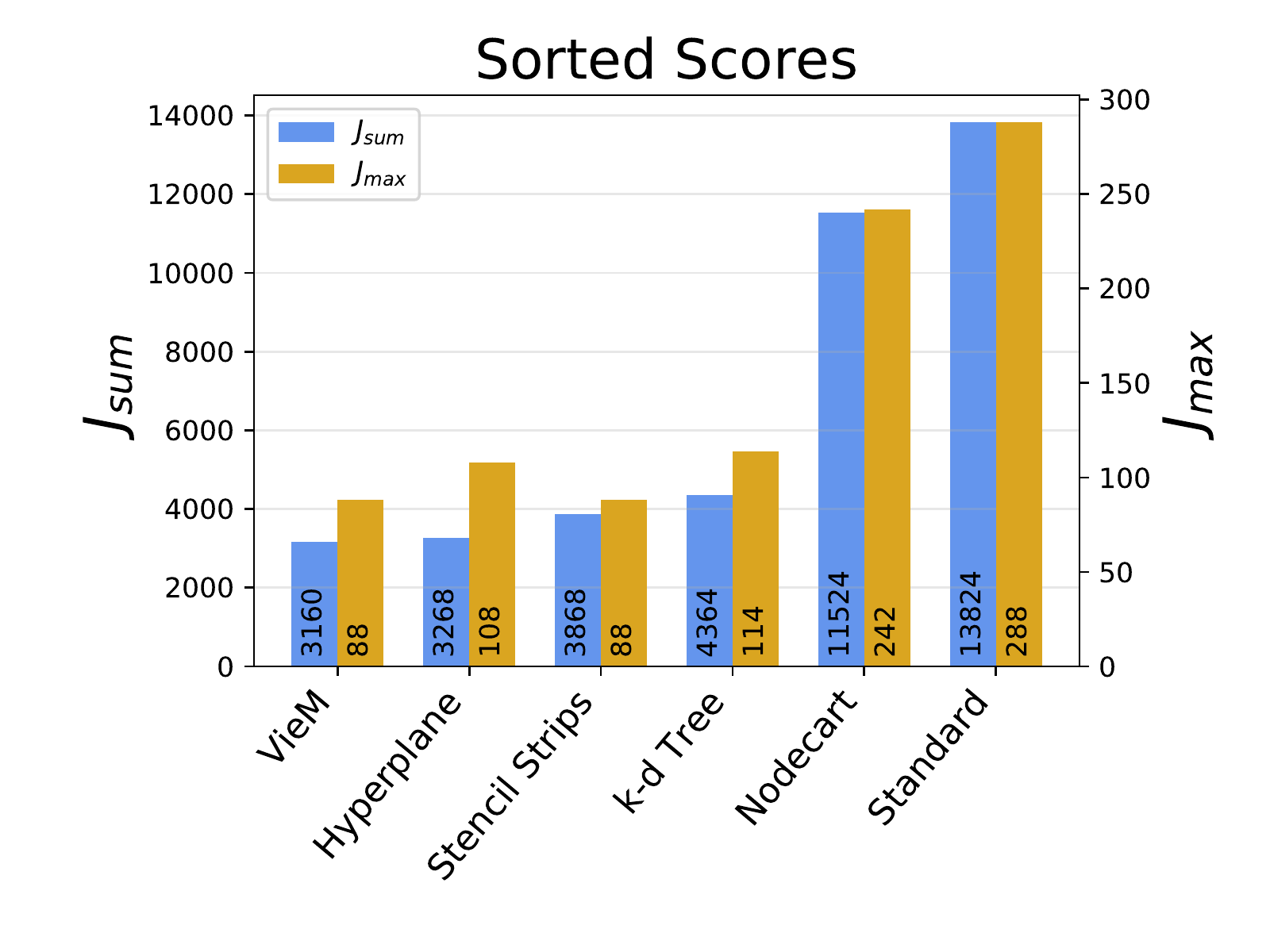}
		\end{subfigure}
		\hfill
		\begin{subfigure}{0.24\textwidth}
			\centering
			\includegraphics[trim = 0 20 0 30, clip, height=3.1cm]{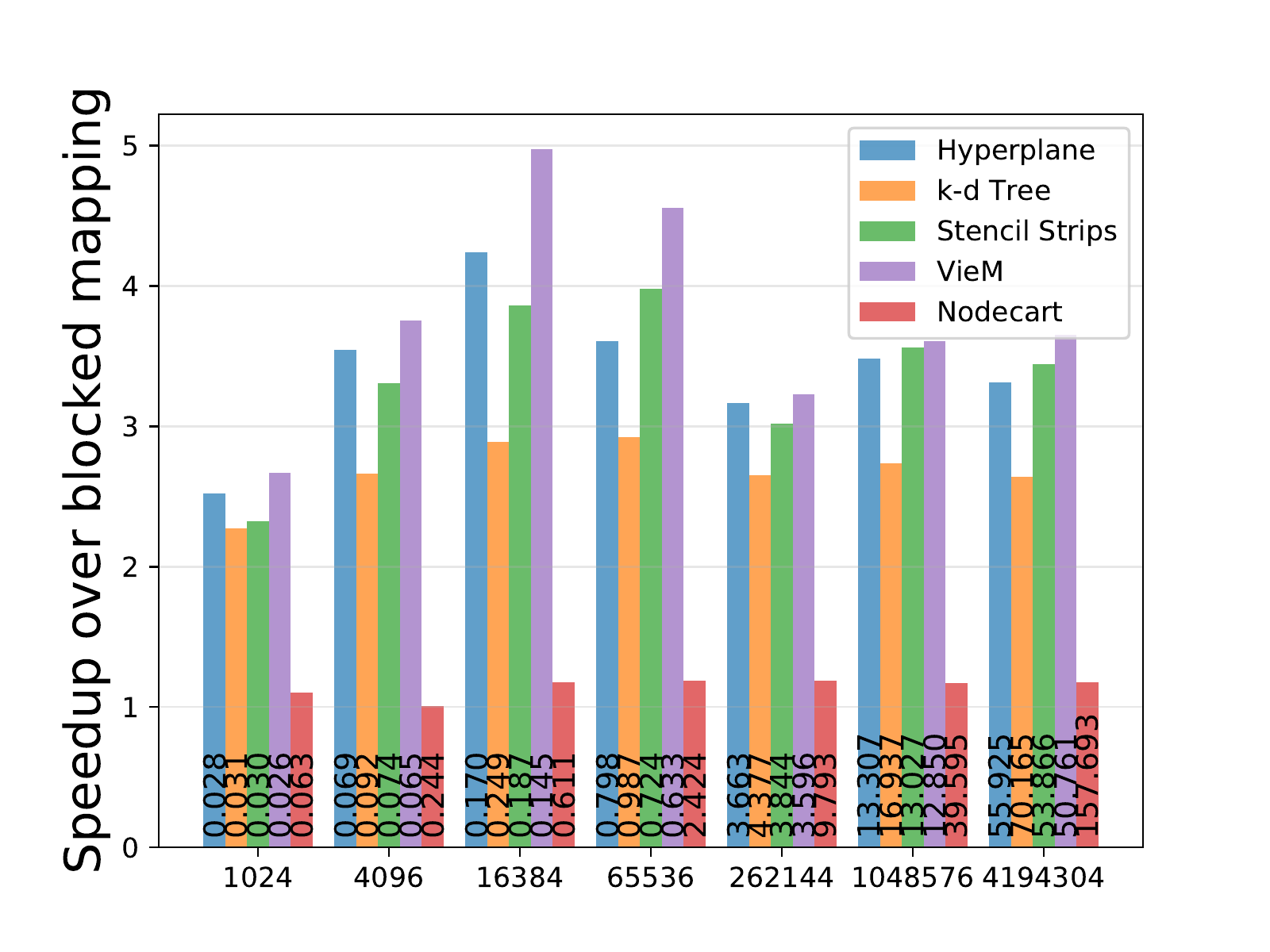}
			\end{subfigure}
		\begin{subfigure}{0.24\textwidth}
			\centering
			\includegraphics[trim = 0 20 0 30, clip, height=3.1cm]{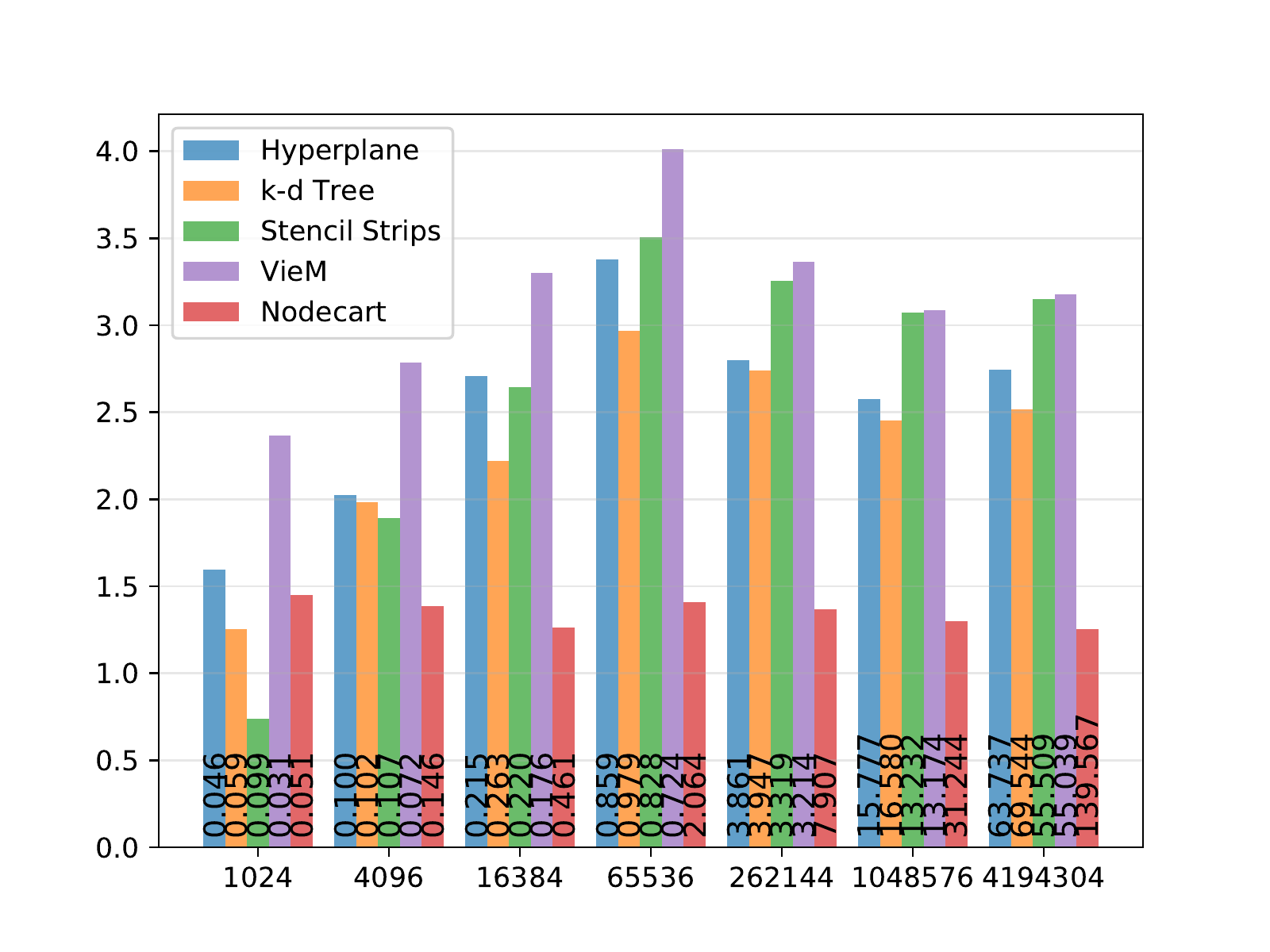}
		\end{subfigure}
		\begin{subfigure}{0.24\textwidth}
			\centering
			\includegraphics[trim = 0 20 0 30, clip, height=3.1cm]{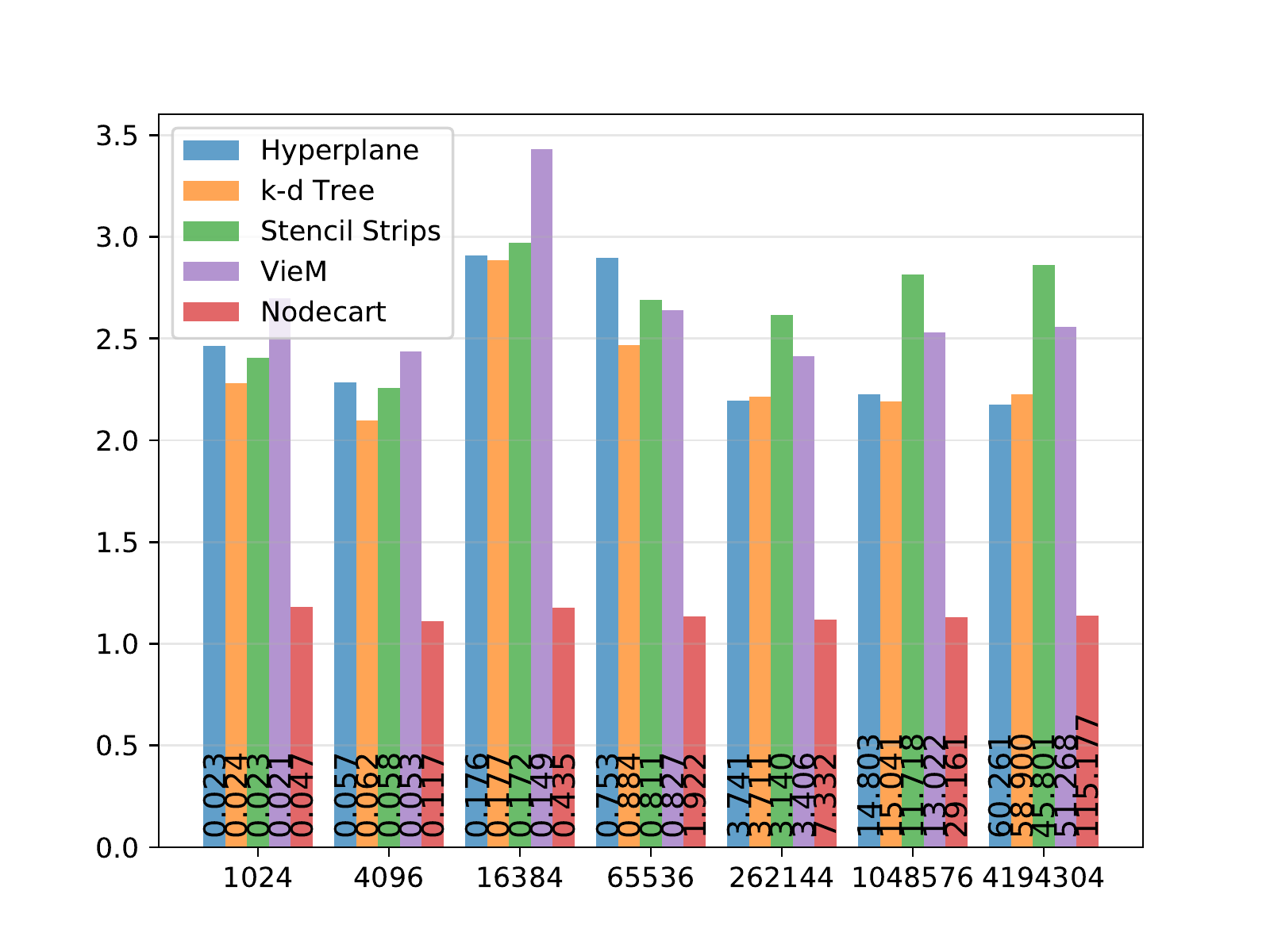}
		\end{subfigure}
	}
	\captionsetup[sub]{labelformat=empty, skip=0pt}
	\subcaptionbox{Component stencil}[\textwidth]['u']{
		\begin{subfigure}{0.24\textwidth}
			\center
			\caption{}
			\includegraphics[trim=0 20 0 10, clip, height=3.3cm]{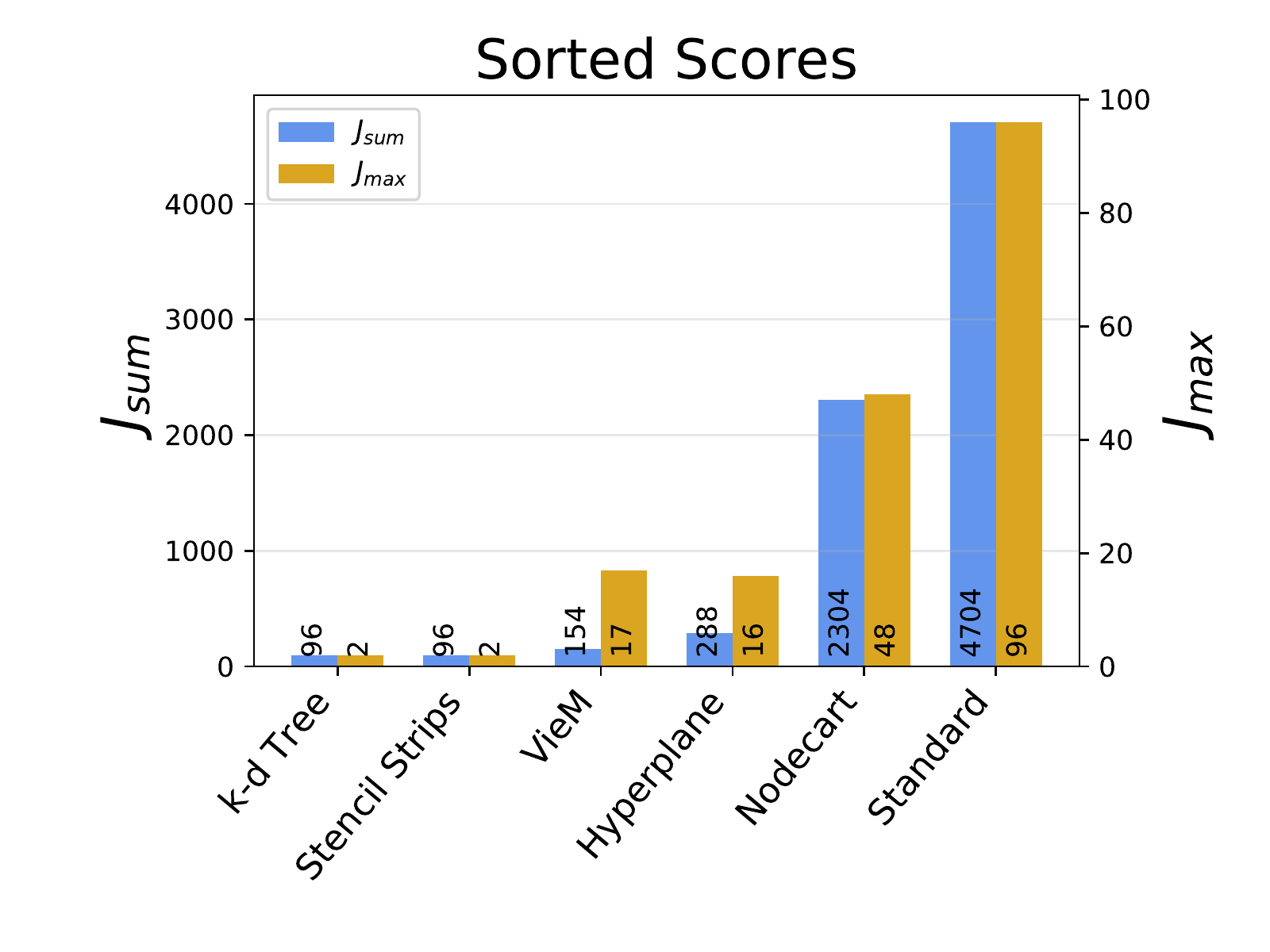}
			\label{fig:relnn50}
		\end{subfigure}
		\hfill
		\begin{subfigure}{0.24\textwidth}
			\centering
			\includegraphics[trim = 0 0 0 30, clip, height=3.2cm]{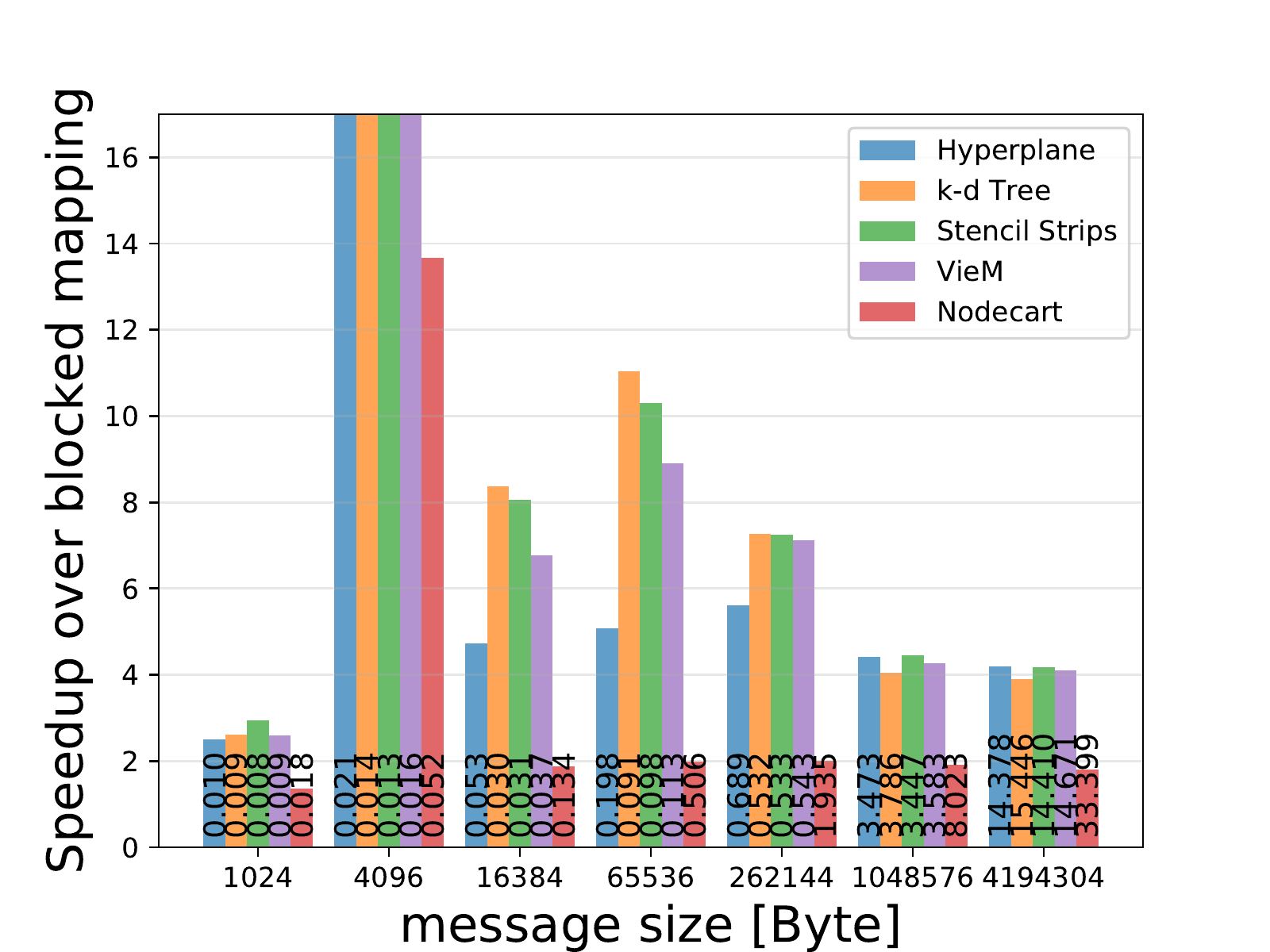}
			\end{subfigure}
		\begin{subfigure}{0.24\textwidth}
			\centering
			\includegraphics[trim = 0 0 0 30, clip, height=3.2cm]{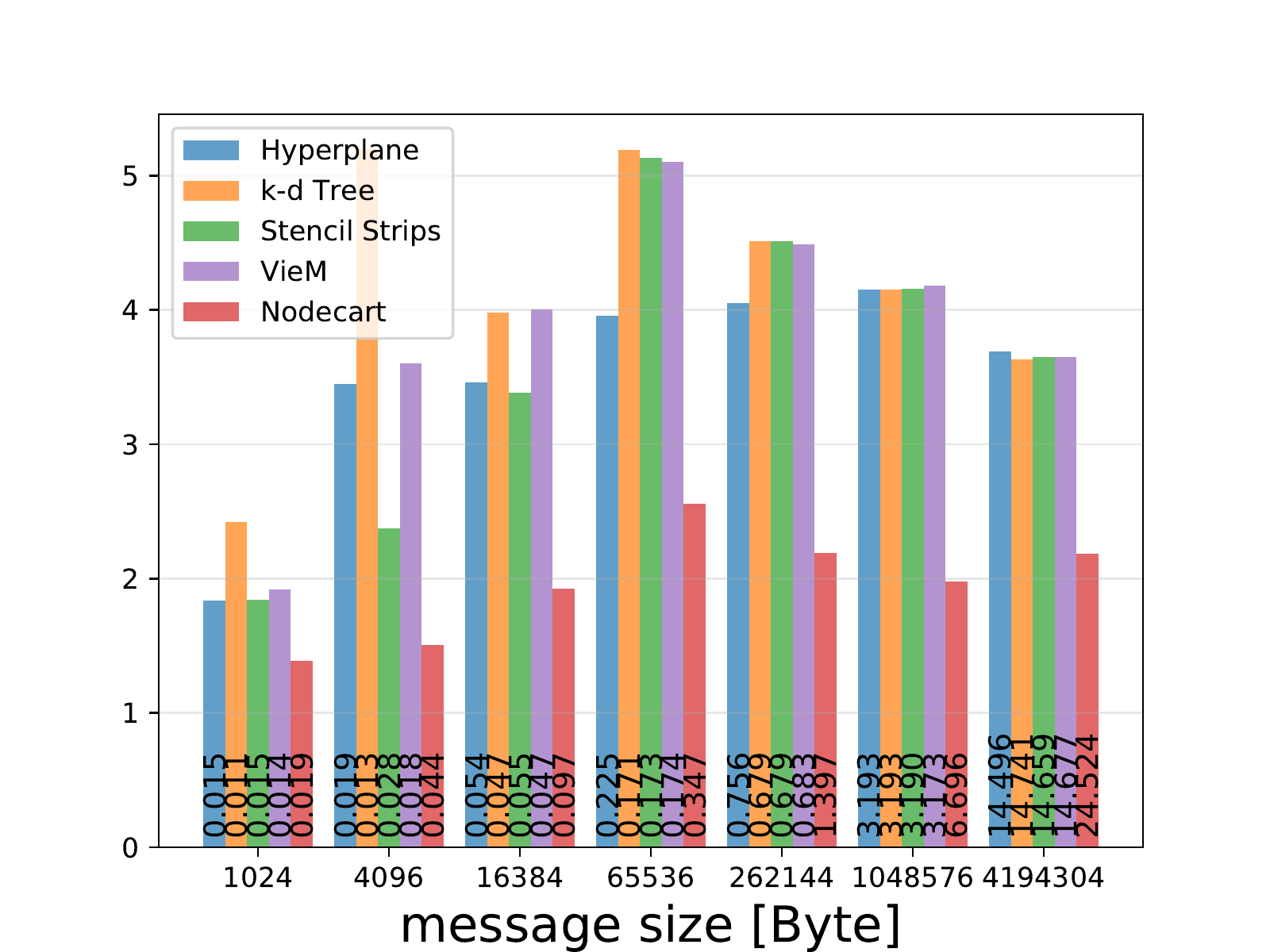}
		\end{subfigure}
		\begin{subfigure}{0.24\textwidth}
			\centering
			\includegraphics[trim = 0 0 0 30, clip, height=3.2cm]{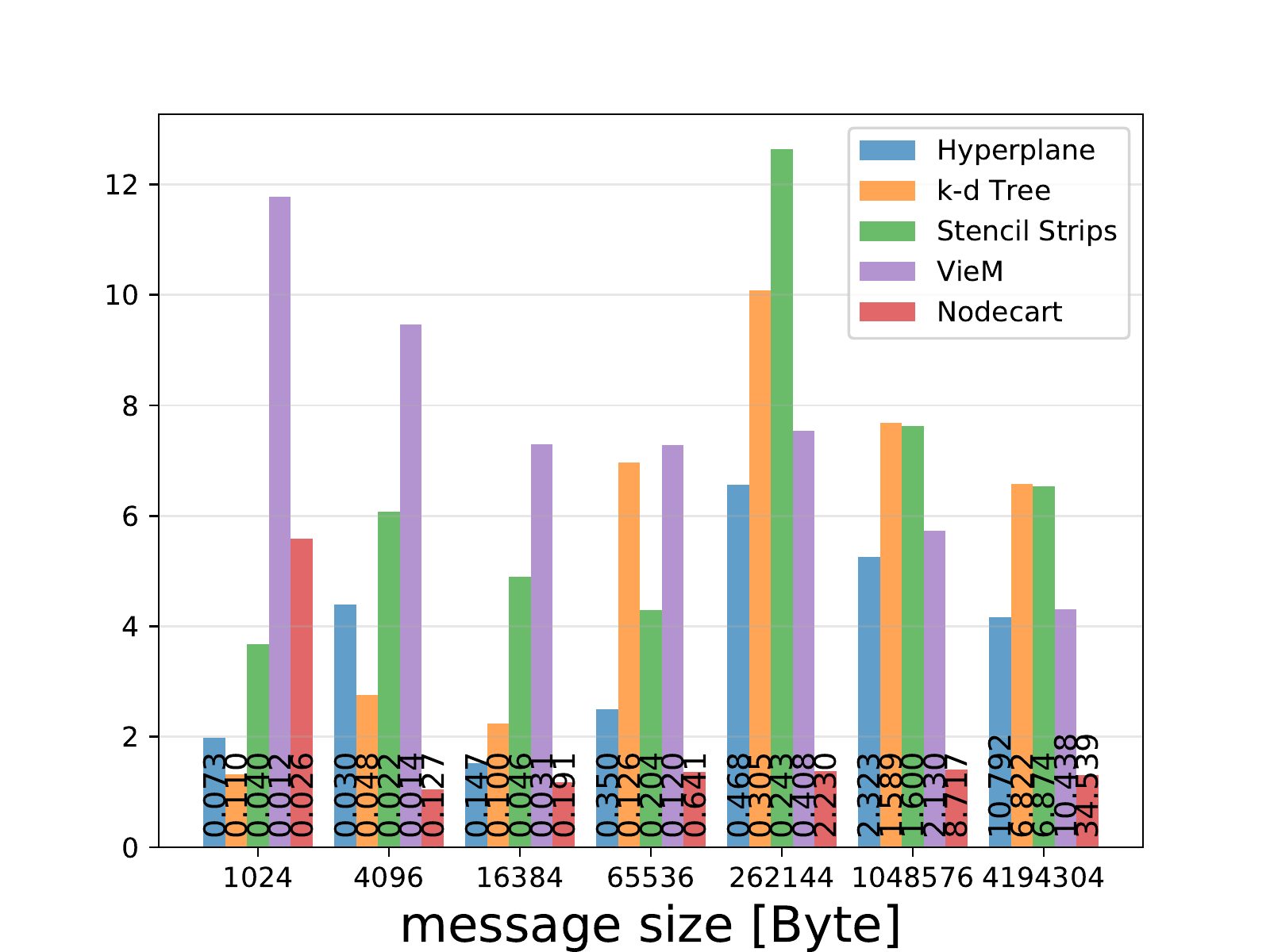}
		\end{subfigure}
	}
	\caption{Left column: scores of the algorithms (smaller is better).
	Right three columns: Speedup over blocked mapping (higher is better) with
	$N=50$ number of nodes and $p=48$ processes per node and grid sizes
	$50\times 48$ for the nearest neighbor stencil,
	the nearest neighbor stencil with hops and the
	component stencil on the \emph{VSC}$4$\xspace, \emph{SuperMUC-NG}\xspace, and \emph{JUWELS}\xspace. Absolute times 
	are written in the corresponding bar.
	}\label{fig:bandwidth50}
\end{figure*}
\begin{figure*}[ht]
	\captionsetup[sub]{labelformat=empty, skip=0pt}
	\subcaptionbox{Nearest neighbor}[\textwidth]['l']{
		\begin{subfigure}{0.24\textwidth}
			\center
			\caption{}
			\includegraphics[trim=0 20 0 10, clip, height=3.3cm]{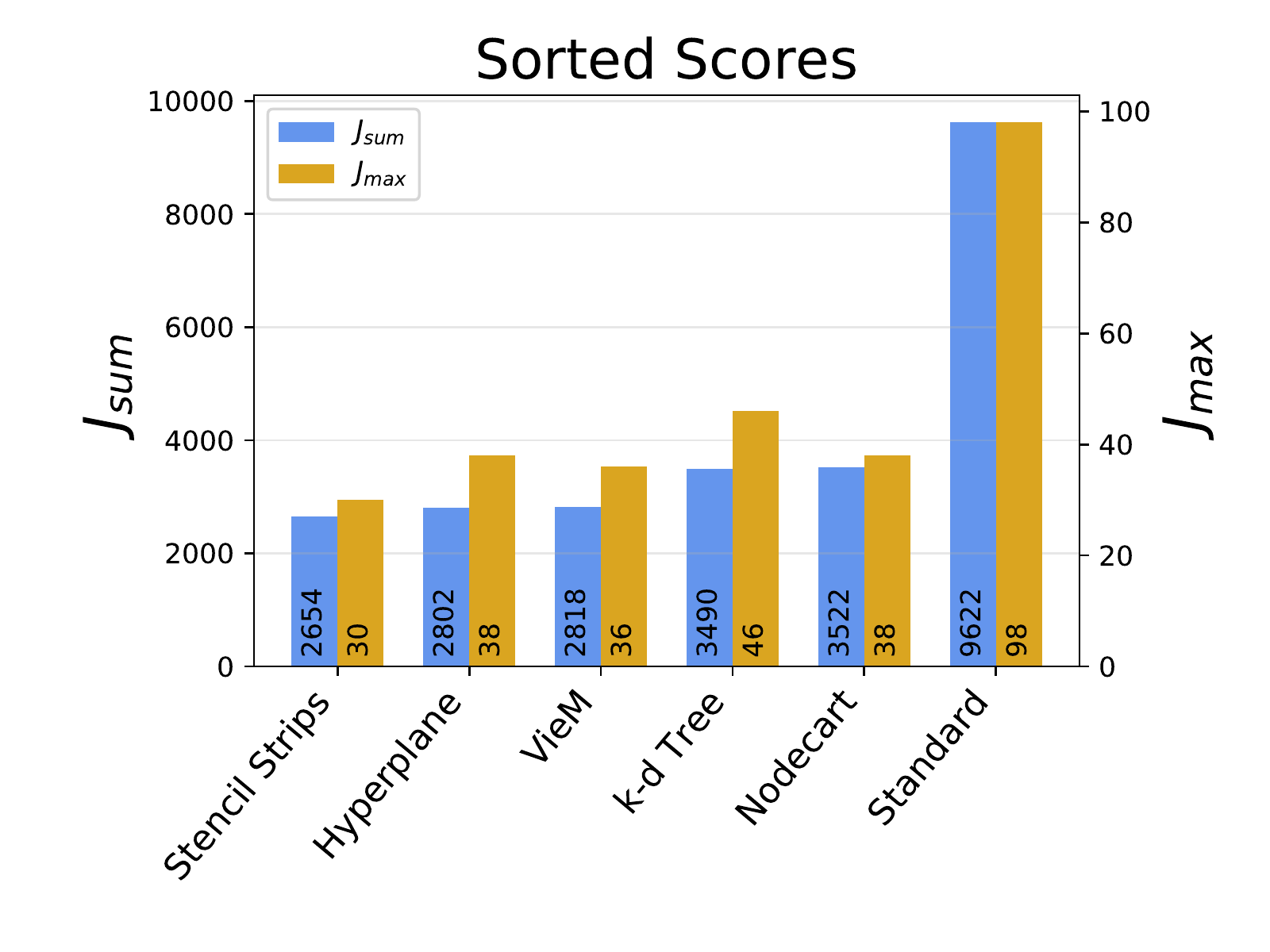}
		\end{subfigure}
		\hfill
		\begin{subfigure}{0.24\textwidth}
			\centering
			\caption{\emph{VSC}$4$\xspace}
			\includegraphics[trim = 0 20 0 30, clip, height=3.1cm]{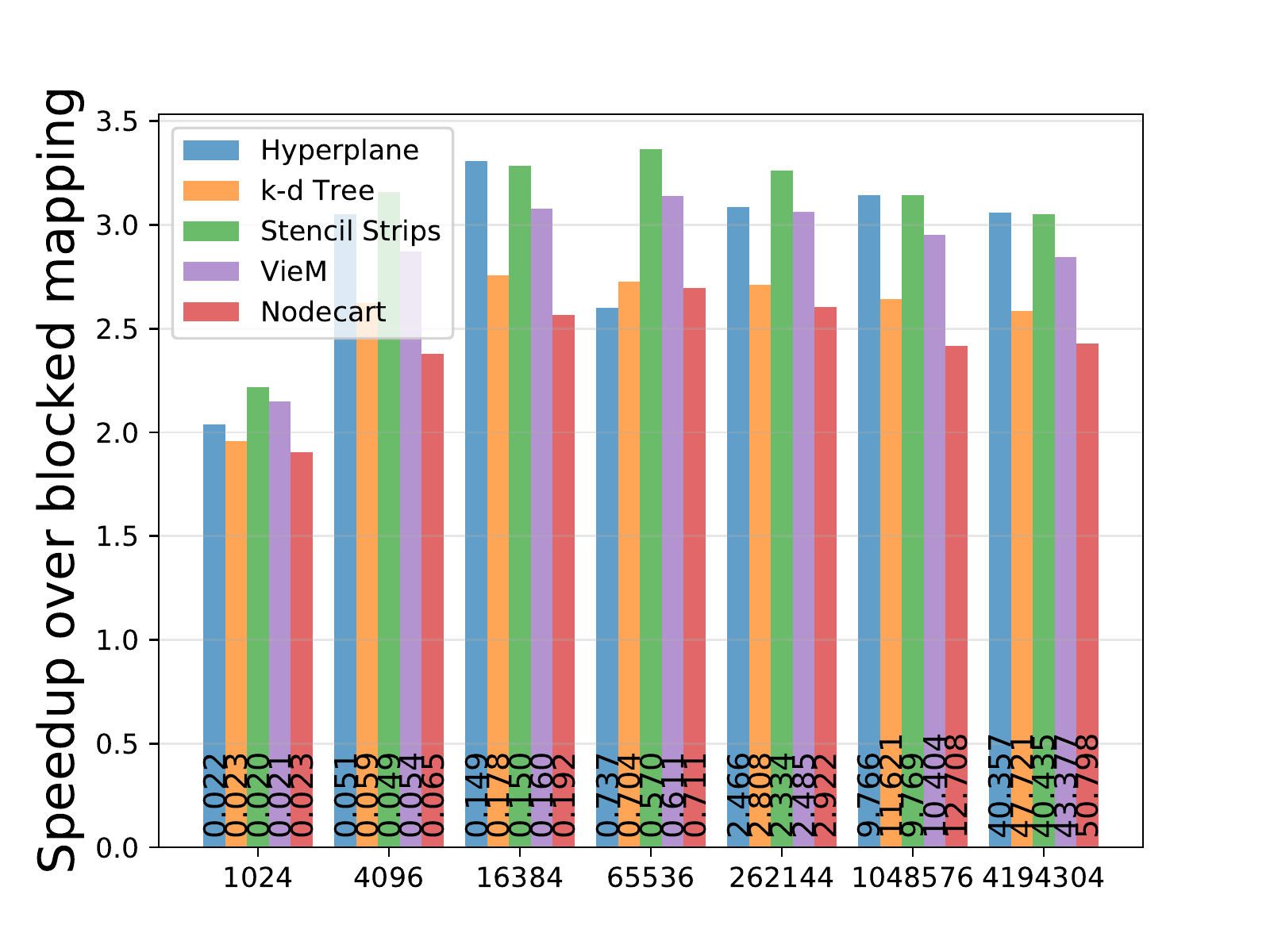}
			\end{subfigure}
		\begin{subfigure}{0.24\textwidth}
			\centering
			\caption{\emph{SuperMUC-NG}\xspace}
			\includegraphics[trim = 0 20 0 30, clip, height=3.1cm]{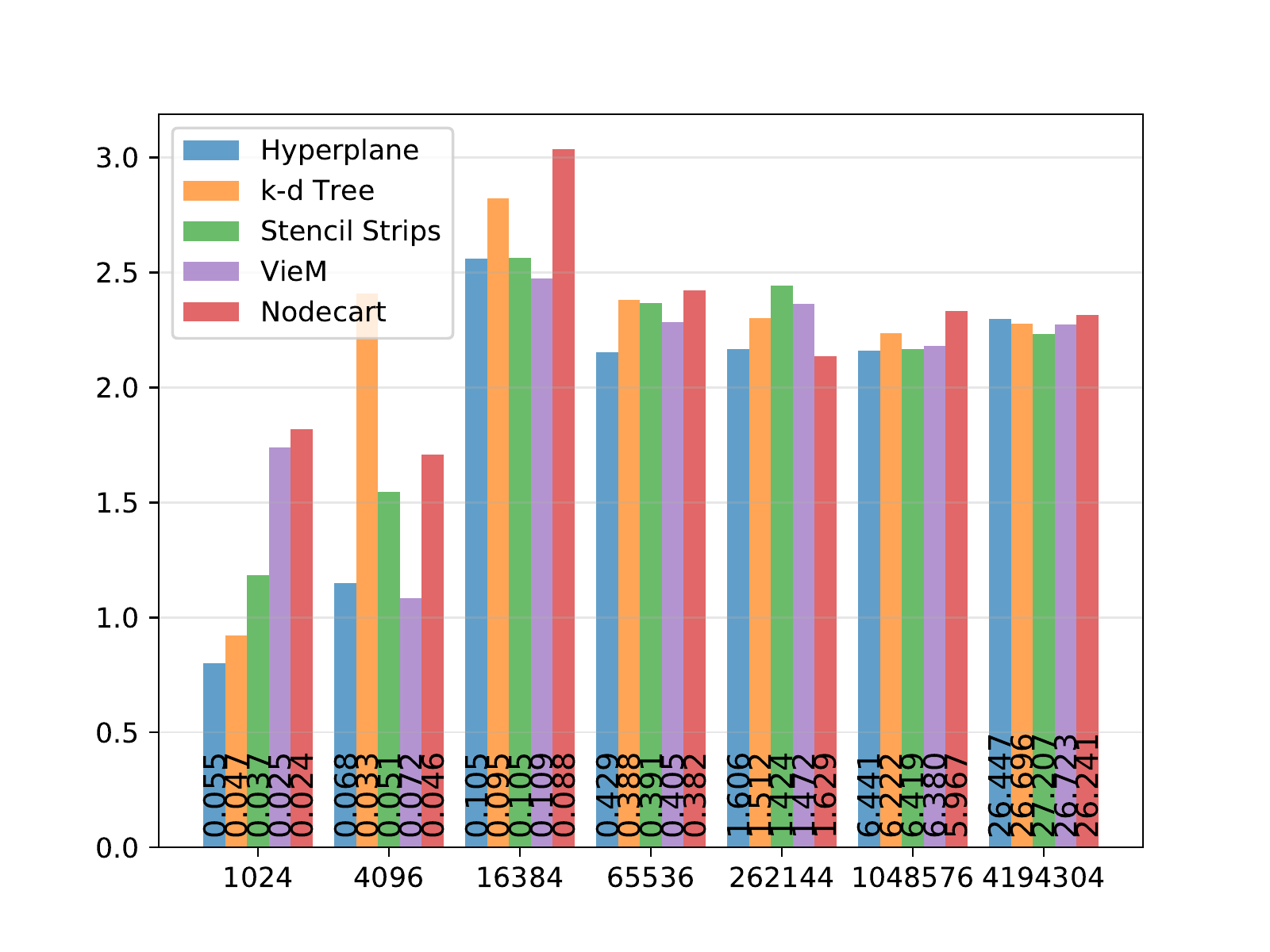}
		\end{subfigure}
		\begin{subfigure}{0.24\textwidth}
			\centering
			\caption{\emph{JUWELS}\xspace}
			\includegraphics[trim = 0 20 0 30, clip, height=3.1cm]{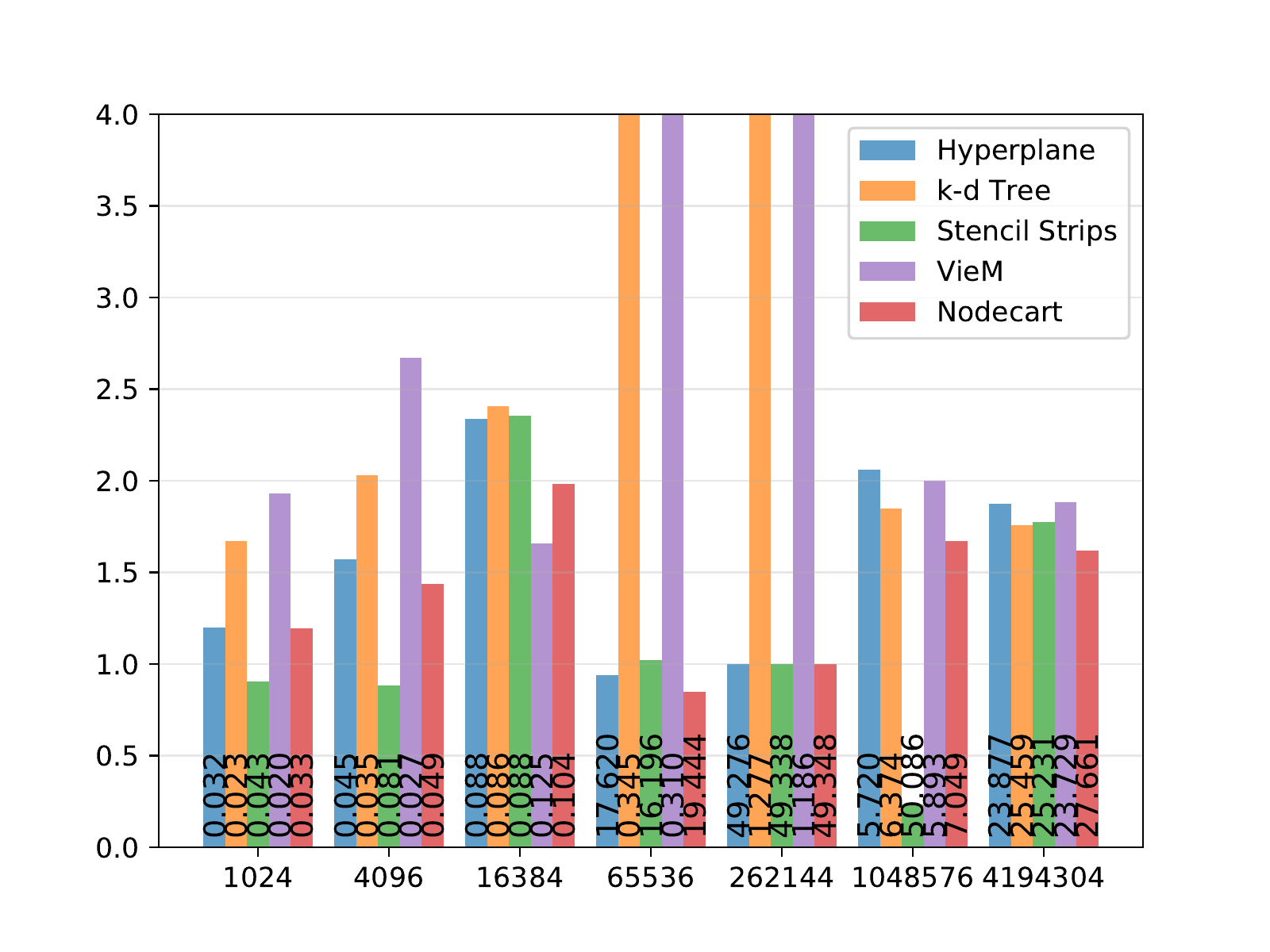}
		\end{subfigure}
	}
	\captionsetup[sub]{labelformat=empty, skip=0pt}
	\subcaptionbox{Nearest neighbor with hops}[\textwidth]['l']{
		\begin{subfigure}{0.24\textwidth}
			\center
			\caption{}
			\includegraphics[trim=0 20 0 10, clip, height=3.3cm]{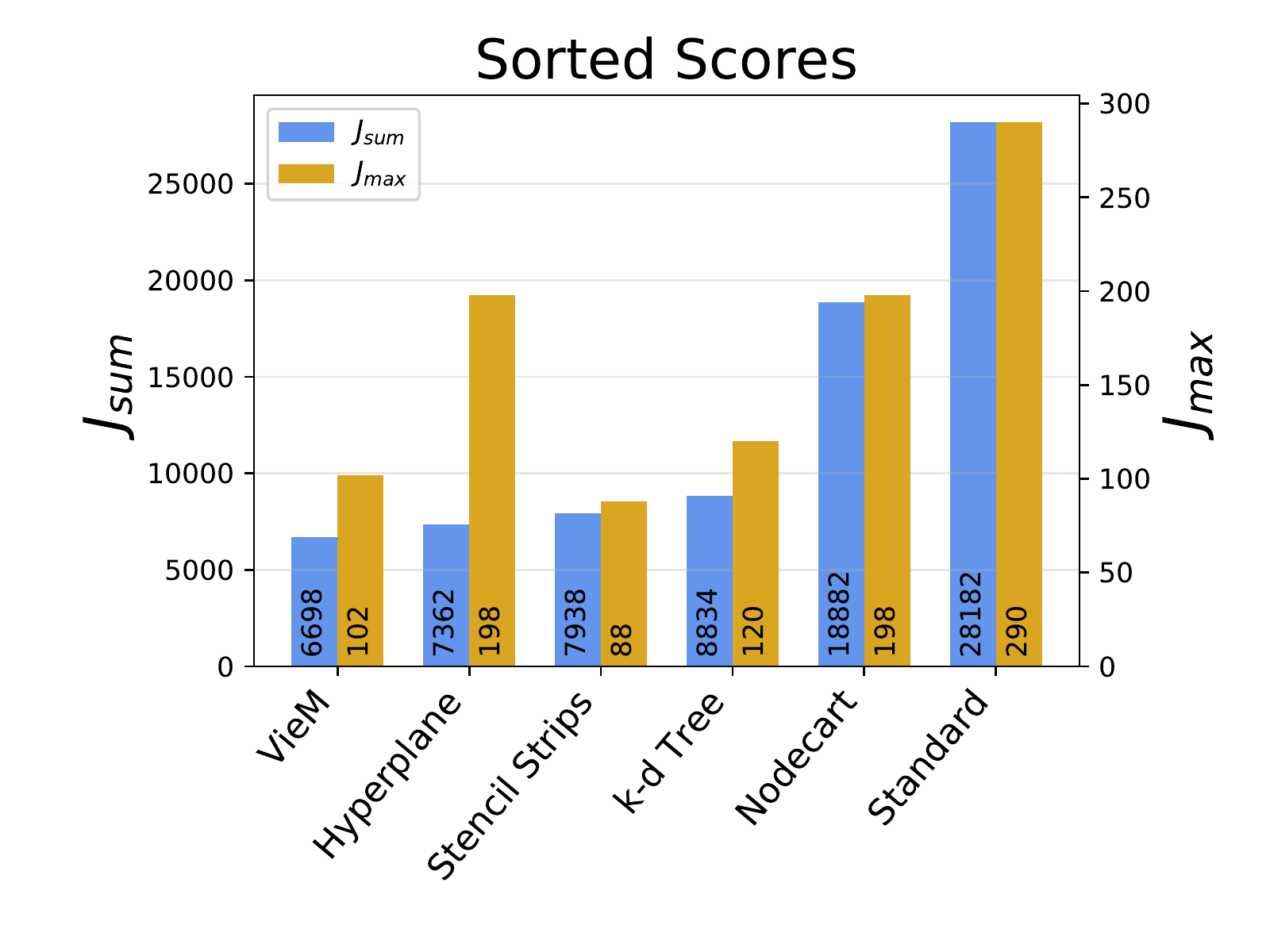}
		\end{subfigure}
		\hfill
		\begin{subfigure}{0.24\textwidth}
			\centering
			\includegraphics[trim = 0 20 0 30, clip, height=3.1cm]{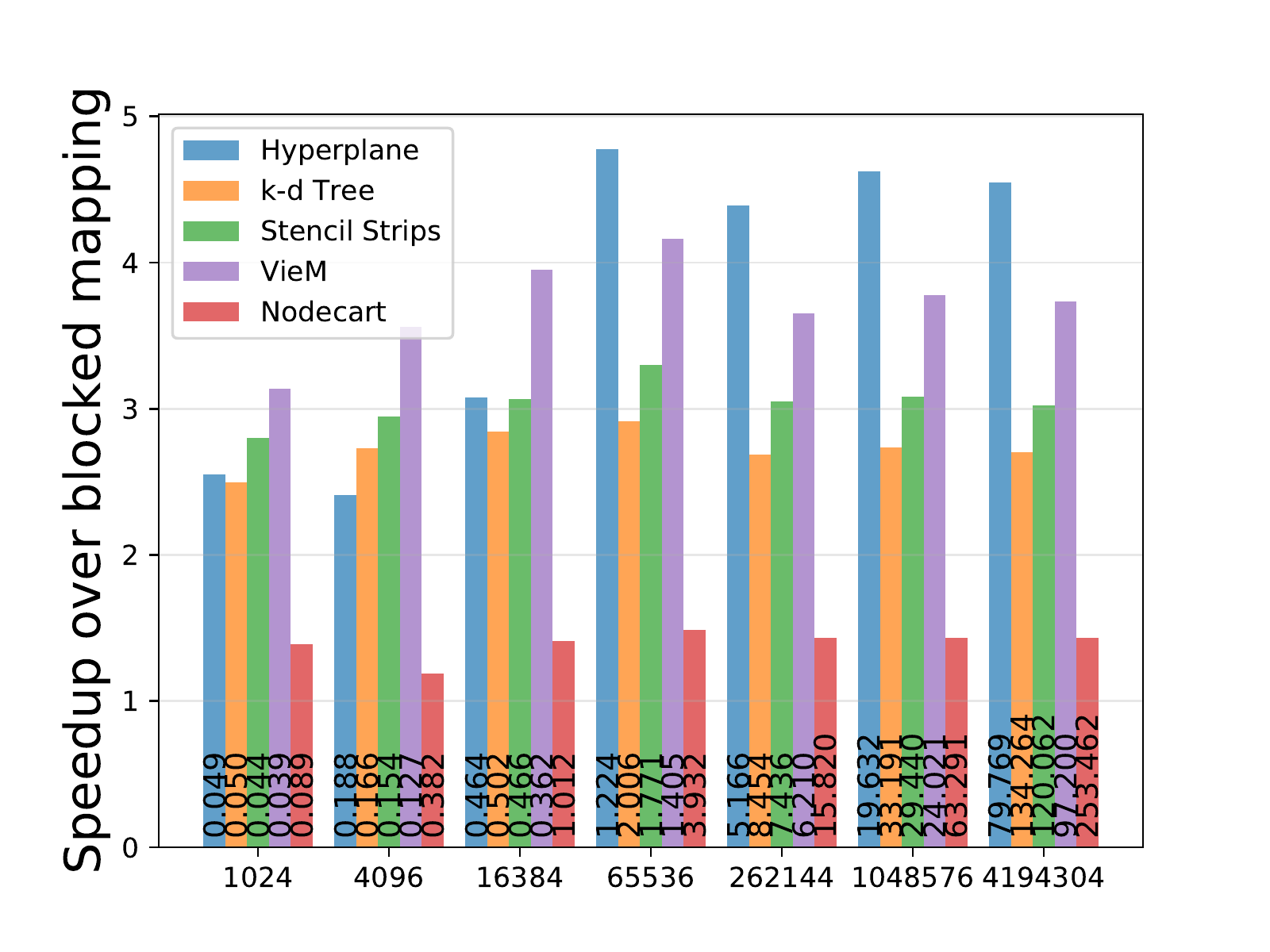}
			\end{subfigure}
		\begin{subfigure}{0.24\textwidth}
			\centering
			\includegraphics[trim = 0 20 0 30, clip, height=3.1cm]{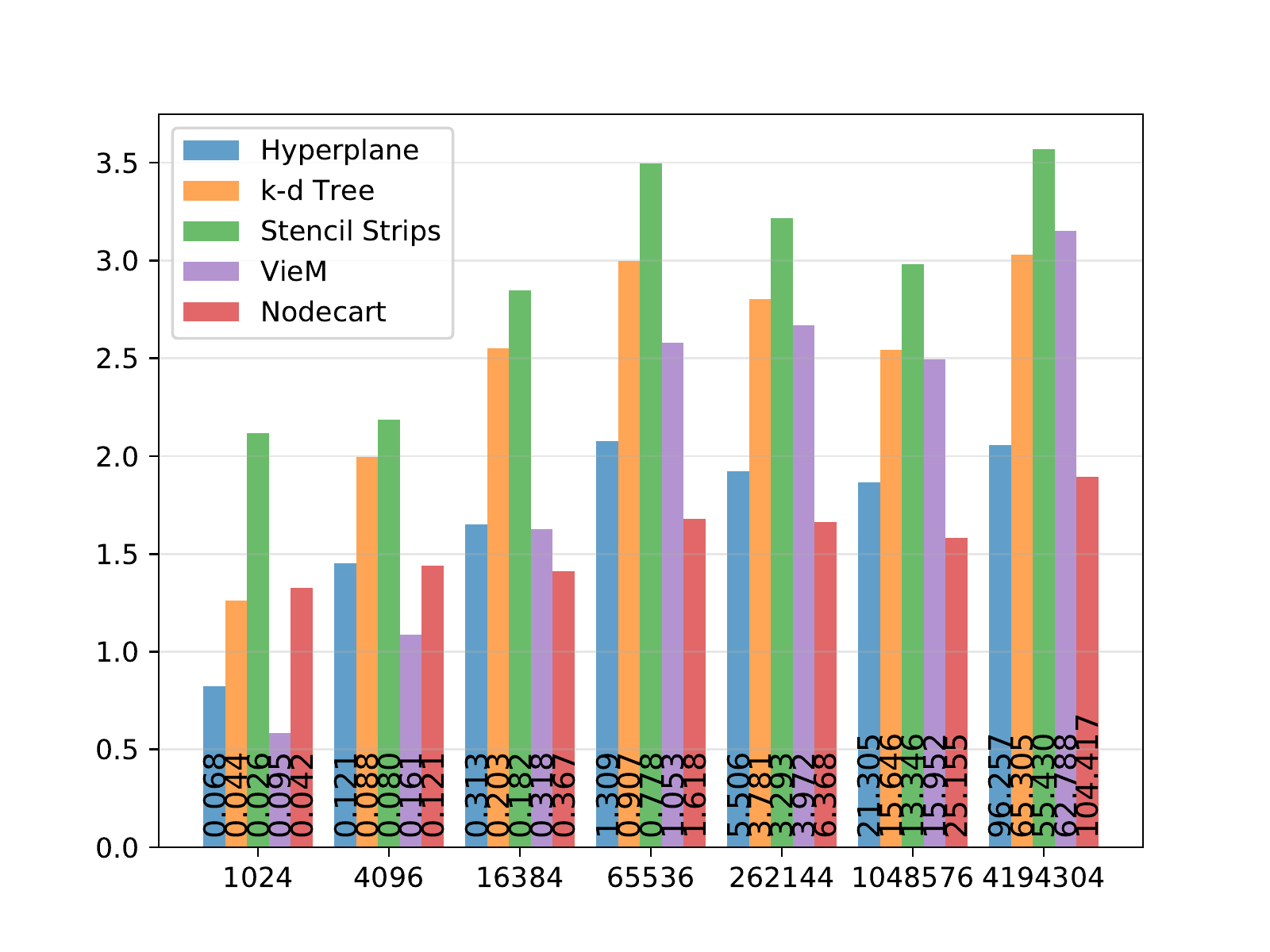}
		\end{subfigure}
		\begin{subfigure}{0.24\textwidth}
			\centering
			\includegraphics[trim = 0 20 0 30, clip, height=3.1cm]{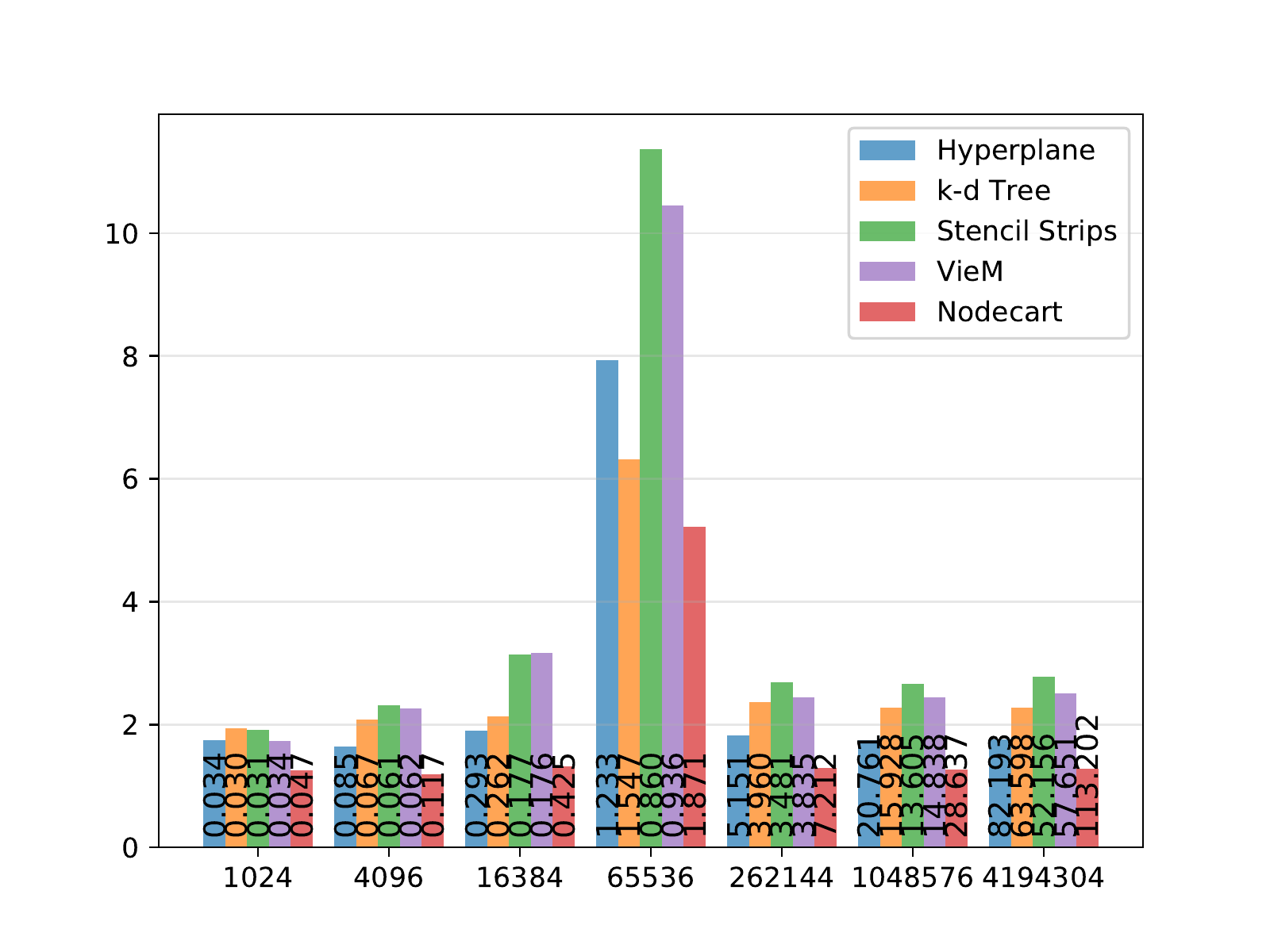}
		\end{subfigure}
	}
	\captionsetup[sub]{labelformat=empty, skip=0pt}
	\subcaptionbox{Component stencil}[\textwidth]['l']{
		\begin{subfigure}{0.24\textwidth}
			\center
			\caption{}
			\includegraphics[trim=0 20 0 10, clip, height=3.3cm]{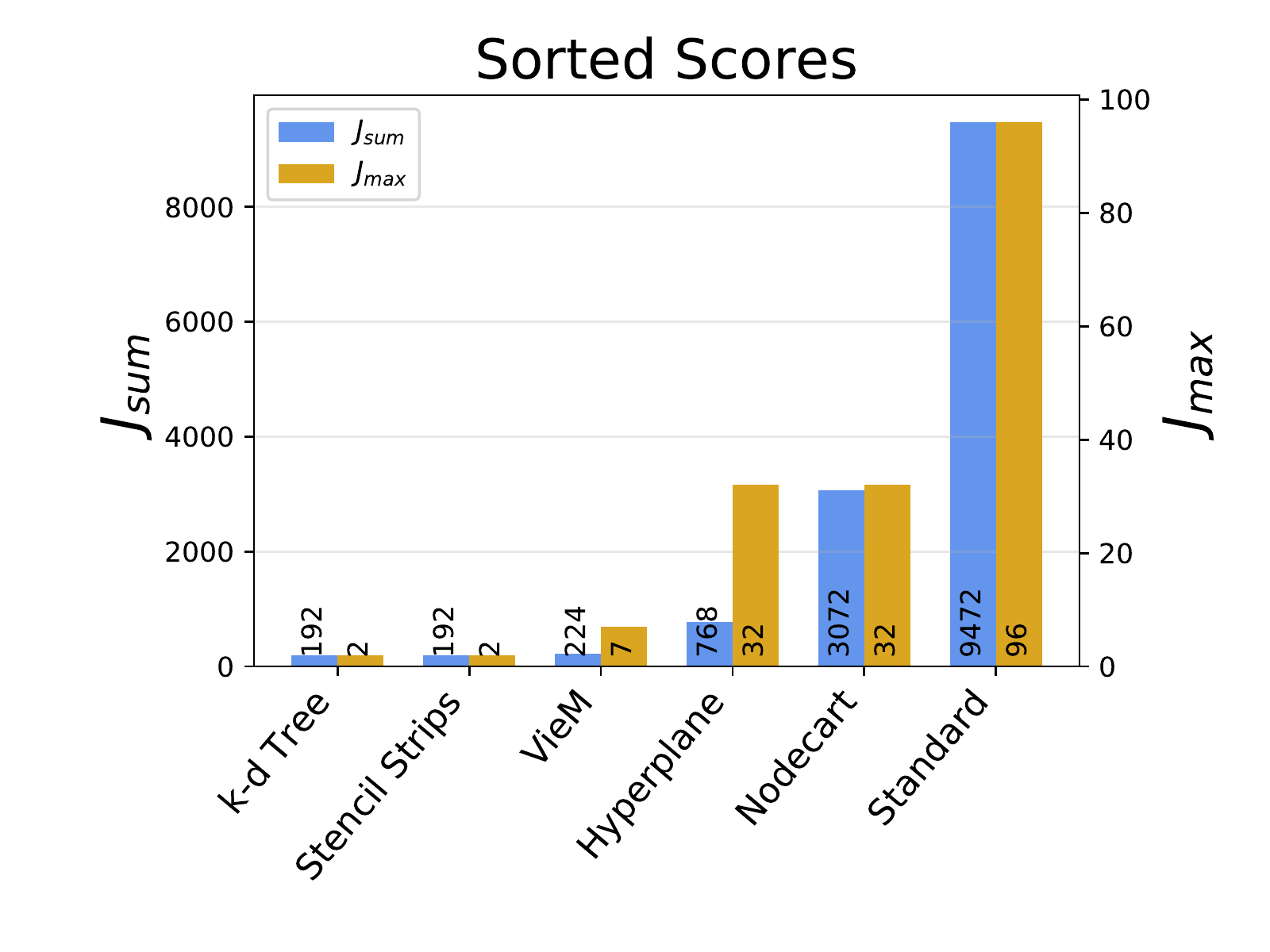}
		\end{subfigure}
		\hfill
		\begin{subfigure}{0.24\textwidth}
			\centering
			\includegraphics[trim = 0 0 0 30, clip, height=3.2cm]{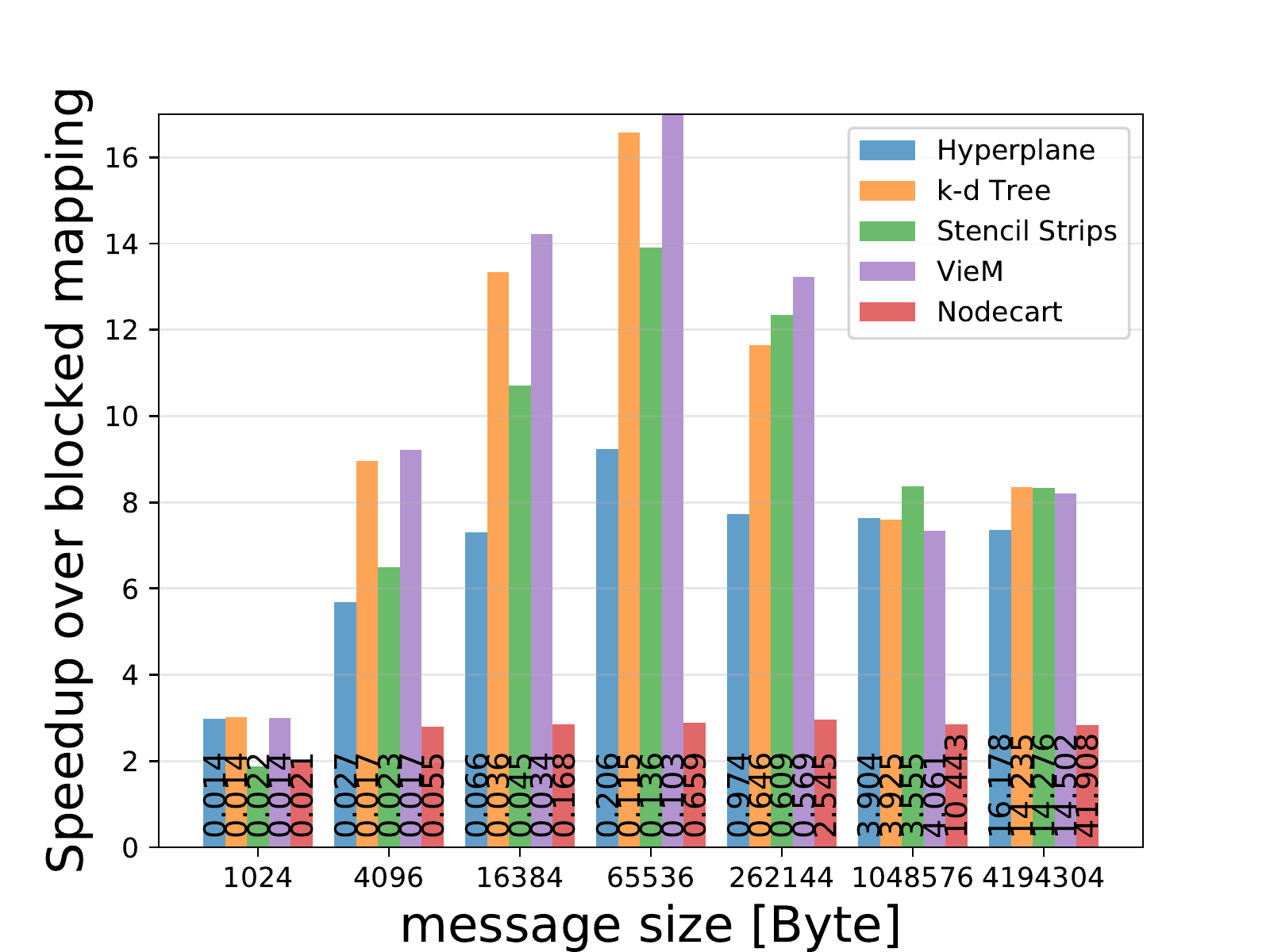}
			\end{subfigure}
		\begin{subfigure}{0.24\textwidth}
			\centering
			\includegraphics[trim = 0 0 0 30, clip, height=3.2cm]{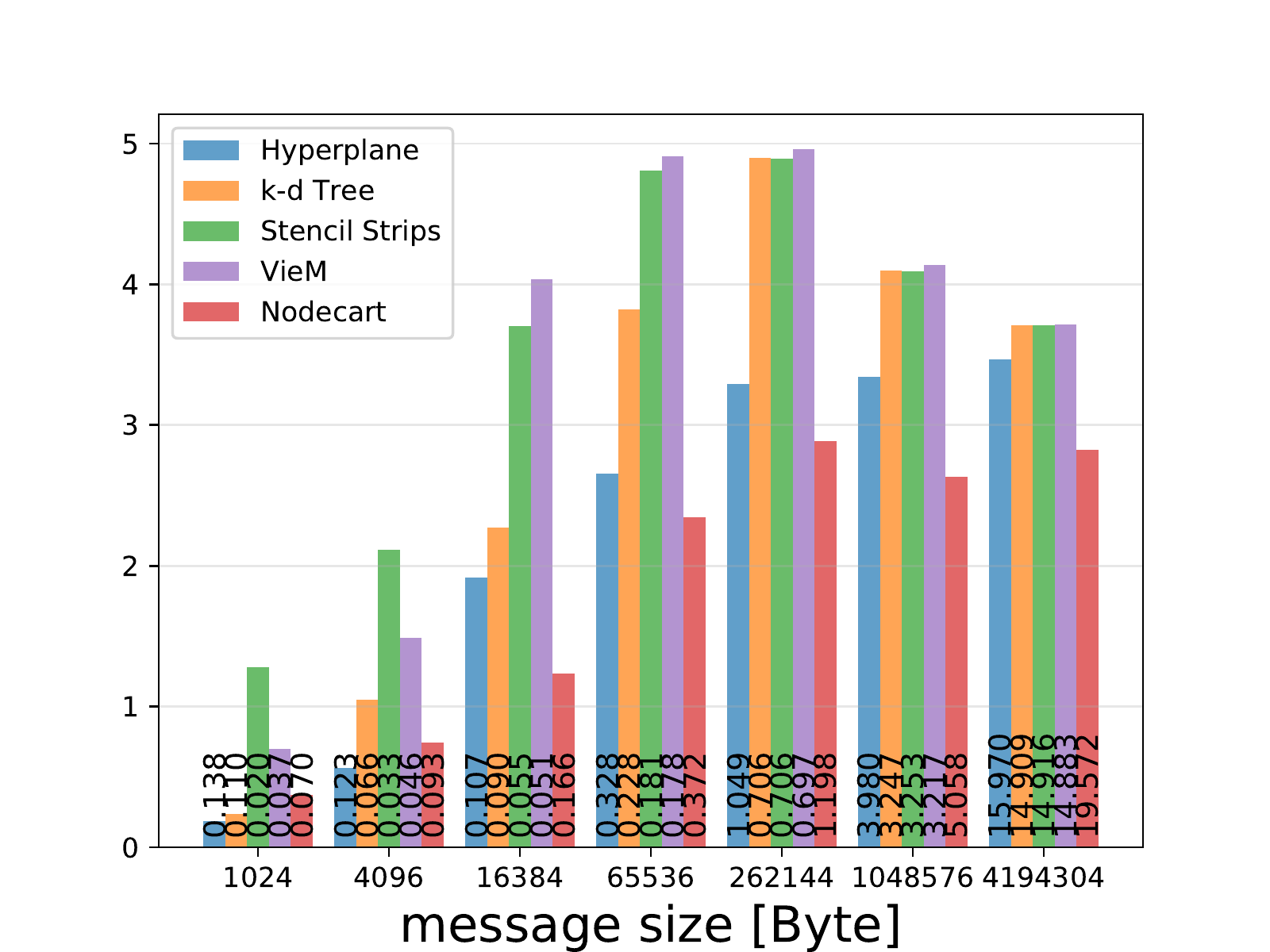}
		\end{subfigure}
		\begin{subfigure}{0.24\textwidth}
			\centering
			\includegraphics[trim = 0 0 0 30, clip, height=3.2cm]{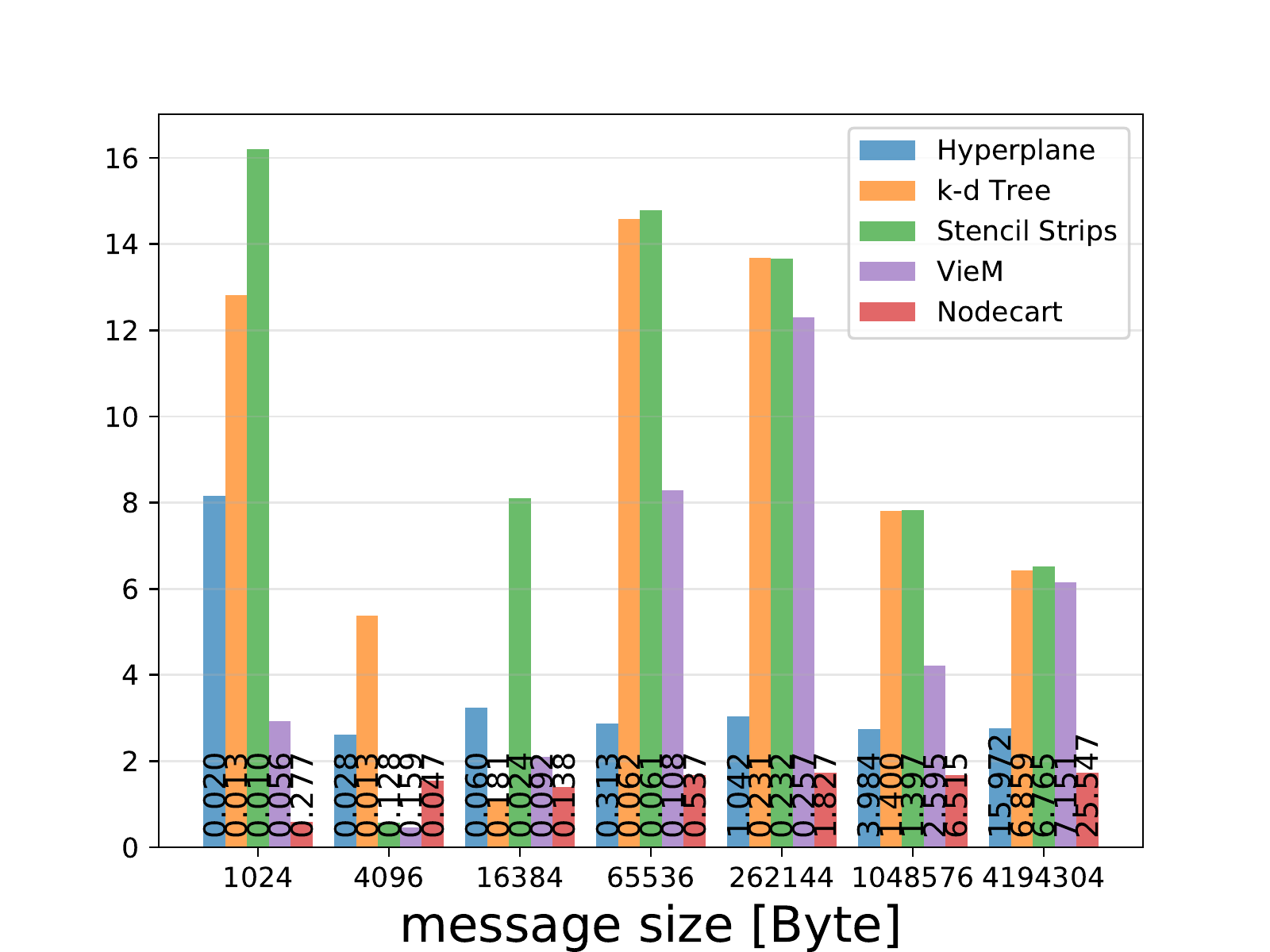}
		\end{subfigure}
	}
	\caption{Left column: scores of the algorithms (smaller is better). 
	Right three columns: Speedup over blocked mapping (higher is better) with
	$N=100$ number of nodes and $p=48$ processes per node and grid sizes
	$75\times 64$ for the nearest neighbor stencil,
	the nearest neighbor stencil with hops and the
	component stencil on the \emph{VSC}$4$\xspace, \emph{SuperMUC-NG}\xspace, and \emph{JUWELS}\xspace. Absolute times 
	are written in the corresponding bar.
	}\label{fig:bandwidth100}
\end{figure*}

\subsection{Machine Description}
\label{sec:systems}
\begin{table}
	\caption{Machines used in the Experiments.}
	\label{table:machines}
	\begin{scriptsize}
		\begin{tabularx}{\the\columnwidth}{l l c r}
		\hline
		Name & Processor & \texttt{MPI}\xspace libraries & Compiler \\
		\hline
		\emph{VSC}$4$\xspace & Intel Skylake Platinum 8174 & Intel\,MPI & icc 19.0.5\\
		\emph{SuperMUC-NG}\xspace & Intel Skylake Platinum 8174 & Intel\,MPI & icc 19.0.5\\
			\emph{JUWELS}\xspace & Intel Xeon Platinum 8168 & Intel\,MPI & icc 19.0.3 \\
			\hline
	\end{tabularx}
	\end{scriptsize}
\end{table}
We perform the experiments in Section~\ref{sec:throughput} 
on the Vienna Scientific Cluster $4$ (\emph{VSC}$4$\xspace), \emph{SuperMUC-NG}\xspace and
\emph{JUWELS}\xspace. \emph{VSC}$4$\xspace is composed of $790$ nodes, where each node is equipped with two Intel Skylake Platinum $8174$ processors running at
\SI{3.1}{\GHz}, \textrm{i.e.}\xspace, each node has $48$ cores. 
The nodes are connected with a two-level fat-tree (blocking factor 2:1) via
OmniPath with a capacity of \SI[per-mode=symbol]{100}{\giga\bit\per\second}.
\emph{SuperMUC-NG}\xspace consists of $\num{6336}$ compute nodes, also equipped with two Intel Skylake Platinum $8174$ processors running at
\SI{3.1}{\GHz}.
The nodes are bundled into islands, with a pruned OmniPath connection
between the islands (pruning factor $1$:$4$) and nodes within an island are connected
via a OmniPath fat-tree.
 \emph{JUWELS}\xspace is composed of $\num{2271}$ compute
nodes, with two Intel Xeon Platinum $8168$, running at \SI{2.7}{\GHz}. The
nodes are connected with a two-level fat-tree InfiniBand network (pruning factor 2:1)
with a capacity of \SI[per-mode=symbol]{100}{\giga\bit\per\second}. A brief summary of the
systems and libraries is given in
Table~\ref{table:machines}.

\subsection{Experimental Setup}
\label{sec:methodology}
\begin{lstlisting}[frame=lines, caption={Interface used for a $k$-neighborhood aware \texttt{MPI}\xspace
Cartesian communicator}, label={stencil_cart}]
int MPIX_Cart_stencil_comm(MPI_Comm oldcomm,
    const int ndims, const int dims[],
    const int periods[], const int reorder,
    const int stencil[], const int k,
    MPI_Comm *cartcomm);
\end{lstlisting}

All algorithms were implemented in \mbox{C\texttt{++}}\xspace$11$.  The code was compiled
with the Intel compiler and full optimization flags (-O$3$) and Intel
MPI.

\begin{figure*}[ht]
	\begin{subfigure}{0.31\textwidth}
		\centering
		\includegraphics[trim = 0 15 0 25, clip, scale=0.35]{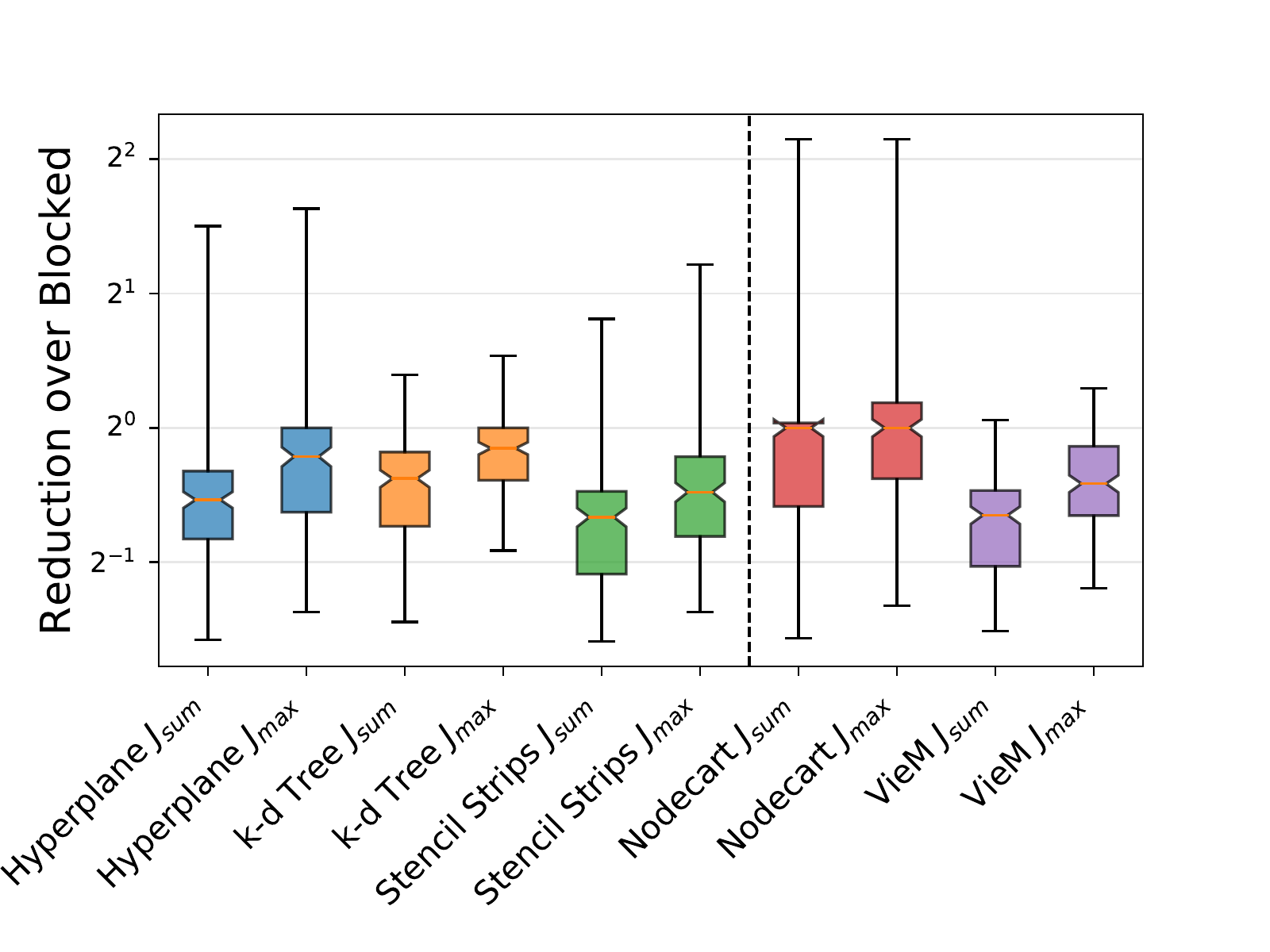}
	\caption{}
	\end{subfigure}
	\hfill
	\begin{subfigure}{0.31\textwidth}
		\includegraphics[trim = 0 15 0 25, clip, scale=0.35]{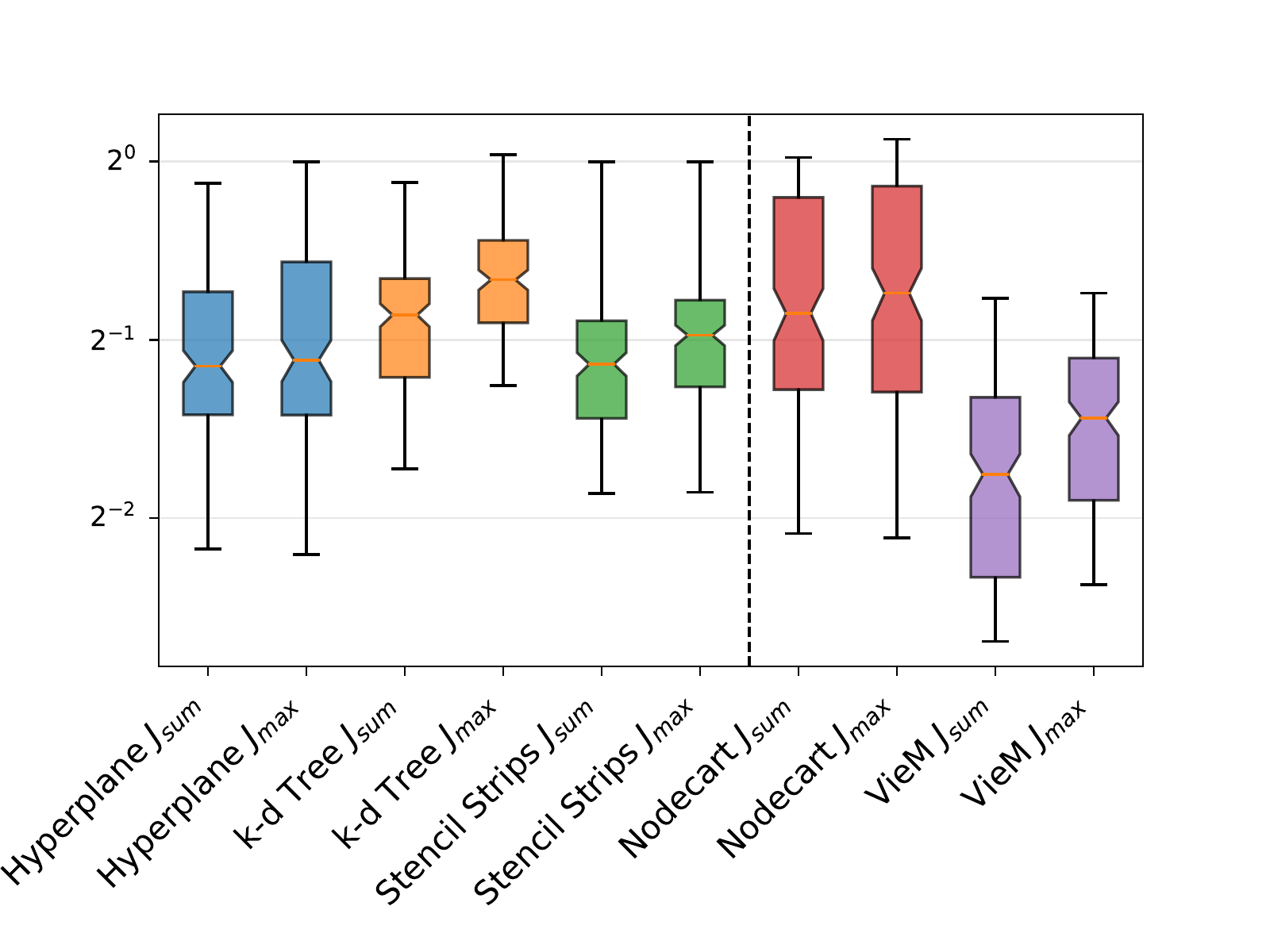}
	\caption{}
	\end{subfigure}
	\hfill
	\begin{subfigure}{0.31\textwidth}
		\includegraphics[trim = 0 15 0 25, clip, scale=0.35]{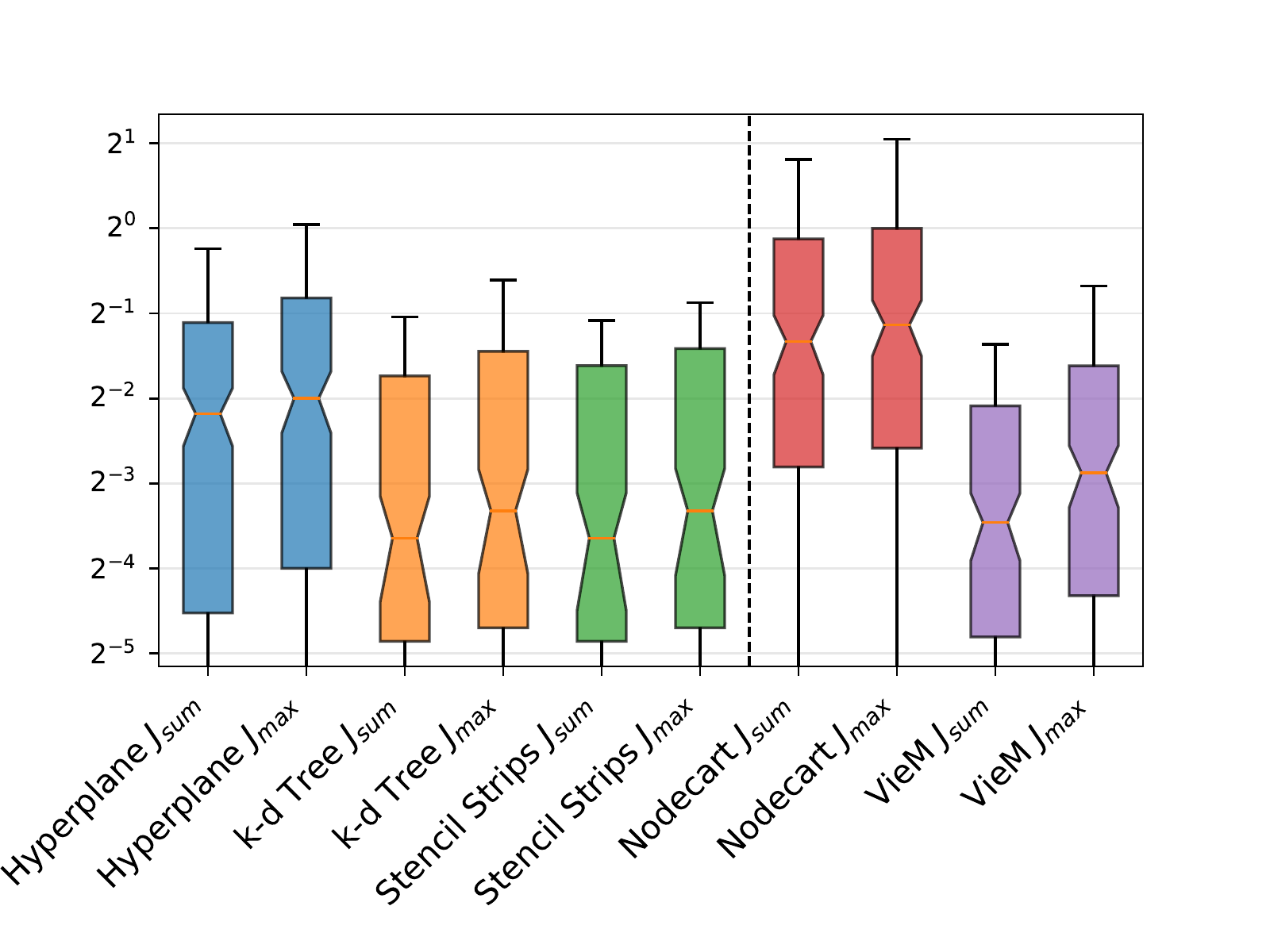}
	\caption{}
	\end{subfigure}
	\caption{Reduction for
	nearest neighbor stencil (a), for the nearest neighbor stencil with
	additional two hops along the first dimension (b) and for the
	artificial component stencil (c). Lower is better. For each algorithm \texttt{sum}
	defines the total reduction, whereas \texttt{max} defines the reduction
	of the bottleneck node. \texttt{Hyperplane}\xspace depicted in blue, $k$-$d$~\texttt{tree}\xspace in
	yellow, \texttt{Stencil Strips}\xspace in green, \texttt{Nodecart}\xspace in red, and \texttt{VieM}\xspace as
	comparison in gray.} 
	\label{fig:distribution}
\end{figure*}
To remap for arbitrary stencils that cannot be expressed with the
\texttt{MPI}\xspace Cartesian interfaces, we use an interface similar to
\texttt{MPI\_\-Cart\_\-create}\xspace as shown in Listing~\ref{stencil_cart}. The array
\texttt{stencil[]} consists of a flattened list of
relative offsets along each dimension per neighbor. The number of
neighbors is~\texttt{k}, thus, the array
\texttt{stencil[]} is of length~$\texttt{k}\cdot \texttt{ndims}$.

As for the $k$-neighborhoods, we choose to test with the presented 
$2$-dimensional stencils from Section~\ref{sec:preliminaries}, \textrm{i.e.}\xspace, 
nearest neighbor stencil, the nearest neighbor stencil with hops and
the component stencil, as depicted in Figure~\ref{fig:stencils}. 

All grids were created according to the \texttt{MPI\_\-Dims\_\-create}\xspace specifications,
that is with the sizes of the dimensions being as close as possible to
each other~\cite{MPI-3.1,Traff15:dimscreate}.

For the stencil exchange, we instantiated a distributed graph communicator from
the Cartesian communicator and the $k$-neighborhood in order to call the
\texttt{MPI\_\-Neighbor\_\-alltoall}\xspace routine. Synchronization between the processes before each
collective message exchange was done with an \texttt{MPI\_\-Barrier}\xspace and we define the time
needed for the operation, as the maximal time any process spent in the
\texttt{MPI\_\-Neighbor\_\-alltoall}\xspace routine.  As we assumed unit-weighted communication edges,
every process sends and receives the same amount of data to its communication neighbors.

\subsection{Inter-Node Communication Analysis}
\label{sec:reduction}
\begin{figure}[b]
	\center
	\includegraphics[trim= 0 70 0 30, clip, scale=0.35]{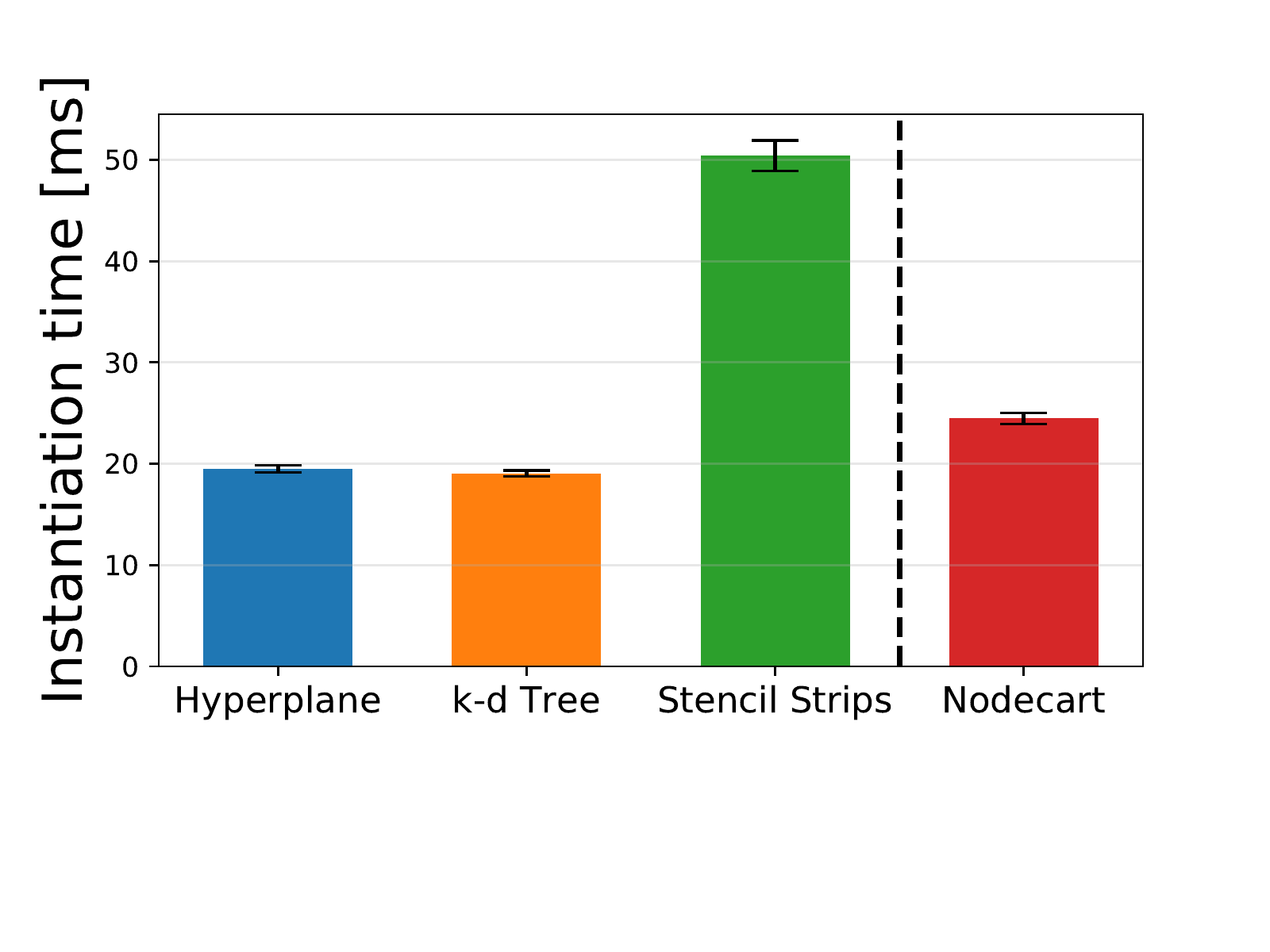}
	\caption{Instantiation time in \si{\milli\second} on a $100\times 48$
	nearest neighbor stencil instance, run on \emph{VSC}$4$\xspace (lower is better). Each algorithm was
	instantiated $200$ times and after outlier removal, we plot the
	mean with a $95\%$ confidence interval. \texttt{VieM}\xspace was omitted since it 
	took two orders of magnitude longer and would distort~the~scale.}
	\label{fig:instantiation} 
\end{figure}

In this section, we investigate the
reduction of~$J_\text{sum}$ and~$J_\text{max}$ of different algorithms over a blocked mapping.
Note that the evaluation of the objective function is machine independent.
We evaluate the performance of the algorithms on a varying input
number of compute nodes, processes per compute node, and number of dimensions.  To
be more precise, the number of nodes is given by the set~$\mathcal{N}=\{10, 13,
16, \dots, 33\}$ and the number of processes per node by the
set~$\mathcal{P}=\{10, 13, 16, \dots, 31\} \cup \{32\}$, while we restrict the set
of dimensions to be $\mathcal{D}=\{2, 3\}$.  The set of instances is therefore
the Cartesian product~$\mathcal{I}=\mathcal{N} \times \mathcal{P} \times
\mathcal{D}$, with $|\mathcal{I}|=144$. We define the reduction to
be~$\frac{C_X}{C_b}$, where~$C_X$ is the amount of inter-node communication
induced by algorithm~$X$ and~$C_b$ the amount of inter-node communication
induced by a blocked mapping. We compare both the
total reduction of inter-node communication~$J_\text{sum}$, as well as the
reduction over the bottleneck node~$J_\text{max}$, as defined in
Section~\ref{sec:preliminaries}.  We evaluate the presented algorithms and
compare ourselves to Gropp's algorithm (\texttt{Nodecart}\xspace) and \texttt{VieM}\xspace, both described
in Section~\ref{sec:relatedwork}.  The configuration for \texttt{VieM}\xspace was set to the
strongest setting, prioritizing quality in terms of reduction of inter-node
communication instead of speed. As for the local-search neighborhood, we
allowed swaps between any connected pair of vertices, \textrm{i.e.}\xspace, we considered the
largest search space. We set the \texttt{hierarchy\_parameter\_string}
and the \texttt{distance\_parameter\_string} to the $n$:$N$ and $0$:$1$, respectively
so that \texttt{VieM}\xspace optimizes our objective function. 

The results for the three different $k$-neighborhood communication patterns can
be seen in Figure~\ref{fig:distribution}. We plot the median improvement (orange)
with a~$95\%$ confidence interval (indicated by the notches), calculated with a
Gaussian-based asymptotic approximation. Since the confidence interval of
the medians do not overlap, we can say with statistical evidence~\cite{murphy2019probability} that the median
reduction improvement of \texttt{Hyperplane}\xspace and \texttt{Stencil Strips}\xspace is better than
\texttt{Nodecart}\xspace in all communication patterns.
Only for the nearest neighbor with hops pattern, there
seems to be no statistical reduction
between the $k$-$d$~\texttt{tree}\xspace and \texttt{Nodecart}\xspace. For the
nearest neighbor and component stencil, the reduction distribution between the
\texttt{Stencil Strips}\xspace and \texttt{VieM}\xspace looks very similar and there is no statistical
difference in the median reduction, indicating that the \texttt{Stencil Strips}\xspace
algorithm, which is not bound to factorization in any way, can greatly exploit
the structure of the grid partitioning problem.

\subsection{Throughput Analysis}
\label{sec:throughput}
We continue by examining the influence of rank reordering on the time
needed for an \texttt{MPI\_\-Neighbor\_\-alltoall}\xspace exchange and different message sizes to be
sent to each communication partner. The experiment was conducted on
\emph{VSC}$4$\xspace, \emph{SuperMUC-NG}\xspace, and \emph{JUWELS}\xspace described in Section~\ref{sec:systems}. In order to 
see how the reordering affects the communication performance of larger instances,
we perform the experiments with $50$ and $100$ nodes and
for each used the full number of processes per node $48$. Each message was sent
$200$ times to obtain large samples.
From each run, we capture the maximum time needed over all processes. Before
each exchange, we synchronize the processes with \texttt{MPI\_\-Barrier}\xspace. In
Figure~\ref{fig:bandwidth50} and Figure~\ref{fig:bandwidth100}, we plot the
mean speedup over the blocked assignments, after removing outliers beyond $1.5$
inter-quartile range from the third and first quartile, respectively.
Along with the speedup, we plot the scores of the
algorithms sorted in increasing order of $J_\text{sum}$ and $J_\text{max}$.
The message size is the number of bytes send per neighbor.
In the appendix we include tables with the absolute
running times and the communication performance of a random assignment of
processes to compute nodes (\texttt{Random}\xspace). Due to constraint space, we omitted
\texttt{Random}\xspace in the speedup plots.

For the nearest neighbor stencil and $N=50$, Figure~\ref{fig:bandwidth50},
\texttt{Hyperplane}\xspace and \texttt{Stencil Strips}\xspace are up to three and four times faster than the
blocked assignment on the \emph{VSC}$4$\xspace, up to two and almost
up to three times faster on \emph{SuperMUC-NG}\xspace and more than five times faster on
\emph{JUWELS}\xspace. All of the presented algorithms are able to outperform \texttt{Nodecart}\xspace in
terms of the mapping metric $J_\text{sum}$, $J_\text{max}$ and communication
performance on \emph{VSC}$4$\xspace and mostly on \emph{SuperMUC-NG}\xspace and \emph{JUWELS}\xspace.  As for \texttt{VieM}\xspace,
\texttt{Stencil Strips}\xspace outperforms it consistently on all machines, whereas
\texttt{Hyperplane}\xspace and $k$-$d$~\texttt{tree}\xspace attain similar or better performance.
Surprisingly, the communication performance of \texttt{VieM}\xspace on \emph{JUWELS}\xspace is worse than
predicted by $J_\text{sum}$, $J_\text{max}$.  In the case of $N=100$,
Figure~\ref{fig:bandwidth100}, \texttt{Hyperplane}\xspace and \texttt{Stencil Strips}\xspace obtain speedups
of up to three over the blocked assignment on \emph{VSC}$4$\xspace, up to
$2.5$ on \emph{SuperMUC-NG}\xspace and up to $2$ on \emph{JUWELS}\xspace. For message sizes
$\num{65536}$, $\num{262144}$ the performance of the blocked assignment,
\texttt{Hyperplane}\xspace, \texttt{Stencil Strips}\xspace, and \texttt{Nodecart}\xspace performed very badly,
resulting in speedups of more than forty, which is why we cut the axis at $4$. 

In the case of the nearest neighbor with hops stencil for $N=50$, Figure~\ref{fig:bandwidth50},
\texttt{Hyperplane}\xspace, \texttt{Stencil Strips}\xspace obtained speedups of up to four on the \emph{VSC}$4$\xspace,
up to $3.5$ on \emph{SuperMUC-NG}\xspace and up to three on \emph{JUWELS}\xspace. \texttt{VieM}\xspace outperforms the presented 
algorithms and both \emph{VSC}$4$\xspace and \emph{SuperMUC-NG}\xspace, but \texttt{Stencil Strips}\xspace obtained similar
to better performance for large message sizes on \emph{JUWELS}\xspace. In comparison to \texttt{Nodecart}\xspace,
the presented algorithms seemed to be up to two and three times faster. 
For $N=100$, Figure~\ref{fig:bandwidth100}, the presented algorithms are up to a factor
$4.5$ faster than the blocked assignment on \emph{VSC}$4$\xspace, up to $3.5$ faster 
on \emph{SuperMUC-NG}\xspace and more than a factor $2$ faster on \emph{JUWELS}\xspace. It seems that the 
performance of the blocked assignment for message size
$\num{65536}$ performed relatively bad in comparison to the other message
sizes. On the \emph{VSC}$4$\xspace, our algorithms are more than a factor two faster than
\texttt{Nodecart}\xspace, whereas on \emph{SuperMUC-NG}\xspace and \emph{JUWELS}\xspace we outperform \texttt{Nodecart}\xspace for
most message sizes by a factor of more than $1.3$.

As for the synthetic component stencil and $N=50$, only $k$-$d$~\texttt{tree}\xspace and \texttt{Stencil Strips}\xspace 
managed to find an optimal mapping, where each compute node has two outgoing communication
edges, as described in Section~\ref{sec:hardness}. This reordering results in speedups
of up to ten over the blocked assignment on \emph{VSC}$4$\xspace, up to five on \emph{SuperMUC-NG}\xspace
and up twelve on \emph{JUWELS}\xspace. Interestingly, \texttt{VieM}\xspace seems to have similar performance on 
\emph{VSC}$4$\xspace and \emph{SuperMUC-NG}\xspace, and outperforms our algorithms for small message sizes on 
\emph{JUWELS}\xspace. 
In the case of $N=100$, again only $k$-$d$~\texttt{tree}\xspace and \texttt{Stencil Strips}\xspace found optimal mappings,
leading to speedups of up to fourteen on \emph{VSC}$4$\xspace, five on \emph{SuperMUC-NG}\xspace, and sixteen on
\emph{JUWELS}\xspace. Surprisingly, \texttt{VieM}\xspace seems to match and outperform the performance
of our algorithms, even though it has a factor $1.6$ more $J_\text{sum}$
and a factor $8.5$ greater $J_\text{max}$.

\subsection{Instantiation Time}
\label{sec:instantiation}

Since the theoretical complexity of the three presented algorithms, \texttt{Nodecart}\xspace and \texttt{VieM}\xspace
all dependent on different parameters, we benchmark the algorithmic runtime
needed to calculate the new ranks only on the largest nearest neighbor stencil
instance, described in Section~\ref{sec:throughput} ($N=100$) on \emph{VSC}$4$\xspace. The algorithmic runtime
does not include the instantiation time of the reordered Cartesian
communicator, rather, it only includes the necessary operations to calculate
the new ranks. For the presented algorithms, this includes creating a sorted
communicator, with precise knowledge over which ranks are assigned to which
nodes. This custom communicator consists of a \emph{node communicator} (grouped
by processes having shared memory) and a \emph{leader communicator}, which
consists of one process per compute node. In this benchmark, the $k$-$d$~\texttt{tree}\xspace
algorithm was implemented with a linear search, for finding the dimension along
which to split, resulting in a theoretical run-time of $\mathcal{O}(d\log p)$.
Each algorithm was instantiated~$200$ times for large sample sizes, and we measured the longest time
needed over all processes. Before each instantiation, the processes were
synchronized using \texttt{MPI\_\-Barrier}\xspace. In Figure~\ref{fig:instantiation}, we plot the
mean instantiation time, after outlier removal (beyond $1.5$ inter-quartile
range from the first and third quartile) with a $95\%$ confidence interval. We
did not include \texttt{VieM}\xspace, since it took on average over \SI{7.95}{\second} 
(a factor of more than $400$) and it would distort the scale of the plot.

We can see, that for this instance the \texttt{Hyperplane}\xspace and $k$-$d$~\texttt{tree}\xspace algorithm 
are the two fastest, and statistically equivalent. \texttt{Nodecart}\xspace seems to 
need $28\%$ longer, whereas the convoluted process of calculating the
strips and strip positions results in \texttt{Stencil Strips}\xspace being the slowest 
algorithm, more than a factor $2$ slower than \texttt{Hyperplane}\xspace and $k$-$d$~\texttt{tree}\xspace.

\section{Conclusion}
\label{sec:conclusion}
We introduced three new efficient algorithms for process to compute node assignment on Cartesian grids and stencils communication patterns. 
By thoroughly exploiting the inherently present structure of the problem, we arrive at algorithms that outperform the state-of-the-art in terms of running time and communication performance.
We implemented the algorithms for \texttt{MPI\_\-Cart\_\-create}\xspace as reordering functions and performed
extensive benchmarks. 
 An intensive experimental evaluation shows that our algorithms are up to two
 orders of magnitude faster in running time than a (sequential) high-quality
 general graph mapping tool \texttt{VieM}\xspace, while obtaining similar results in
 communication performance. Furthermore, our algorithms are three times
 faster in an \texttt{MPI\_\-Neighbor\_\-alltoall}\xspace exchange than a state-of-the-art Cartesian grid mapping
 algorithm (\texttt{Nodecart}\xspace) by achieving a significantly better mapping quality.
 Considering the good results, we plan to release the algorithms and integrate
 them into publicly available \texttt{MPI}\xspace implementations.

\subsection*{Acknowledgments}
This work was partially supported by the Austrian Science Fund (FWF): project P
31763-N31.  The computational results presented have been achieved in part
using the Vienna Scientific Cluster (VSC).  We acknowledge PRACE for awarding
us access to JUWELS at GCS@FZJ, Germany and SuperMUC-NG at GCS@LRZ, Germany.

\FloatBarrier
\bibliographystyle{IEEEtran}
\bibliography{reorder} 
\newpage
\appendix
	\begin{sidewaystable*}
		\begin{scriptsize}
		\caption{Time needed for an \texttt{MPI\_\-Neighbor\_\-alltoall}\xspace exchange, different $k$-neighborhoods and reordering algorithms on \emph{VSC}$4$\xspace with $N=50$ and $p=48$. 
			Experiment performed as described in~\ref{sec:throughput}. We present the mean time in \SI{}{\milli\second} and the~$95\%$ confidence interval range. }
		\begin{widetable}{\textwidth}{l  r r r r r r r r}
			\toprule
			Stencil & Size [\SI{}{\byte}] & \texttt{Blocked}\xspace & \texttt{Hyperplane}\xspace & $k$-$d$~\texttt{tree}\xspace & \texttt{Stencil Strips}\xspace & \texttt{Nodecart}\xspace & \texttt{VieM}\xspace & \texttt{Random}\xspace \\
                        \midrule
\multirow{14}{*}{\shortstack[l]{Nearest\\ neighbor}}
& $\num{64}$ & $0.017_{-0.000}^{+0.000}$ & $0.014_{-0.001}^{+0.001}$ & $0.013_{-0.001}^{+0.001}$ & $0.013_{-0.001}^{+0.001}$ & $0.017_{-0.001}^{+0.001}$ & $0.011_{-0.000}^{+0.000}$ & $0.027_{-0.000}^{+0.000}$ \\ 
& $\num{128}$ & $0.027_{-0.000}^{+0.000}$ & $0.016_{-0.001}^{+0.001}$ & $0.017_{-0.001}^{+0.001}$ & $0.015_{-0.001}^{+0.001}$ & $0.023_{-0.001}^{+0.001}$ & $0.015_{-0.001}^{+0.001}$ & $0.044_{-0.000}^{+0.000}$ \\ 
& $\num{256}$ & $0.048_{-0.000}^{+0.000}$ & $0.022_{-0.001}^{+0.001}$ & $0.022_{-0.001}^{+0.001}$ & $0.020_{-0.001}^{+0.001}$ & $0.030_{-0.001}^{+0.001}$ & $0.021_{-0.000}^{+0.000}$ & $nan_{-nan}^{+nan}$ \\ 
& $\num{512}$ & $0.100_{-0.000}^{+0.000}$ & $0.033_{-0.000}^{+0.000}$ & $0.035_{-0.001}^{+0.001}$ & $0.030_{-0.001}^{+0.001}$ & $0.051_{-0.000}^{+0.000}$ & $0.037_{-0.000}^{+0.000}$ & $0.160_{-0.000}^{+0.000}$ \\ 
& $\num{1024}$ & $0.199_{-0.000}^{+0.000}$ & $0.061_{-0.000}^{+0.000}$ & $0.065_{-0.000}^{+0.000}$ & $0.052_{-0.001}^{+0.001}$ & $0.095_{-0.000}^{+0.000}$ & $0.073_{-0.001}^{+0.001}$ & $0.315_{-0.001}^{+0.001}$ \\ 
& $\num{2048}$ & $0.251_{-0.000}^{+0.000}$ & $0.083_{-0.000}^{+0.000}$ & $0.088_{-0.000}^{+0.000}$ & $0.068_{-0.001}^{+0.001}$ & $0.136_{-0.000}^{+0.000}$ & $0.087_{-0.000}^{+0.000}$ & $0.474_{-0.000}^{+0.000}$ \\ 
& $\num{4096}$ & $0.482_{-0.000}^{+0.000}$ & $0.153_{-0.000}^{+0.000}$ & $0.145_{-0.000}^{+0.000}$ & $0.113_{-0.000}^{+0.000}$ & $0.257_{-0.000}^{+0.000}$ & $0.164_{-0.000}^{+0.000}$ & $0.932_{-0.001}^{+0.001}$ \\ 
& $\num{8192}$ & $0.975_{-0.001}^{+0.001}$ & $0.308_{-0.000}^{+0.000}$ & $0.303_{-0.001}^{+0.001}$ & $0.245_{-0.001}^{+0.001}$ & $0.518_{-0.000}^{+0.000}$ & $0.324_{-0.001}^{+0.001}$ & $2.085_{-0.021}^{+0.021}$ \\ 
& $\num{16384}$ & $1.913_{-0.001}^{+0.001}$ & $0.599_{-0.001}^{+0.001}$ & $0.603_{-0.002}^{+0.002}$ & $0.535_{-0.002}^{+0.002}$ & $1.128_{-0.024}^{+0.024}$ & $0.633_{-0.000}^{+0.000}$ & $4.078_{-0.019}^{+0.019}$ \\ 
& $\num{32768}$ & $3.842_{-0.001}^{+0.001}$ & $1.512_{-0.003}^{+0.003}$ & $1.428_{-0.003}^{+0.003}$ & $1.460_{-0.003}^{+0.003}$ & $2.144_{-0.003}^{+0.003}$ & $1.557_{-0.003}^{+0.003}$ & $7.933_{-0.036}^{+0.036}$ \\ 
& $\num{65536}$ & $7.823_{-0.004}^{+0.004}$ & $3.609_{-0.014}^{+0.014}$ & $3.515_{-0.015}^{+0.015}$ & $3.691_{-0.024}^{+0.024}$ & $4.621_{-0.011}^{+0.011}$ & $3.626_{-0.012}^{+0.012}$ & $16.068_{-0.070}^{+0.070}$ \\ 
& $\num{131072}$ & $16.065_{-0.014}^{+0.014}$ & $6.440_{-0.071}^{+0.071}$ & $6.497_{-0.056}^{+0.056}$ & $6.372_{-0.035}^{+0.035}$ & $9.195_{-0.004}^{+0.004}$ & $6.425_{-0.031}^{+0.031}$ & $37.130_{-0.571}^{+0.571}$ \\ 
& $\num{262144}$ & $31.801_{-0.040}^{+0.040}$ & $11.783_{-0.051}^{+0.051}$ & $12.059_{-0.056}^{+0.056}$ & $12.809_{-0.076}^{+0.076}$ & $18.367_{-0.021}^{+0.021}$ & $12.356_{-0.026}^{+0.026}$ & $66.954_{-0.324}^{+0.324}$ \\ 
& $\num{524288}$ & $64.077_{-0.058}^{+0.058}$ & $24.092_{-0.020}^{+0.020}$ & $24.006_{-0.072}^{+0.072}$ & $23.764_{-0.168}^{+0.168}$ & $37.508_{-0.049}^{+0.049}$ & $24.838_{-0.033}^{+0.033}$ & $130.730_{-0.506}^{+0.506}$ \\ 
                        \midrule
\multirow{14}{*}{\shortstack[l]{Nearest\\ neighbor\\ with hops}}
& $\num{64}$ & $0.041_{-0.000}^{+0.000}$ & $0.023_{-0.001}^{+0.001}$ & $0.024_{-0.001}^{+0.001}$ & $0.025_{-0.001}^{+0.001}$ & $0.037_{-0.000}^{+0.000}$ & $0.018_{-0.000}^{+0.000}$ & $nan_{-nan}^{+nan}$ \\ 
& $\num{128}$ & $0.070_{-0.000}^{+0.000}$ & $0.028_{-0.001}^{+0.001}$ & $0.031_{-0.000}^{+0.000}$ & $0.030_{-0.001}^{+0.001}$ & $0.063_{-0.000}^{+0.000}$ & $0.026_{-0.001}^{+0.001}$ & $0.084_{-0.000}^{+0.000}$ \\ 
& $\num{256}$ & $0.202_{-0.039}^{+0.039}$ & $0.041_{-0.001}^{+0.001}$ & $0.050_{-0.000}^{+0.000}$ & $0.042_{-0.001}^{+0.001}$ & $0.123_{-0.000}^{+0.000}$ & $0.038_{-0.001}^{+0.001}$ & $0.497_{-0.115}^{+0.115}$ \\ 
& $\num{512}$ & $0.245_{-0.000}^{+0.000}$ & $0.069_{-0.000}^{+0.000}$ & $0.092_{-0.000}^{+0.000}$ & $0.074_{-0.000}^{+0.000}$ & $0.244_{-0.000}^{+0.000}$ & $0.065_{-0.001}^{+0.001}$ & $0.310_{-0.001}^{+0.001}$ \\ 
& $\num{1024}$ & $0.481_{-0.000}^{+0.000}$ & $0.133_{-0.001}^{+0.001}$ & $0.179_{-0.000}^{+0.000}$ & $0.143_{-0.000}^{+0.000}$ & $0.481_{-0.001}^{+0.001}$ & $0.115_{-0.001}^{+0.001}$ & $0.625_{-0.001}^{+0.001}$ \\ 
& $\num{2048}$ & $0.720_{-0.000}^{+0.000}$ & $0.170_{-0.000}^{+0.000}$ & $0.249_{-0.000}^{+0.000}$ & $0.187_{-0.000}^{+0.000}$ & $0.611_{-0.000}^{+0.000}$ & $0.145_{-0.001}^{+0.001}$ & $0.889_{-0.000}^{+0.000}$ \\ 
& $\num{4096}$ & $1.429_{-0.000}^{+0.000}$ & $0.317_{-0.001}^{+0.001}$ & $0.483_{-0.000}^{+0.000}$ & $0.798_{-0.085}^{+0.085}$ & $1.206_{-0.001}^{+0.001}$ & $0.259_{-0.001}^{+0.001}$ & $1.807_{-0.004}^{+0.004}$ \\ 
& $\num{8192}$ & $2.882_{-0.004}^{+0.004}$ & $0.798_{-0.003}^{+0.003}$ & $0.987_{-0.001}^{+0.001}$ & $0.724_{-0.000}^{+0.000}$ & $2.424_{-0.001}^{+0.001}$ & $0.633_{-0.003}^{+0.003}$ & $4.040_{-0.017}^{+0.017}$ \\ 
& $\num{16384}$ & $6.130_{-0.026}^{+0.026}$ & $1.650_{-0.008}^{+0.008}$ & $1.994_{-0.001}^{+0.001}$ & $1.512_{-0.004}^{+0.004}$ & $5.012_{-0.021}^{+0.021}$ & $1.471_{-0.003}^{+0.003}$ & $8.116_{-0.064}^{+0.064}$ \\ 
& $\num{32768}$ & $11.604_{-0.029}^{+0.029}$ & $3.663_{-0.016}^{+0.016}$ & $4.377_{-0.008}^{+0.008}$ & $3.844_{-0.072}^{+0.072}$ & $9.793_{-0.002}^{+0.002}$ & $3.596_{-0.012}^{+0.012}$ & $15.941_{-0.069}^{+0.069}$ \\ 
& $\num{65536}$ & $22.908_{-0.009}^{+0.009}$ & $7.208_{-0.020}^{+0.020}$ & $8.761_{-0.007}^{+0.007}$ & $7.005_{-0.029}^{+0.029}$ & $19.796_{-0.027}^{+0.027}$ & $7.509_{-0.056}^{+0.056}$ & $32.125_{-0.132}^{+0.132}$ \\ 
& $\num{131072}$ & $46.370_{-0.098}^{+0.098}$ & $13.307_{-0.042}^{+0.042}$ & $16.937_{-0.008}^{+0.008}$ & $13.027_{-0.021}^{+0.021}$ & $39.595_{-0.050}^{+0.050}$ & $12.850_{-0.057}^{+0.057}$ & $66.153_{-0.467}^{+0.467}$ \\ 
& $\num{262144}$ & $92.184_{-0.028}^{+0.028}$ & $27.115_{-0.082}^{+0.082}$ & $33.777_{-0.016}^{+0.016}$ & $26.176_{-0.031}^{+0.031}$ & $78.921_{-0.070}^{+0.070}$ & $25.932_{-0.115}^{+0.115}$ & $125.223_{-0.521}^{+0.521}$ \\ 
& $\num{524288}$ & $185.346_{-0.129}^{+0.129}$ & $55.925_{-0.151}^{+0.151}$ & $70.165_{-0.136}^{+0.136}$ & $53.866_{-0.063}^{+0.063}$ & $157.693_{-0.076}^{+0.076}$ & $50.761_{-0.202}^{+0.202}$ & $235.573_{-0.607}^{+0.607}$ \\ 
                        \midrule
\multirow{14}{*}{\shortstack[l]{Component}}
& $\num{64}$ & $0.015_{-0.000}^{+0.000}$ & $0.010_{-0.001}^{+0.001}$ & $0.009_{-0.001}^{+0.001}$ & $0.008_{-0.000}^{+0.000}$ & $0.014_{-0.001}^{+0.001}$ & $0.008_{-0.000}^{+0.000}$ & $0.015_{-0.000}^{+0.000}$ \\ 
& $\num{128}$ & $0.024_{-0.000}^{+0.000}$ & $0.010_{-0.000}^{+0.000}$ & $0.009_{-0.001}^{+0.001}$ & $0.008_{-0.000}^{+0.000}$ & $0.018_{-0.001}^{+0.001}$ & $0.009_{-0.000}^{+0.000}$ & $0.026_{-0.001}^{+0.001}$ \\ 
& $\num{256}$ & $0.044_{-0.000}^{+0.000}$ & $0.013_{-0.000}^{+0.000}$ & $0.011_{-0.001}^{+0.001}$ & $0.010_{-0.001}^{+0.001}$ & $0.028_{-0.000}^{+0.000}$ & $0.011_{-0.001}^{+0.001}$ & $0.043_{-0.000}^{+0.000}$ \\ 
& $\num{512}$ & $0.712_{-0.156}^{+0.156}$ & $0.021_{-0.000}^{+0.000}$ & $0.014_{-0.001}^{+0.001}$ & $0.013_{-0.001}^{+0.001}$ & $0.052_{-0.000}^{+0.000}$ & $0.016_{-0.000}^{+0.000}$ & $0.082_{-0.000}^{+0.000}$ \\ 
& $\num{1024}$ & $0.170_{-0.000}^{+0.000}$ & $0.038_{-0.000}^{+0.000}$ & $0.018_{-0.001}^{+0.001}$ & $0.019_{-0.001}^{+0.001}$ & $0.104_{-0.000}^{+0.000}$ & $0.027_{-0.000}^{+0.000}$ & $0.163_{-0.000}^{+0.000}$ \\ 
& $\num{2048}$ & $0.251_{-0.000}^{+0.000}$ & $0.053_{-0.000}^{+0.000}$ & $0.030_{-0.000}^{+0.000}$ & $0.031_{-0.001}^{+0.001}$ & $0.134_{-0.000}^{+0.000}$ & $0.037_{-0.000}^{+0.000}$ & $0.244_{-0.000}^{+0.000}$ \\ 
& $\num{4096}$ & $0.483_{-0.000}^{+0.000}$ & $0.095_{-0.000}^{+0.000}$ & $0.049_{-0.001}^{+0.001}$ & $0.052_{-0.001}^{+0.001}$ & $0.251_{-0.000}^{+0.000}$ & $0.061_{-0.000}^{+0.000}$ & $0.461_{-0.000}^{+0.000}$ \\ 
& $\num{8192}$ & $1.006_{-0.006}^{+0.006}$ & $0.198_{-0.001}^{+0.001}$ & $0.091_{-0.001}^{+0.001}$ & $0.098_{-0.002}^{+0.002}$ & $0.506_{-0.001}^{+0.001}$ & $0.113_{-0.000}^{+0.000}$ & $0.932_{-0.002}^{+0.002}$ \\ 
& $\num{16384}$ & $1.980_{-0.011}^{+0.011}$ & $0.384_{-0.002}^{+0.002}$ & $0.206_{-0.003}^{+0.003}$ & $0.213_{-0.004}^{+0.004}$ & $1.002_{-0.003}^{+0.003}$ & $0.213_{-0.001}^{+0.001}$ & $2.052_{-0.011}^{+0.011}$ \\ 
& $\num{32768}$ & $3.864_{-0.013}^{+0.013}$ & $0.689_{-0.002}^{+0.002}$ & $0.532_{-0.008}^{+0.008}$ & $0.533_{-0.009}^{+0.009}$ & $1.935_{-0.001}^{+0.001}$ & $0.543_{-0.004}^{+0.004}$ & $4.120_{-0.015}^{+0.015}$ \\ 
& $\num{65536}$ & $7.766_{-0.033}^{+0.033}$ & $1.728_{-0.002}^{+0.002}$ & $1.927_{-0.062}^{+0.062}$ & $1.556_{-0.004}^{+0.004}$ & $3.886_{-0.002}^{+0.002}$ & $1.604_{-0.007}^{+0.007}$ & $8.431_{-0.054}^{+0.054}$ \\ 
& $\num{131072}$ & $15.317_{-0.050}^{+0.050}$ & $3.473_{-0.013}^{+0.013}$ & $3.786_{-0.003}^{+0.003}$ & $3.447_{-0.012}^{+0.012}$ & $8.023_{-0.003}^{+0.003}$ & $3.583_{-0.020}^{+0.020}$ & $22.377_{-0.615}^{+0.615}$ \\ 
& $\num{262144}$ & $30.142_{-0.002}^{+0.002}$ & $7.797_{-0.042}^{+0.042}$ & $7.921_{-0.038}^{+0.038}$ & $7.564_{-0.032}^{+0.032}$ & $16.321_{-0.025}^{+0.025}$ & $7.981_{-0.031}^{+0.031}$ & $37.493_{-0.484}^{+0.484}$ \\ 
& $\num{524288}$ & $60.201_{-0.001}^{+0.001}$ & $14.378_{-0.071}^{+0.071}$ & $15.446_{-0.067}^{+0.067}$ & $14.410_{-0.059}^{+0.059}$ & $33.399_{-0.072}^{+0.072}$ & $14.671_{-0.050}^{+0.050}$ & $70.881_{-0.542}^{+0.542}$ \\ 
			\bottomrule
		\end{widetable}
		\end{scriptsize}
	\end{sidewaystable*}
	\begin{sidewaystable*}
		\begin{scriptsize}
		\caption{Time needed for an \texttt{MPI\_\-Neighbor\_\-alltoall}\xspace exchange, different $k$-neighborhoods and reordering algorithms on \emph{VSC}$4$\xspace with $N=100$ and $p=48$. 
			Experiment performed as described in~\ref{sec:throughput}. We present the mean time in \SI{}{\milli\second} and the~$95\%$ confidence interval range. }
		\begin{widetable}{\textwidth}{l  r r r r r r r r}
			\toprule
			Stencil & Size [\SI{}{\byte}] & \texttt{Blocked}\xspace & \texttt{Hyperplane}\xspace & $k$-$d$~\texttt{tree}\xspace & \texttt{Stencil Strips}\xspace & \texttt{Nodecart}\xspace & \texttt{VieM}\xspace & \texttt{Random}\xspace \\
                        \midrule
\multirow{14}{*}{\shortstack[l]{Nearest\\ neighbor}}
& $\num{64}$ & $0.028_{-0.000}^{+0.000}$ & $0.018_{-0.001}^{+0.001}$ & $0.019_{-0.001}^{+0.001}$ & $0.017_{-0.001}^{+0.001}$ & $0.021_{-0.001}^{+0.001}$ & $0.019_{-0.001}^{+0.001}$ & $0.047_{-0.000}^{+0.000}$ \\ 
& $\num{128}$ & $0.044_{-0.000}^{+0.000}$ & $0.022_{-0.001}^{+0.001}$ & $0.023_{-0.001}^{+0.001}$ & $0.020_{-0.001}^{+0.001}$ & $0.023_{-0.001}^{+0.001}$ & $0.021_{-0.001}^{+0.001}$ & $0.080_{-0.000}^{+0.000}$ \\ 
& $\num{256}$ & $0.081_{-0.000}^{+0.000}$ & $0.029_{-0.000}^{+0.000}$ & $0.033_{-0.000}^{+0.000}$ & $0.027_{-0.000}^{+0.000}$ & $0.035_{-0.000}^{+0.000}$ & $0.029_{-0.000}^{+0.000}$ & $0.151_{-0.000}^{+0.000}$ \\ 
& $\num{512}$ & $0.155_{-0.000}^{+0.000}$ & $0.051_{-0.000}^{+0.000}$ & $0.059_{-0.000}^{+0.000}$ & $0.049_{-0.000}^{+0.000}$ & $0.065_{-0.000}^{+0.000}$ & $0.054_{-0.000}^{+0.000}$ & $0.290_{-0.000}^{+0.000}$ \\ 
& $\num{1024}$ & $0.302_{-0.000}^{+0.000}$ & $0.099_{-0.001}^{+0.001}$ & $0.117_{-0.000}^{+0.000}$ & $0.094_{-0.000}^{+0.000}$ & $0.124_{-0.001}^{+0.001}$ & $0.100_{-0.000}^{+0.000}$ & $0.573_{-0.001}^{+0.001}$ \\ 
& $\num{2048}$ & $0.492_{-0.000}^{+0.000}$ & $0.149_{-0.000}^{+0.000}$ & $0.178_{-0.000}^{+0.000}$ & $0.150_{-0.000}^{+0.000}$ & $0.192_{-0.003}^{+0.003}$ & $0.160_{-0.000}^{+0.000}$ & $0.931_{-0.000}^{+0.000}$ \\ 
& $\num{4096}$ & $0.959_{-0.000}^{+0.000}$ & $0.281_{-0.000}^{+0.000}$ & $0.343_{-0.000}^{+0.000}$ & $0.287_{-0.000}^{+0.000}$ & $0.423_{-0.021}^{+0.021}$ & $0.303_{-0.000}^{+0.000}$ & $1.844_{-0.000}^{+0.000}$ \\ 
& $\num{8192}$ & $1.918_{-0.001}^{+0.001}$ & $0.737_{-0.053}^{+0.053}$ & $0.704_{-0.008}^{+0.008}$ & $0.570_{-0.001}^{+0.001}$ & $0.711_{-0.002}^{+0.002}$ & $0.611_{-0.001}^{+0.001}$ & $3.849_{-0.007}^{+0.007}$ \\ 
& $\num{16384}$ & $3.826_{-0.001}^{+0.001}$ & $1.621_{-0.108}^{+0.108}$ & $1.548_{-0.054}^{+0.054}$ & $1.140_{-0.001}^{+0.001}$ & $1.408_{-0.002}^{+0.002}$ & $1.209_{-0.000}^{+0.000}$ & $7.818_{-0.021}^{+0.021}$ \\ 
& $\num{32768}$ & $7.609_{-0.001}^{+0.001}$ & $2.466_{-0.040}^{+0.040}$ & $2.808_{-0.002}^{+0.002}$ & $2.334_{-0.001}^{+0.001}$ & $2.922_{-0.007}^{+0.007}$ & $2.485_{-0.002}^{+0.002}$ & $15.386_{-0.047}^{+0.047}$ \\ 
& $\num{65536}$ & $15.305_{-0.004}^{+0.004}$ & $4.945_{-0.005}^{+0.005}$ & $5.865_{-0.005}^{+0.005}$ & $4.940_{-0.006}^{+0.006}$ & $6.239_{-0.018}^{+0.018}$ & $5.319_{-0.006}^{+0.006}$ & $29.963_{-0.077}^{+0.077}$ \\ 
& $\num{131072}$ & $30.705_{-0.021}^{+0.021}$ & $9.766_{-0.004}^{+0.004}$ & $11.621_{-0.003}^{+0.003}$ & $9.769_{-0.010}^{+0.010}$ & $12.708_{-0.027}^{+0.027}$ & $10.404_{-0.002}^{+0.002}$ & $63.105_{-0.499}^{+0.499}$ \\ 
& $\num{262144}$ & $61.480_{-0.037}^{+0.037}$ & $19.668_{-0.009}^{+0.009}$ & $23.330_{-0.011}^{+0.011}$ & $19.870_{-0.027}^{+0.027}$ & $24.930_{-0.023}^{+0.023}$ & $21.258_{-0.029}^{+0.029}$ & $122.308_{-0.460}^{+0.460}$ \\ 
& $\num{524288}$ & $123.391_{-0.095}^{+0.095}$ & $40.357_{-0.016}^{+0.016}$ & $47.721_{-0.034}^{+0.034}$ & $40.435_{-0.027}^{+0.027}$ & $50.798_{-0.054}^{+0.054}$ & $43.377_{-0.046}^{+0.046}$ & $238.045_{-0.462}^{+0.462}$ \\ 
                        \midrule
\multirow{14}{*}{\shortstack[l]{Nearest\\ neighbor\\ with hops}}
& $\num{64}$ & $0.070_{-0.000}^{+0.000}$ & $0.039_{-0.002}^{+0.002}$ & $0.034_{-0.001}^{+0.001}$ & $0.029_{-0.001}^{+0.001}$ & $0.052_{-0.000}^{+0.000}$ & $0.029_{-0.001}^{+0.001}$ & $0.089_{-0.000}^{+0.000}$ \\ 
& $\num{128}$ & $0.124_{-0.000}^{+0.000}$ & $0.049_{-0.001}^{+0.001}$ & $0.050_{-0.000}^{+0.000}$ & $0.044_{-0.000}^{+0.000}$ & $0.089_{-0.000}^{+0.000}$ & $0.039_{-0.000}^{+0.000}$ & $0.159_{-0.000}^{+0.000}$ \\ 
& $\num{256}$ & $0.232_{-0.000}^{+0.000}$ & $0.077_{-0.004}^{+0.004}$ & $0.089_{-0.000}^{+0.000}$ & $0.080_{-0.000}^{+0.000}$ & $0.164_{-0.000}^{+0.000}$ & $0.066_{-0.000}^{+0.000}$ & $0.299_{-0.000}^{+0.000}$ \\ 
& $\num{512}$ & $0.454_{-0.000}^{+0.000}$ & $0.188_{-0.021}^{+0.021}$ & $0.166_{-0.000}^{+0.000}$ & $0.154_{-0.000}^{+0.000}$ & $0.382_{-0.021}^{+0.021}$ & $0.127_{-0.000}^{+0.000}$ & $0.577_{-0.000}^{+0.000}$ \\ 
& $\num{1024}$ & $0.898_{-0.001}^{+0.001}$ & $0.336_{-0.024}^{+0.024}$ & $0.372_{-0.014}^{+0.014}$ & $0.298_{-0.001}^{+0.001}$ & $0.670_{-0.014}^{+0.014}$ & $0.253_{-0.000}^{+0.000}$ & $1.143_{-0.001}^{+0.001}$ \\ 
& $\num{2048}$ & $1.428_{-0.000}^{+0.000}$ & $0.464_{-0.029}^{+0.029}$ & $0.502_{-0.000}^{+0.000}$ & $0.466_{-0.009}^{+0.009}$ & $1.012_{-0.016}^{+0.016}$ & $0.362_{-0.002}^{+0.002}$ & $1.850_{-0.001}^{+0.001}$ \\ 
& $\num{4096}$ & $2.843_{-0.001}^{+0.001}$ & $0.650_{-0.015}^{+0.015}$ & $1.422_{-0.105}^{+0.105}$ & $0.863_{-0.000}^{+0.000}$ & $1.948_{-0.002}^{+0.002}$ & $0.695_{-0.000}^{+0.000}$ & $3.901_{-0.010}^{+0.010}$ \\ 
& $\num{8192}$ & $5.846_{-0.017}^{+0.017}$ & $1.224_{-0.004}^{+0.004}$ & $2.006_{-0.001}^{+0.001}$ & $1.771_{-0.001}^{+0.001}$ & $3.932_{-0.001}^{+0.001}$ & $1.405_{-0.001}^{+0.001}$ & $8.167_{-0.027}^{+0.027}$ \\ 
& $\num{16384}$ & $11.356_{-0.010}^{+0.010}$ & $2.515_{-0.009}^{+0.009}$ & $4.046_{-0.002}^{+0.002}$ & $3.586_{-0.001}^{+0.001}$ & $7.801_{-0.001}^{+0.001}$ & $2.914_{-0.009}^{+0.009}$ & $16.309_{-0.044}^{+0.044}$ \\ 
& $\num{32768}$ & $22.690_{-0.001}^{+0.001}$ & $5.166_{-0.010}^{+0.010}$ & $8.454_{-0.011}^{+0.011}$ & $7.436_{-0.003}^{+0.003}$ & $15.820_{-0.004}^{+0.004}$ & $6.210_{-0.013}^{+0.013}$ & $31.562_{-0.072}^{+0.072}$ \\ 
& $\num{65536}$ & $45.306_{-0.003}^{+0.003}$ & $10.102_{-0.017}^{+0.017}$ & $16.869_{-0.021}^{+0.021}$ & $14.931_{-0.005}^{+0.005}$ & $31.746_{-0.009}^{+0.009}$ & $12.376_{-0.024}^{+0.024}$ & $60.201_{-0.140}^{+0.140}$ \\ 
& $\num{131072}$ & $90.735_{-0.033}^{+0.033}$ & $19.632_{-0.048}^{+0.048}$ & $33.191_{-0.013}^{+0.013}$ & $29.440_{-0.003}^{+0.003}$ & $63.291_{-0.023}^{+0.023}$ & $24.021_{-0.046}^{+0.046}$ & $121.179_{-0.316}^{+0.316}$ \\ 
& $\num{262144}$ & $181.396_{-0.044}^{+0.044}$ & $39.677_{-0.159}^{+0.159}$ & $66.514_{-0.021}^{+0.021}$ & $59.531_{-0.040}^{+0.040}$ & $126.631_{-0.033}^{+0.033}$ & $48.108_{-0.067}^{+0.067}$ & $239.648_{-0.414}^{+0.414}$ \\ 
& $\num{524288}$ & $362.859_{-0.089}^{+0.089}$ & $79.769_{-0.129}^{+0.129}$ & $134.264_{-0.047}^{+0.047}$ & $120.062_{-0.047}^{+0.047}$ & $253.462_{-0.078}^{+0.078}$ & $97.200_{-0.070}^{+0.070}$ & $471.859_{-0.739}^{+0.739}$ \\ 
                        \midrule
\multirow{14}{*}{\shortstack[l]{Component}}
& $\num{64}$ & $0.024_{-0.000}^{+0.000}$ & $0.014_{-0.001}^{+0.001}$ & $0.014_{-0.001}^{+0.001}$ & $0.018_{-0.001}^{+0.001}$ & $0.016_{-0.001}^{+0.001}$ & $0.013_{-0.001}^{+0.001}$ & $0.026_{-0.000}^{+0.000}$ \\ 
& $\num{128}$ & $0.041_{-0.000}^{+0.000}$ & $0.014_{-0.001}^{+0.001}$ & $0.014_{-0.001}^{+0.001}$ & $0.022_{-0.001}^{+0.001}$ & $0.021_{-0.001}^{+0.001}$ & $0.014_{-0.001}^{+0.001}$ & $0.044_{-0.000}^{+0.000}$ \\ 
& $\num{256}$ & $0.078_{-0.000}^{+0.000}$ & $0.019_{-0.001}^{+0.001}$ & $0.014_{-0.001}^{+0.001}$ & $0.023_{-0.001}^{+0.001}$ & $0.035_{-0.001}^{+0.001}$ & $0.015_{-0.001}^{+0.001}$ & $0.080_{-0.000}^{+0.000}$ \\ 
& $\num{512}$ & $0.152_{-0.000}^{+0.000}$ & $0.027_{-0.000}^{+0.000}$ & $0.017_{-0.001}^{+0.001}$ & $0.023_{-0.001}^{+0.001}$ & $0.055_{-0.000}^{+0.000}$ & $0.017_{-0.001}^{+0.001}$ & $0.151_{-0.000}^{+0.000}$ \\ 
& $\num{1024}$ & $0.301_{-0.000}^{+0.000}$ & $0.048_{-0.000}^{+0.000}$ & $0.023_{-0.001}^{+0.001}$ & $0.027_{-0.001}^{+0.001}$ & $0.107_{-0.000}^{+0.000}$ & $0.021_{-0.001}^{+0.001}$ & $0.300_{-0.000}^{+0.000}$ \\ 
& $\num{2048}$ & $0.482_{-0.000}^{+0.000}$ & $0.066_{-0.000}^{+0.000}$ & $0.036_{-0.001}^{+0.001}$ & $0.045_{-0.001}^{+0.001}$ & $0.168_{-0.000}^{+0.000}$ & $0.034_{-0.001}^{+0.001}$ & $0.478_{-0.000}^{+0.000}$ \\ 
& $\num{4096}$ & $0.940_{-0.000}^{+0.000}$ & $0.106_{-0.000}^{+0.000}$ & $0.059_{-0.001}^{+0.001}$ & $0.069_{-0.001}^{+0.001}$ & $0.326_{-0.000}^{+0.000}$ & $0.054_{-0.001}^{+0.001}$ & $0.933_{-0.000}^{+0.000}$ \\ 
& $\num{8192}$ & $1.900_{-0.002}^{+0.002}$ & $0.206_{-0.001}^{+0.001}$ & $0.115_{-0.002}^{+0.002}$ & $0.136_{-0.005}^{+0.005}$ & $0.659_{-0.003}^{+0.003}$ & $0.103_{-0.002}^{+0.002}$ & $1.878_{-0.001}^{+0.001}$ \\ 
& $\num{16384}$ & $3.796_{-0.006}^{+0.006}$ & $0.464_{-0.005}^{+0.005}$ & $0.250_{-0.004}^{+0.004}$ & $0.262_{-0.005}^{+0.005}$ & $1.277_{-0.001}^{+0.001}$ & $0.226_{-0.004}^{+0.004}$ & $3.760_{-0.006}^{+0.006}$ \\ 
& $\num{32768}$ & $7.526_{-0.001}^{+0.001}$ & $0.974_{-0.012}^{+0.012}$ & $0.646_{-0.031}^{+0.031}$ & $0.609_{-0.020}^{+0.020}$ & $2.545_{-0.001}^{+0.001}$ & $0.569_{-0.009}^{+0.009}$ & $7.733_{-0.029}^{+0.029}$ \\ 
& $\num{65536}$ & $14.918_{-0.001}^{+0.001}$ & $2.207_{-0.074}^{+0.074}$ & $2.254_{-0.112}^{+0.112}$ & $1.586_{-0.034}^{+0.034}$ & $5.101_{-0.001}^{+0.001}$ & $2.540_{-0.081}^{+0.081}$ & $15.220_{-0.046}^{+0.046}$ \\ 
& $\num{131072}$ & $29.789_{-0.001}^{+0.001}$ & $3.904_{-0.032}^{+0.032}$ & $3.925_{-0.073}^{+0.073}$ & $3.555_{-0.018}^{+0.018}$ & $10.443_{-0.001}^{+0.001}$ & $4.061_{-0.081}^{+0.081}$ & $39.689_{-0.413}^{+0.413}$ \\ 
& $\num{262144}$ & $59.512_{-0.001}^{+0.001}$ & $8.009_{-0.048}^{+0.048}$ & $7.286_{-0.037}^{+0.037}$ & $7.068_{-0.019}^{+0.019}$ & $20.962_{-0.002}^{+0.002}$ & $7.828_{-0.042}^{+0.042}$ & $66.571_{-0.433}^{+0.433}$ \\ 
& $\num{524288}$ & $118.942_{-0.001}^{+0.001}$ & $16.178_{-0.051}^{+0.051}$ & $14.235_{-0.046}^{+0.046}$ & $14.276_{-0.043}^{+0.043}$ & $41.908_{-0.003}^{+0.003}$ & $14.502_{-0.042}^{+0.042}$ & $128.570_{-0.299}^{+0.299}$ \\ 
			\bottomrule
		\end{widetable}
		\end{scriptsize}
	\end{sidewaystable*}
	\begin{sidewaystable*}
		\begin{scriptsize}
		\caption{Time needed for an \texttt{MPI\_\-Neighbor\_\-alltoall}\xspace exchange, different $k$-neighborhoods and reordering algorithms on \emph{SuperMUC-NG}\xspace with $N=50$ and $p=48$. 
			Experiment performed as described in~\ref{sec:throughput}. We present the mean time in \SI{}{\milli\second} and the~$95\%$ confidence interval range. }
		\begin{widetable}{\textwidth}{l  r r r r r r r r}
			\toprule
			Stencil & Size [\SI{}{\byte}] & \texttt{Blocked}\xspace & \texttt{Hyperplane}\xspace & $k$-$d$~\texttt{tree}\xspace & \texttt{Stencil Strips}\xspace & \texttt{Nodecart}\xspace & \texttt{VieM}\xspace & \texttt{Random}\xspace \\
                        \midrule
\multirow{14}{*}{\shortstack[l]{Nearest\\ neighbor}}
& $\num{64}$ & $0.022_{-0.001}^{+0.001}$ & $0.012_{-0.000}^{+0.000}$ & $0.044_{-0.002}^{+0.002}$ & $0.020_{-0.001}^{+0.001}$ & $0.028_{-0.001}^{+0.001}$ & $0.041_{-0.003}^{+0.003}$ & $0.028_{-0.002}^{+0.002}$ \\ 
& $\num{128}$ & $0.023_{-0.000}^{+0.000}$ & $0.019_{-0.001}^{+0.001}$ & $0.039_{-0.002}^{+0.002}$ & $0.026_{-0.001}^{+0.001}$ & $0.033_{-0.002}^{+0.002}$ & $0.048_{-0.002}^{+0.002}$ & $0.044_{-0.002}^{+0.002}$ \\ 
& $\num{256}$ & $0.036_{-0.000}^{+0.000}$ & $0.019_{-0.000}^{+0.000}$ & $0.038_{-0.002}^{+0.002}$ & $0.031_{-0.002}^{+0.002}$ & $0.044_{-0.002}^{+0.002}$ & $0.039_{-0.003}^{+0.003}$ & $0.067_{-0.001}^{+0.001}$ \\ 
& $\num{512}$ & $0.063_{-0.000}^{+0.000}$ & $0.035_{-0.000}^{+0.000}$ & $0.052_{-0.002}^{+0.002}$ & $0.044_{-0.002}^{+0.002}$ & $0.047_{-0.002}^{+0.002}$ & $0.056_{-0.002}^{+0.002}$ & $0.113_{-0.000}^{+0.000}$ \\ 
& $\num{1024}$ & $0.126_{-0.000}^{+0.000}$ & $0.067_{-0.000}^{+0.000}$ & $0.087_{-0.002}^{+0.002}$ & $0.062_{-0.002}^{+0.002}$ & $0.089_{-0.001}^{+0.001}$ & $0.071_{-0.001}^{+0.001}$ & $0.223_{-0.001}^{+0.001}$ \\ 
& $\num{2048}$ & $0.171_{-0.001}^{+0.001}$ & $0.097_{-0.001}^{+0.001}$ & $0.122_{-0.003}^{+0.003}$ & $0.078_{-0.001}^{+0.001}$ & $0.111_{-0.001}^{+0.001}$ & $0.101_{-0.002}^{+0.002}$ & $0.298_{-0.001}^{+0.001}$ \\ 
& $\num{4096}$ & $0.358_{-0.003}^{+0.003}$ & $0.179_{-0.003}^{+0.003}$ & $0.208_{-0.006}^{+0.006}$ & $0.132_{-0.001}^{+0.001}$ & $0.233_{-0.001}^{+0.001}$ & $0.193_{-0.006}^{+0.006}$ & $0.590_{-0.001}^{+0.001}$ \\ 
& $\num{8192}$ & $0.831_{-0.008}^{+0.008}$ & $0.425_{-0.006}^{+0.006}$ & $0.370_{-0.002}^{+0.002}$ & $0.290_{-0.001}^{+0.001}$ & $0.482_{-0.002}^{+0.002}$ & $0.394_{-0.006}^{+0.006}$ & $1.647_{-0.010}^{+0.010}$ \\ 
& $\num{16384}$ & $1.565_{-0.011}^{+0.011}$ & $0.824_{-0.017}^{+0.017}$ & $0.692_{-0.003}^{+0.003}$ & $0.548_{-0.003}^{+0.003}$ & $1.002_{-0.004}^{+0.004}$ & $0.788_{-0.018}^{+0.018}$ & $3.371_{-0.021}^{+0.021}$ \\ 
& $\num{32768}$ & $3.430_{-0.034}^{+0.034}$ & $1.559_{-0.006}^{+0.006}$ & $1.466_{-0.004}^{+0.004}$ & $1.303_{-0.005}^{+0.005}$ & $2.008_{-0.009}^{+0.009}$ & $1.437_{-0.018}^{+0.018}$ & $6.388_{-0.053}^{+0.053}$ \\ 
& $\num{65536}$ & $7.233_{-0.043}^{+0.043}$ & $3.406_{-0.019}^{+0.019}$ & $3.207_{-0.011}^{+0.011}$ & $2.944_{-0.011}^{+0.011}$ & $4.188_{-0.018}^{+0.018}$ & $3.022_{-0.013}^{+0.013}$ & $11.969_{-0.067}^{+0.067}$ \\ 
& $\num{131072}$ & $13.057_{-0.034}^{+0.034}$ & $6.401_{-0.022}^{+0.022}$ & $5.997_{-0.026}^{+0.026}$ & $5.422_{-0.003}^{+0.003}$ & $8.184_{-0.042}^{+0.042}$ & $6.364_{-0.004}^{+0.004}$ & $24.770_{-0.199}^{+0.199}$ \\ 
& $\num{262144}$ & $27.503_{-0.146}^{+0.146}$ & $13.433_{-0.034}^{+0.034}$ & $12.525_{-0.036}^{+0.036}$ & $11.012_{-0.011}^{+0.011}$ & $15.631_{-0.064}^{+0.064}$ & $12.883_{-0.007}^{+0.007}$ & $61.167_{-0.286}^{+0.286}$ \\ 
& $\num{524288}$ & $56.395_{-0.548}^{+0.548}$ & $28.185_{-0.070}^{+0.070}$ & $25.727_{-0.069}^{+0.069}$ & $22.358_{-0.021}^{+0.021}$ & $32.738_{-0.105}^{+0.105}$ & $26.496_{-0.018}^{+0.018}$ & $145.241_{-0.686}^{+0.686}$ \\ 
                        \midrule
\multirow{14}{*}{\shortstack[l]{Nearest\\ neighbor\\ with hops}}
& $\num{64}$ & $0.068_{-0.004}^{+0.004}$ & $0.062_{-0.006}^{+0.006}$ & $0.063_{-0.008}^{+0.008}$ & $0.072_{-0.008}^{+0.008}$ & $0.031_{-0.001}^{+0.001}$ & $0.045_{-0.003}^{+0.003}$ & $0.145_{-0.013}^{+0.013}$ \\ 
& $\num{128}$ & $0.074_{-0.003}^{+0.003}$ & $0.046_{-0.004}^{+0.004}$ & $0.059_{-0.006}^{+0.006}$ & $0.099_{-0.010}^{+0.010}$ & $0.051_{-0.001}^{+0.001}$ & $0.031_{-0.002}^{+0.002}$ & $0.095_{-0.005}^{+0.005}$ \\ 
& $\num{256}$ & $0.113_{-0.004}^{+0.004}$ & $0.055_{-0.004}^{+0.004}$ & $0.076_{-0.006}^{+0.006}$ & $0.078_{-0.007}^{+0.007}$ & $0.075_{-0.000}^{+0.000}$ & $0.070_{-0.004}^{+0.004}$ & $0.183_{-0.008}^{+0.008}$ \\ 
& $\num{512}$ & $0.202_{-0.004}^{+0.004}$ & $0.100_{-0.006}^{+0.006}$ & $0.102_{-0.006}^{+0.006}$ & $0.107_{-0.006}^{+0.006}$ & $0.146_{-0.001}^{+0.001}$ & $0.072_{-0.003}^{+0.003}$ & $0.270_{-0.009}^{+0.009}$ \\ 
& $\num{1024}$ & $0.402_{-0.004}^{+0.004}$ & $0.157_{-0.004}^{+0.004}$ & $0.202_{-0.010}^{+0.010}$ & $0.177_{-0.008}^{+0.008}$ & $0.308_{-0.000}^{+0.000}$ & $0.146_{-0.005}^{+0.005}$ & $0.497_{-0.007}^{+0.007}$ \\ 
& $\num{2048}$ & $0.582_{-0.008}^{+0.008}$ & $0.215_{-0.007}^{+0.007}$ & $0.263_{-0.015}^{+0.015}$ & $0.220_{-0.009}^{+0.009}$ & $0.461_{-0.001}^{+0.001}$ & $0.176_{-0.005}^{+0.005}$ & $0.757_{-0.024}^{+0.024}$ \\ 
& $\num{4096}$ & $1.126_{-0.004}^{+0.004}$ & $0.369_{-0.004}^{+0.004}$ & $0.461_{-0.016}^{+0.016}$ & $0.370_{-0.010}^{+0.010}$ & $0.921_{-0.001}^{+0.001}$ & $0.341_{-0.009}^{+0.009}$ & $1.359_{-0.033}^{+0.033}$ \\ 
& $\num{8192}$ & $2.904_{-0.026}^{+0.026}$ & $0.859_{-0.006}^{+0.006}$ & $0.979_{-0.029}^{+0.029}$ & $0.828_{-0.015}^{+0.015}$ & $2.064_{-0.021}^{+0.021}$ & $0.724_{-0.007}^{+0.007}$ & $3.461_{-0.035}^{+0.035}$ \\ 
& $\num{16384}$ & $5.422_{-0.037}^{+0.037}$ & $1.801_{-0.011}^{+0.011}$ & $1.928_{-0.041}^{+0.041}$ & $1.579_{-0.020}^{+0.020}$ & $4.005_{-0.019}^{+0.019}$ & $1.478_{-0.012}^{+0.012}$ & $7.076_{-0.071}^{+0.071}$ \\ 
& $\num{32768}$ & $10.810_{-0.056}^{+0.056}$ & $3.861_{-0.016}^{+0.016}$ & $3.947_{-0.064}^{+0.064}$ & $3.319_{-0.039}^{+0.039}$ & $7.907_{-0.030}^{+0.030}$ & $3.214_{-0.015}^{+0.015}$ & $13.531_{-0.129}^{+0.129}$ \\ 
& $\num{65536}$ & $20.723_{-0.127}^{+0.127}$ & $7.899_{-0.026}^{+0.026}$ & $7.565_{-0.080}^{+0.080}$ & $6.644_{-0.055}^{+0.055}$ & $15.594_{-0.075}^{+0.075}$ & $6.507_{-0.032}^{+0.032}$ & $26.569_{-0.197}^{+0.197}$ \\ 
& $\num{131072}$ & $40.657_{-0.192}^{+0.192}$ & $15.777_{-0.036}^{+0.036}$ & $16.580_{-0.133}^{+0.133}$ & $13.232_{-0.053}^{+0.053}$ & $31.244_{-0.090}^{+0.090}$ & $13.174_{-0.037}^{+0.037}$ & $66.695_{-0.450}^{+0.450}$ \\ 
& $\num{262144}$ & $84.403_{-0.339}^{+0.339}$ & $31.937_{-0.064}^{+0.064}$ & $35.010_{-0.408}^{+0.408}$ & $27.330_{-0.121}^{+0.121}$ & $66.675_{-0.257}^{+0.257}$ & $26.812_{-0.094}^{+0.094}$ & $158.638_{-0.791}^{+0.791}$ \\ 
& $\num{524288}$ & $174.863_{-0.893}^{+0.893}$ & $63.737_{-0.141}^{+0.141}$ & $69.544_{-0.672}^{+0.672}$ & $55.509_{-0.154}^{+0.154}$ & $139.567_{-0.640}^{+0.640}$ & $55.039_{-0.260}^{+0.260}$ & $321.429_{-1.658}^{+1.658}$ \\ 
                        \midrule
\multirow{14}{*}{\shortstack[l]{Component}}
& $\num{64}$ & $0.021_{-0.002}^{+0.002}$ & $0.024_{-0.001}^{+0.001}$ & $0.013_{-0.001}^{+0.001}$ & $0.018_{-0.001}^{+0.001}$ & $0.011_{-0.001}^{+0.001}$ & $0.014_{-0.001}^{+0.001}$ & $0.028_{-0.003}^{+0.003}$ \\ 
& $\num{128}$ & $0.027_{-0.001}^{+0.001}$ & $0.015_{-0.001}^{+0.001}$ & $0.011_{-0.001}^{+0.001}$ & $0.015_{-0.001}^{+0.001}$ & $0.019_{-0.001}^{+0.001}$ & $0.014_{-0.001}^{+0.001}$ & $0.036_{-0.002}^{+0.002}$ \\ 
& $\num{256}$ & $0.040_{-0.001}^{+0.001}$ & $0.020_{-0.001}^{+0.001}$ & $0.009_{-0.000}^{+0.000}$ & $0.022_{-0.001}^{+0.001}$ & $0.026_{-0.001}^{+0.001}$ & $0.016_{-0.001}^{+0.001}$ & $0.065_{-0.004}^{+0.004}$ \\ 
& $\num{512}$ & $0.066_{-0.001}^{+0.001}$ & $0.019_{-0.001}^{+0.001}$ & $0.013_{-0.000}^{+0.000}$ & $0.028_{-0.002}^{+0.002}$ & $0.044_{-0.001}^{+0.001}$ & $0.018_{-0.001}^{+0.001}$ & $0.088_{-0.004}^{+0.004}$ \\ 
& $\num{1024}$ & $0.128_{-0.000}^{+0.000}$ & $0.047_{-0.002}^{+0.002}$ & $0.024_{-0.001}^{+0.001}$ & $0.029_{-0.002}^{+0.002}$ & $0.072_{-0.001}^{+0.001}$ & $0.028_{-0.000}^{+0.000}$ & $0.153_{-0.004}^{+0.004}$ \\ 
& $\num{2048}$ & $0.186_{-0.000}^{+0.000}$ & $0.054_{-0.001}^{+0.001}$ & $0.047_{-0.001}^{+0.001}$ & $0.055_{-0.002}^{+0.002}$ & $0.097_{-0.000}^{+0.000}$ & $0.047_{-0.001}^{+0.001}$ & $0.178_{-0.004}^{+0.004}$ \\ 
& $\num{4096}$ & $0.363_{-0.000}^{+0.000}$ & $0.100_{-0.002}^{+0.002}$ & $0.089_{-0.002}^{+0.002}$ & $0.096_{-0.002}^{+0.002}$ & $0.163_{-0.001}^{+0.001}$ & $0.092_{-0.002}^{+0.002}$ & $0.308_{-0.002}^{+0.002}$ \\ 
& $\num{8192}$ & $0.888_{-0.009}^{+0.009}$ & $0.225_{-0.003}^{+0.003}$ & $0.171_{-0.004}^{+0.004}$ & $0.173_{-0.004}^{+0.004}$ & $0.347_{-0.002}^{+0.002}$ & $0.174_{-0.004}^{+0.004}$ & $0.776_{-0.007}^{+0.007}$ \\ 
& $\num{16384}$ & $1.685_{-0.018}^{+0.018}$ & $0.435_{-0.008}^{+0.008}$ & $0.338_{-0.009}^{+0.009}$ & $0.334_{-0.009}^{+0.009}$ & $0.693_{-0.004}^{+0.004}$ & $0.335_{-0.009}^{+0.009}$ & $1.648_{-0.012}^{+0.012}$ \\ 
& $\num{32768}$ & $3.063_{-0.027}^{+0.027}$ & $0.756_{-0.017}^{+0.017}$ & $0.679_{-0.018}^{+0.018}$ & $0.679_{-0.018}^{+0.018}$ & $1.397_{-0.007}^{+0.007}$ & $0.683_{-0.019}^{+0.019}$ & $3.243_{-0.026}^{+0.026}$ \\ 
& $\num{65536}$ & $6.553_{-0.109}^{+0.109}$ & $1.444_{-0.027}^{+0.027}$ & $1.294_{-0.012}^{+0.012}$ & $1.292_{-0.012}^{+0.012}$ & $3.018_{-0.018}^{+0.018}$ & $1.326_{-0.017}^{+0.017}$ & $6.367_{-0.052}^{+0.052}$ \\ 
& $\num{131072}$ & $13.257_{-0.218}^{+0.218}$ & $3.193_{-0.004}^{+0.004}$ & $3.193_{-0.002}^{+0.002}$ & $3.190_{-0.002}^{+0.002}$ & $6.696_{-0.031}^{+0.031}$ & $3.173_{-0.002}^{+0.002}$ & $12.685_{-0.101}^{+0.101}$ \\ 
& $\num{262144}$ & $25.316_{-0.405}^{+0.405}$ & $6.685_{-0.010}^{+0.010}$ & $6.721_{-0.003}^{+0.003}$ & $6.681_{-0.003}^{+0.003}$ & $13.032_{-0.054}^{+0.054}$ & $6.710_{-0.002}^{+0.002}$ & $25.332_{-0.194}^{+0.194}$ \\ 
& $\num{524288}$ & $53.532_{-0.330}^{+0.330}$ & $14.496_{-0.006}^{+0.006}$ & $14.741_{-0.014}^{+0.014}$ & $14.659_{-0.016}^{+0.016}$ & $24.524_{-0.070}^{+0.070}$ & $14.677_{-0.020}^{+0.020}$ & $72.583_{-0.455}^{+0.455}$ \\ 
			\bottomrule
		\end{widetable}
		\end{scriptsize}
	\end{sidewaystable*}
	\begin{sidewaystable*}
		\begin{scriptsize}
		\caption{Time needed for an \texttt{MPI\_\-Neighbor\_\-alltoall}\xspace exchange, different $k$-neighborhoods and reordering algorithms on \emph{SuperMUC-NG}\xspace with $N=100$ and $p=48$. 
			Experiment performed as described in~\ref{sec:throughput}. We present the mean time in \SI{}{\milli\second} and the~$95\%$ confidence interval range. }
		\begin{widetable}{\textwidth}{l  r r r r r r r r}
			\toprule
			Stencil & Size [\SI{}{\byte}] & \texttt{Blocked}\xspace & \texttt{Hyperplane}\xspace & $k$-$d$~\texttt{tree}\xspace & \texttt{Stencil Strips}\xspace & \texttt{Nodecart}\xspace & \texttt{VieM}\xspace & \texttt{Random}\xspace \\
                        \midrule
\multirow{14}{*}{\shortstack[l]{Nearest\\ neighbor}}
& $\num{64}$ & $0.039_{-0.002}^{+0.002}$ & $0.050_{-0.003}^{+0.003}$ & $0.038_{-0.002}^{+0.002}$ & $0.024_{-0.001}^{+0.001}$ & $0.022_{-0.001}^{+0.001}$ & $0.034_{-0.002}^{+0.002}$ & $0.053_{-0.002}^{+0.002}$ \\ 
& $\num{128}$ & $0.044_{-0.001}^{+0.001}$ & $0.055_{-0.003}^{+0.003}$ & $0.047_{-0.003}^{+0.003}$ & $0.037_{-0.002}^{+0.002}$ & $0.024_{-0.001}^{+0.001}$ & $0.025_{-0.001}^{+0.001}$ & $0.062_{-0.002}^{+0.002}$ \\ 
& $\num{256}$ & $0.051_{-0.002}^{+0.002}$ & $0.054_{-0.003}^{+0.003}$ & $0.026_{-0.001}^{+0.001}$ & $0.026_{-0.001}^{+0.001}$ & $0.035_{-0.001}^{+0.001}$ & $0.036_{-0.001}^{+0.001}$ & $0.087_{-0.002}^{+0.002}$ \\ 
& $\num{512}$ & $0.078_{-0.001}^{+0.001}$ & $0.068_{-0.003}^{+0.003}$ & $0.033_{-0.000}^{+0.000}$ & $0.051_{-0.001}^{+0.001}$ & $0.046_{-0.002}^{+0.002}$ & $0.072_{-0.004}^{+0.004}$ & $0.125_{-0.001}^{+0.001}$ \\ 
& $\num{1024}$ & $0.156_{-0.002}^{+0.002}$ & $0.076_{-0.001}^{+0.001}$ & $0.065_{-0.001}^{+0.001}$ & $0.074_{-0.001}^{+0.001}$ & $0.075_{-0.003}^{+0.003}$ & $0.071_{-0.001}^{+0.001}$ & $0.257_{-0.002}^{+0.002}$ \\ 
& $\num{2048}$ & $0.269_{-0.008}^{+0.008}$ & $0.105_{-0.001}^{+0.001}$ & $0.095_{-0.000}^{+0.000}$ & $0.105_{-0.002}^{+0.002}$ & $0.088_{-0.001}^{+0.001}$ & $0.109_{-0.002}^{+0.002}$ & $0.322_{-0.002}^{+0.002}$ \\ 
& $\num{4096}$ & $0.395_{-0.004}^{+0.004}$ & $0.186_{-0.003}^{+0.003}$ & $0.163_{-0.000}^{+0.000}$ & $0.184_{-0.004}^{+0.004}$ & $0.164_{-0.001}^{+0.001}$ & $0.186_{-0.003}^{+0.003}$ & $0.667_{-0.003}^{+0.003}$ \\ 
& $\num{8192}$ & $0.924_{-0.006}^{+0.006}$ & $0.429_{-0.006}^{+0.006}$ & $0.388_{-0.002}^{+0.002}$ & $0.391_{-0.006}^{+0.006}$ & $0.382_{-0.002}^{+0.002}$ & $0.405_{-0.006}^{+0.006}$ & $1.766_{-0.013}^{+0.013}$ \\ 
& $\num{16384}$ & $1.739_{-0.016}^{+0.016}$ & $0.852_{-0.017}^{+0.017}$ & $0.718_{-0.003}^{+0.003}$ & $0.755_{-0.017}^{+0.017}$ & $0.766_{-0.004}^{+0.004}$ & $0.803_{-0.015}^{+0.015}$ & $3.633_{-0.028}^{+0.028}$ \\ 
& $\num{32768}$ & $3.478_{-0.027}^{+0.027}$ & $1.606_{-0.005}^{+0.005}$ & $1.512_{-0.005}^{+0.005}$ & $1.424_{-0.017}^{+0.017}$ & $1.629_{-0.007}^{+0.007}$ & $1.472_{-0.016}^{+0.016}$ & $7.029_{-0.040}^{+0.040}$ \\ 
& $\num{65536}$ & $7.208_{-0.045}^{+0.045}$ & $3.563_{-0.010}^{+0.010}$ & $3.299_{-0.011}^{+0.011}$ & $3.049_{-0.015}^{+0.015}$ & $3.337_{-0.021}^{+0.021}$ & $3.080_{-0.010}^{+0.010}$ & $13.909_{-0.069}^{+0.069}$ \\ 
& $\num{131072}$ & $13.916_{-0.038}^{+0.038}$ & $6.441_{-0.019}^{+0.019}$ & $6.222_{-0.021}^{+0.021}$ & $6.419_{-0.003}^{+0.003}$ & $5.967_{-0.017}^{+0.017}$ & $6.380_{-0.003}^{+0.003}$ & $27.183_{-0.116}^{+0.116}$ \\ 
& $\num{262144}$ & $27.489_{-0.071}^{+0.071}$ & $13.216_{-0.040}^{+0.040}$ & $13.047_{-0.035}^{+0.035}$ & $12.995_{-0.006}^{+0.006}$ & $14.443_{-0.281}^{+0.281}$ & $12.923_{-0.006}^{+0.006}$ & $82.569_{-0.468}^{+0.468}$ \\ 
& $\num{524288}$ & $60.771_{-0.294}^{+0.294}$ & $26.447_{-0.077}^{+0.077}$ & $26.696_{-0.065}^{+0.065}$ & $27.207_{-0.024}^{+0.024}$ & $26.241_{-0.086}^{+0.086}$ & $26.723_{-0.023}^{+0.023}$ & $203.706_{-1.110}^{+1.110}$ \\ 
                        \midrule
\multirow{14}{*}{\shortstack[l]{Nearest\\ neighbor\\ with hops}}
& $\num{64}$ & $0.044_{-0.001}^{+0.001}$ & $0.063_{-0.004}^{+0.004}$ & $0.039_{-0.003}^{+0.003}$ & $0.052_{-0.003}^{+0.003}$ & $0.031_{-0.001}^{+0.001}$ & $0.070_{-0.009}^{+0.009}$ & $0.122_{-0.013}^{+0.013}$ \\ 
& $\num{128}$ & $0.056_{-0.001}^{+0.001}$ & $0.068_{-0.004}^{+0.004}$ & $0.044_{-0.003}^{+0.003}$ & $0.026_{-0.001}^{+0.001}$ & $0.042_{-0.001}^{+0.001}$ & $0.095_{-0.013}^{+0.013}$ & $0.121_{-0.011}^{+0.011}$ \\ 
& $\num{256}$ & $0.092_{-0.001}^{+0.001}$ & $0.064_{-0.001}^{+0.001}$ & $0.052_{-0.002}^{+0.002}$ & $0.053_{-0.003}^{+0.003}$ & $0.063_{-0.001}^{+0.001}$ & $0.134_{-0.014}^{+0.014}$ & $0.214_{-0.017}^{+0.017}$ \\ 
& $\num{512}$ & $0.175_{-0.001}^{+0.001}$ & $0.121_{-0.002}^{+0.002}$ & $0.088_{-0.002}^{+0.002}$ & $0.080_{-0.002}^{+0.002}$ & $0.121_{-0.001}^{+0.001}$ & $0.161_{-0.015}^{+0.015}$ & $0.293_{-0.014}^{+0.014}$ \\ 
& $\num{1024}$ & $0.370_{-0.001}^{+0.001}$ & $0.227_{-0.000}^{+0.000}$ & $0.157_{-0.002}^{+0.002}$ & $0.128_{-0.001}^{+0.001}$ & $0.260_{-0.000}^{+0.000}$ & $0.341_{-0.024}^{+0.024}$ & $0.516_{-0.016}^{+0.016}$ \\ 
& $\num{2048}$ & $0.517_{-0.002}^{+0.002}$ & $0.313_{-0.001}^{+0.001}$ & $0.203_{-0.002}^{+0.002}$ & $0.182_{-0.003}^{+0.003}$ & $0.367_{-0.000}^{+0.000}$ & $0.318_{-0.019}^{+0.019}$ & $0.729_{-0.025}^{+0.025}$ \\ 
& $\num{4096}$ & $1.047_{-0.004}^{+0.004}$ & $0.614_{-0.001}^{+0.001}$ & $0.407_{-0.002}^{+0.002}$ & $0.322_{-0.001}^{+0.001}$ & $0.737_{-0.001}^{+0.001}$ & $0.546_{-0.030}^{+0.030}$ & $1.543_{-0.044}^{+0.044}$ \\ 
& $\num{8192}$ & $2.720_{-0.023}^{+0.023}$ & $1.309_{-0.004}^{+0.004}$ & $0.907_{-0.003}^{+0.003}$ & $0.778_{-0.007}^{+0.007}$ & $1.618_{-0.019}^{+0.019}$ & $1.053_{-0.044}^{+0.044}$ & $4.039_{-0.072}^{+0.072}$ \\ 
& $\num{16384}$ & $5.158_{-0.032}^{+0.032}$ & $2.755_{-0.005}^{+0.005}$ & $1.813_{-0.007}^{+0.007}$ & $1.559_{-0.005}^{+0.005}$ & $3.038_{-0.010}^{+0.010}$ & $1.861_{-0.063}^{+0.063}$ & $8.099_{-0.119}^{+0.119}$ \\ 
& $\num{32768}$ & $10.596_{-0.042}^{+0.042}$ & $5.506_{-0.006}^{+0.006}$ & $3.781_{-0.018}^{+0.018}$ & $3.293_{-0.010}^{+0.010}$ & $6.368_{-0.022}^{+0.022}$ & $3.972_{-0.127}^{+0.127}$ & $13.444_{-0.096}^{+0.096}$ \\ 
& $\num{65536}$ & $19.830_{-0.069}^{+0.069}$ & $10.874_{-0.015}^{+0.015}$ & $7.406_{-0.016}^{+0.016}$ & $6.636_{-0.058}^{+0.058}$ & $12.439_{-0.035}^{+0.035}$ & $7.956_{-0.246}^{+0.246}$ & $27.049_{-0.142}^{+0.142}$ \\ 
& $\num{131072}$ & $39.773_{-0.108}^{+0.108}$ & $21.305_{-0.024}^{+0.024}$ & $15.646_{-0.021}^{+0.021}$ & $13.346_{-0.020}^{+0.020}$ & $25.155_{-0.065}^{+0.065}$ & $15.952_{-0.345}^{+0.345}$ & $74.959_{-0.406}^{+0.406}$ \\ 
& $\num{262144}$ & $88.509_{-0.599}^{+0.599}$ & $46.045_{-0.085}^{+0.085}$ & $31.922_{-0.052}^{+0.052}$ & $26.989_{-0.034}^{+0.034}$ & $50.261_{-0.093}^{+0.093}$ & $31.863_{-0.554}^{+0.554}$ & $183.225_{-0.707}^{+0.707}$ \\ 
& $\num{524288}$ & $197.841_{-2.145}^{+2.145}$ & $96.257_{-0.213}^{+0.213}$ & $65.305_{-0.137}^{+0.137}$ & $55.430_{-0.087}^{+0.087}$ & $104.417_{-0.217}^{+0.217}$ & $62.788_{-0.923}^{+0.923}$ & $371.839_{-1.489}^{+1.489}$ \\ 
                        \midrule
\multirow{14}{*}{\shortstack[l]{Component}}
& $\num{64}$ & $0.019_{-0.001}^{+0.001}$ & $0.125_{-0.011}^{+0.011}$ & $0.065_{-0.007}^{+0.007}$ & $0.013_{-0.001}^{+0.001}$ & $0.090_{-0.011}^{+0.011}$ & $0.049_{-0.002}^{+0.002}$ & $0.080_{-0.004}^{+0.004}$ \\ 
& $\num{128}$ & $0.026_{-0.001}^{+0.001}$ & $0.138_{-0.012}^{+0.012}$ & $0.110_{-0.011}^{+0.011}$ & $0.020_{-0.001}^{+0.001}$ & $0.070_{-0.006}^{+0.006}$ & $0.037_{-0.003}^{+0.003}$ & $0.078_{-0.003}^{+0.003}$ \\ 
& $\num{256}$ & $0.042_{-0.001}^{+0.001}$ & $0.143_{-0.012}^{+0.012}$ & $0.102_{-0.010}^{+0.010}$ & $0.035_{-0.002}^{+0.002}$ & $0.081_{-0.008}^{+0.008}$ & $0.026_{-0.002}^{+0.002}$ & $0.086_{-0.002}^{+0.002}$ \\ 
& $\num{512}$ & $0.069_{-0.001}^{+0.001}$ & $0.123_{-0.011}^{+0.011}$ & $0.066_{-0.008}^{+0.008}$ & $0.033_{-0.002}^{+0.002}$ & $0.093_{-0.008}^{+0.008}$ & $0.046_{-0.003}^{+0.003}$ & $0.098_{-0.003}^{+0.003}$ \\ 
& $\num{1024}$ & $0.133_{-0.002}^{+0.002}$ & $0.082_{-0.007}^{+0.007}$ & $0.084_{-0.007}^{+0.007}$ & $0.042_{-0.001}^{+0.001}$ & $0.128_{-0.010}^{+0.010}$ & $0.031_{-0.001}^{+0.001}$ & $0.153_{-0.003}^{+0.003}$ \\ 
& $\num{2048}$ & $0.204_{-0.004}^{+0.004}$ & $0.107_{-0.007}^{+0.007}$ & $0.090_{-0.007}^{+0.007}$ & $0.055_{-0.002}^{+0.002}$ & $0.166_{-0.011}^{+0.011}$ & $0.051_{-0.001}^{+0.001}$ & $0.189_{-0.002}^{+0.002}$ \\ 
& $\num{4096}$ & $0.319_{-0.001}^{+0.001}$ & $0.167_{-0.006}^{+0.006}$ & $0.136_{-0.007}^{+0.007}$ & $0.096_{-0.002}^{+0.002}$ & $0.262_{-0.013}^{+0.013}$ & $0.101_{-0.003}^{+0.003}$ & $0.327_{-0.002}^{+0.002}$ \\ 
& $\num{8192}$ & $0.872_{-0.004}^{+0.004}$ & $0.328_{-0.007}^{+0.007}$ & $0.228_{-0.009}^{+0.009}$ & $0.181_{-0.004}^{+0.004}$ & $0.372_{-0.014}^{+0.014}$ & $0.178_{-0.004}^{+0.004}$ & $0.872_{-0.008}^{+0.008}$ \\ 
& $\num{16384}$ & $1.709_{-0.007}^{+0.007}$ & $0.591_{-0.009}^{+0.009}$ & $0.359_{-0.011}^{+0.011}$ & $0.349_{-0.009}^{+0.009}$ & $0.707_{-0.020}^{+0.020}$ & $0.358_{-0.009}^{+0.009}$ & $1.777_{-0.013}^{+0.013}$ \\ 
& $\num{32768}$ & $3.456_{-0.021}^{+0.021}$ & $1.049_{-0.017}^{+0.017}$ & $0.706_{-0.019}^{+0.019}$ & $0.706_{-0.018}^{+0.018}$ & $1.198_{-0.016}^{+0.016}$ & $0.697_{-0.018}^{+0.018}$ & $3.442_{-0.023}^{+0.023}$ \\ 
& $\num{65536}$ & $6.692_{-0.048}^{+0.048}$ & $1.963_{-0.024}^{+0.024}$ & $1.326_{-0.012}^{+0.012}$ & $1.316_{-0.011}^{+0.011}$ & $2.486_{-0.037}^{+0.037}$ & $1.321_{-0.014}^{+0.014}$ & $7.076_{-0.061}^{+0.061}$ \\ 
& $\num{131072}$ & $13.309_{-0.098}^{+0.098}$ & $3.980_{-0.022}^{+0.022}$ & $3.247_{-0.002}^{+0.002}$ & $3.253_{-0.004}^{+0.004}$ & $5.058_{-0.068}^{+0.068}$ & $3.217_{-0.003}^{+0.003}$ & $13.551_{-0.071}^{+0.071}$ \\ 
& $\num{262144}$ & $25.742_{-0.186}^{+0.186}$ & $8.281_{-0.044}^{+0.044}$ & $6.821_{-0.005}^{+0.005}$ & $6.775_{-0.002}^{+0.002}$ & $10.322_{-0.125}^{+0.125}$ & $6.777_{-0.003}^{+0.003}$ & $26.792_{-0.147}^{+0.147}$ \\ 
& $\num{524288}$ & $55.323_{-0.303}^{+0.303}$ & $15.970_{-0.054}^{+0.054}$ & $14.909_{-0.006}^{+0.006}$ & $14.916_{-0.007}^{+0.007}$ & $19.572_{-0.151}^{+0.151}$ & $14.883_{-0.007}^{+0.007}$ & $78.029_{-0.472}^{+0.472}$ \\ 
			\bottomrule
		\end{widetable}
		\end{scriptsize}
	\end{sidewaystable*}
	\begin{sidewaystable*}
		\begin{scriptsize}
		\caption{Time needed for an \texttt{MPI\_\-Neighbor\_\-alltoall}\xspace exchange, different $k$-neighborhoods and reordering algorithms on \emph{JUWELS}\xspace with $N=50$ and $p=48$. 
			Experiment performed as described in~\ref{sec:throughput}. We present the mean time in \SI{}{\milli\second} and the~$95\%$ confidence interval range. }
		\begin{widetable}{\textwidth}{l  r r r r r r r r}
			\toprule
			Stencil & Size [\SI{}{\byte}] & \texttt{Blocked}\xspace & \texttt{Hyperplane}\xspace & $k$-$d$~\texttt{tree}\xspace & \texttt{Stencil Strips}\xspace & \texttt{Nodecart}\xspace & \texttt{VieM}\xspace & \texttt{Random}\xspace \\
			\midrule
\multirow{14}{*}{\shortstack[l]{Nearest\\ neighbor}}
& $\num{64}$ & $0.032_{-0.003}^{+0.003}$ & $0.059_{-0.004}^{+0.004}$ & $0.053_{-0.005}^{+0.005}$ & $0.026_{-0.002}^{+0.002}$ & $0.040_{-0.004}^{+0.004}$ & $0.070_{-0.012}^{+0.012}$ & $0.074_{-0.007}^{+0.007}$ \\ 
& $\num{128}$ & $0.078_{-0.005}^{+0.005}$ & $0.059_{-0.004}^{+0.004}$ & $0.051_{-0.005}^{+0.005}$ & $0.026_{-0.002}^{+0.002}$ & $0.041_{-0.002}^{+0.002}$ & $0.079_{-0.013}^{+0.013}$ & $0.075_{-0.004}^{+0.004}$ \\ 
& $\num{256}$ & $0.036_{-0.001}^{+0.001}$ & $0.054_{-0.005}^{+0.005}$ & $0.064_{-0.005}^{+0.005}$ & $0.024_{-0.002}^{+0.002}$ & $0.038_{-0.002}^{+0.002}$ & $0.069_{-0.011}^{+0.011}$ & $0.096_{-0.004}^{+0.004}$ \\ 
& $\num{512}$ & $0.056_{-0.001}^{+0.001}$ & $0.048_{-0.003}^{+0.003}$ & $0.063_{-0.005}^{+0.005}$ & $0.042_{-0.004}^{+0.004}$ & $0.072_{-0.005}^{+0.005}$ & $0.240_{-0.028}^{+0.028}$ & $0.151_{-0.001}^{+0.001}$ \\ 
& $\num{1024}$ & $0.824_{-0.133}^{+0.133}$ & $0.079_{-0.005}^{+0.005}$ & $0.065_{-0.003}^{+0.003}$ & $0.049_{-0.002}^{+0.002}$ & $0.100_{-0.005}^{+0.005}$ & $0.112_{-0.013}^{+0.013}$ & $0.344_{-0.002}^{+0.002}$ \\ 
& $\num{2048}$ & $0.391_{-0.029}^{+0.029}$ & $0.087_{-0.003}^{+0.003}$ & $0.108_{-0.004}^{+0.004}$ & $0.074_{-0.002}^{+0.002}$ & $0.249_{-0.016}^{+0.016}$ & $0.166_{-0.017}^{+0.017}$ & $0.735_{-0.004}^{+0.004}$ \\ 
& $\num{4096}$ & $0.807_{-0.070}^{+0.070}$ & $0.141_{-0.005}^{+0.005}$ & $0.296_{-0.021}^{+0.021}$ & $0.352_{-0.021}^{+0.021}$ & $0.332_{-0.013}^{+0.013}$ & $0.251_{-0.021}^{+0.021}$ & $3.142_{-0.274}^{+0.274}$ \\ 
& $\num{8192}$ & $1.622_{-0.134}^{+0.134}$ & $0.322_{-0.011}^{+0.011}$ & $0.358_{-0.009}^{+0.009}$ & $0.302_{-0.009}^{+0.009}$ & $0.720_{-0.019}^{+0.019}$ & $0.569_{-0.040}^{+0.040}$ & $3.212_{-0.020}^{+0.020}$ \\ 
& $\num{16384}$ & $2.924_{-0.239}^{+0.239}$ & $0.612_{-0.014}^{+0.014}$ & $0.664_{-0.012}^{+0.012}$ & $0.507_{-0.003}^{+0.003}$ & $1.465_{-0.045}^{+0.045}$ & $1.112_{-0.070}^{+0.070}$ & $7.480_{-0.391}^{+0.391}$ \\ 
& $\num{32768}$ & $5.907_{-0.506}^{+0.506}$ & $4.223_{-0.592}^{+0.592}$ & $4.037_{-0.603}^{+0.603}$ & $1.265_{-0.051}^{+0.051}$ & $2.891_{-0.074}^{+0.074}$ & $1.840_{-0.087}^{+0.087}$ & $12.742_{-0.038}^{+0.038}$ \\ 
& $\num{65536}$ & $13.384_{-0.990}^{+0.990}$ & $2.631_{-0.031}^{+0.031}$ & $2.895_{-0.039}^{+0.039}$ & $2.539_{-0.032}^{+0.032}$ & $5.537_{-0.217}^{+0.217}$ & $6.365_{-0.800}^{+0.800}$ & $25.280_{-0.076}^{+0.076}$ \\ 
& $\num{131072}$ & $24.604_{-1.907}^{+1.907}$ & $5.459_{-0.044}^{+0.044}$ & $6.244_{-0.051}^{+0.051}$ & $5.516_{-0.087}^{+0.087}$ & $11.021_{-0.320}^{+0.320}$ & $7.577_{-0.329}^{+0.329}$ & $50.665_{-0.429}^{+0.429}$ \\ 
& $\num{262144}$ & $22.814_{-0.139}^{+0.139}$ & $11.058_{-0.093}^{+0.093}$ & $12.887_{-0.092}^{+0.092}$ & $11.136_{-0.132}^{+0.132}$ & $21.487_{-0.493}^{+0.493}$ & $14.446_{-0.408}^{+0.408}$ & $102.591_{-0.854}^{+0.854}$ \\ 
& $\num{524288}$ & $45.121_{-0.212}^{+0.212}$ & $22.214_{-0.170}^{+0.170}$ & $26.384_{-0.125}^{+0.125}$ & $22.422_{-0.269}^{+0.269}$ & $41.894_{-1.041}^{+1.041}$ & $27.933_{-0.582}^{+0.582}$ & $201.109_{-1.042}^{+1.042}$ \\ 
                        \midrule
\multirow{14}{*}{\shortstack[l]{Nearest\\ neighbor\\ with hops}}
& $\num{64}$ & $0.027_{-0.001}^{+0.001}$ & $0.016_{-0.000}^{+0.000}$ & $0.017_{-0.000}^{+0.000}$ & $0.017_{-0.000}^{+0.000}$ & $0.024_{-0.000}^{+0.000}$ & $0.017_{-0.000}^{+0.000}$ & $0.044_{-0.000}^{+0.000}$ \\ 
& $\num{128}$ & $0.056_{-0.000}^{+0.000}$ & $0.023_{-0.000}^{+0.000}$ & $0.024_{-0.000}^{+0.000}$ & $0.023_{-0.000}^{+0.000}$ & $0.047_{-0.000}^{+0.000}$ & $0.021_{-0.000}^{+0.000}$ & $0.075_{-0.000}^{+0.000}$ \\ 
& $\num{256}$ & $0.093_{-0.000}^{+0.000}$ & $0.035_{-0.000}^{+0.000}$ & $0.038_{-0.000}^{+0.000}$ & $0.033_{-0.000}^{+0.000}$ & $0.080_{-0.000}^{+0.000}$ & $0.030_{-0.000}^{+0.000}$ & $0.133_{-0.000}^{+0.000}$ \\ 
& $\num{512}$ & $0.130_{-0.001}^{+0.001}$ & $0.057_{-0.000}^{+0.000}$ & $0.062_{-0.000}^{+0.000}$ & $0.058_{-0.000}^{+0.000}$ & $0.117_{-0.001}^{+0.001}$ & $0.053_{-0.000}^{+0.000}$ & $0.254_{-0.000}^{+0.000}$ \\ 
& $\num{1024}$ & $0.255_{-0.001}^{+0.001}$ & $0.094_{-0.000}^{+0.000}$ & $0.104_{-0.000}^{+0.000}$ & $0.103_{-0.001}^{+0.001}$ & $0.229_{-0.001}^{+0.001}$ & $0.093_{-0.001}^{+0.001}$ & $0.494_{-0.000}^{+0.000}$ \\ 
& $\num{2048}$ & $0.512_{-0.002}^{+0.002}$ & $0.176_{-0.000}^{+0.000}$ & $0.177_{-0.001}^{+0.001}$ & $0.172_{-0.001}^{+0.001}$ & $0.435_{-0.001}^{+0.001}$ & $0.149_{-0.000}^{+0.000}$ & $0.978_{-0.000}^{+0.000}$ \\ 
& $\num{4096}$ & $1.051_{-0.006}^{+0.006}$ & $0.474_{-0.020}^{+0.020}$ & $0.510_{-0.020}^{+0.020}$ & $0.325_{-0.001}^{+0.001}$ & $0.878_{-0.003}^{+0.003}$ & $0.290_{-0.001}^{+0.001}$ & $2.030_{-0.008}^{+0.008}$ \\ 
& $\num{8192}$ & $2.182_{-0.010}^{+0.010}$ & $0.753_{-0.005}^{+0.005}$ & $0.884_{-0.007}^{+0.007}$ & $0.811_{-0.004}^{+0.004}$ & $1.922_{-0.008}^{+0.008}$ & $0.827_{-0.004}^{+0.004}$ & $4.489_{-0.024}^{+0.024}$ \\ 
& $\num{16384}$ & $4.217_{-0.016}^{+0.016}$ & $1.678_{-0.015}^{+0.015}$ & $1.793_{-0.014}^{+0.014}$ & $1.551_{-0.008}^{+0.008}$ & $3.721_{-0.016}^{+0.016}$ & $1.605_{-0.009}^{+0.009}$ & $9.363_{-0.046}^{+0.046}$ \\ 
& $\num{32768}$ & $8.217_{-0.024}^{+0.024}$ & $3.741_{-0.021}^{+0.021}$ & $3.711_{-0.015}^{+0.015}$ & $3.140_{-0.016}^{+0.016}$ & $7.332_{-0.030}^{+0.030}$ & $3.406_{-0.014}^{+0.014}$ & $19.386_{-0.085}^{+0.085}$ \\ 
& $\num{65536}$ & $16.543_{-0.047}^{+0.047}$ & $7.495_{-0.032}^{+0.032}$ & $7.468_{-0.029}^{+0.029}$ & $5.974_{-0.030}^{+0.030}$ & $14.440_{-0.045}^{+0.045}$ & $6.663_{-0.031}^{+0.031}$ & $40.110_{-0.121}^{+0.121}$ \\ 
& $\num{131072}$ & $32.977_{-0.090}^{+0.090}$ & $14.803_{-0.067}^{+0.067}$ & $15.041_{-0.064}^{+0.064}$ & $11.718_{-0.039}^{+0.039}$ & $29.161_{-0.126}^{+0.126}$ & $13.022_{-0.049}^{+0.049}$ & $81.422_{-0.268}^{+0.268}$ \\ 
& $\num{262144}$ & $65.735_{-0.159}^{+0.159}$ & $30.113_{-0.124}^{+0.124}$ & $29.774_{-0.072}^{+0.072}$ & $23.318_{-0.110}^{+0.110}$ & $58.178_{-0.198}^{+0.198}$ & $25.848_{-0.099}^{+0.099}$ & $163.502_{-0.455}^{+0.455}$ \\ 
& $\num{524288}$ & $131.086_{-0.314}^{+0.314}$ & $60.261_{-0.152}^{+0.152}$ & $58.900_{-0.103}^{+0.103}$ & $45.801_{-0.175}^{+0.175}$ & $115.177_{-0.246}^{+0.246}$ & $51.268_{-0.165}^{+0.165}$ & $328.816_{-0.797}^{+0.797}$ \\ 
                        \midrule
\multirow{14}{*}{\shortstack[l]{Component}}
& $\num{64}$ & $0.121_{-0.014}^{+0.014}$ & $0.248_{-0.017}^{+0.017}$ & $0.158_{-0.022}^{+0.022}$ & $0.024_{-0.003}^{+0.003}$ & $0.052_{-0.008}^{+0.008}$ & $0.013_{-0.000}^{+0.000}$ & $0.035_{-0.004}^{+0.004}$ \\ 
& $\num{128}$ & $0.146_{-0.016}^{+0.016}$ & $0.073_{-0.010}^{+0.010}$ & $0.110_{-0.016}^{+0.016}$ & $0.040_{-0.006}^{+0.006}$ & $0.026_{-0.001}^{+0.001}$ & $0.012_{-0.000}^{+0.000}$ & $0.043_{-0.003}^{+0.003}$ \\ 
& $\num{256}$ & $0.098_{-0.009}^{+0.009}$ & $0.174_{-0.019}^{+0.019}$ & $0.038_{-0.006}^{+0.006}$ & $0.058_{-0.009}^{+0.009}$ & $0.093_{-0.011}^{+0.011}$ & $0.013_{-0.000}^{+0.000}$ & $0.052_{-0.003}^{+0.003}$ \\ 
& $\num{512}$ & $0.133_{-0.012}^{+0.012}$ & $0.030_{-0.003}^{+0.003}$ & $0.048_{-0.008}^{+0.008}$ & $0.022_{-0.002}^{+0.002}$ & $0.127_{-0.015}^{+0.015}$ & $0.014_{-0.000}^{+0.000}$ & $0.069_{-0.002}^{+0.002}$ \\ 
& $\num{1024}$ & $0.145_{-0.009}^{+0.009}$ & $0.101_{-0.013}^{+0.013}$ & $0.024_{-0.002}^{+0.002}$ & $0.070_{-0.009}^{+0.009}$ & $0.162_{-0.015}^{+0.015}$ & $0.019_{-0.000}^{+0.000}$ & $0.125_{-0.001}^{+0.001}$ \\ 
& $\num{2048}$ & $0.225_{-0.008}^{+0.008}$ & $0.147_{-0.017}^{+0.017}$ & $0.100_{-0.013}^{+0.013}$ & $0.046_{-0.005}^{+0.005}$ & $0.191_{-0.012}^{+0.012}$ & $0.031_{-0.000}^{+0.000}$ & $0.250_{-0.001}^{+0.001}$ \\ 
& $\num{4096}$ & $0.475_{-0.021}^{+0.021}$ & $0.204_{-0.022}^{+0.022}$ & $0.046_{-0.003}^{+0.003}$ & $0.052_{-0.004}^{+0.004}$ & $0.270_{-0.011}^{+0.011}$ & $0.052_{-0.001}^{+0.001}$ & $0.517_{-0.003}^{+0.003}$ \\ 
& $\num{8192}$ & $0.876_{-0.015}^{+0.015}$ & $0.350_{-0.036}^{+0.036}$ & $0.126_{-0.013}^{+0.013}$ & $0.204_{-0.021}^{+0.021}$ & $0.641_{-0.025}^{+0.025}$ & $0.120_{-0.001}^{+0.001}$ & $1.100_{-0.010}^{+0.010}$ \\ 
& $\num{16384}$ & $1.627_{-0.024}^{+0.024}$ & $0.720_{-0.115}^{+0.115}$ & $0.297_{-0.024}^{+0.024}$ & $0.120_{-0.003}^{+0.003}$ & $1.275_{-0.047}^{+0.047}$ & $0.217_{-0.001}^{+0.001}$ & $2.227_{-0.015}^{+0.015}$ \\ 
& $\num{32768}$ & $3.075_{-0.029}^{+0.029}$ & $0.468_{-0.012}^{+0.012}$ & $0.305_{-0.016}^{+0.016}$ & $0.243_{-0.004}^{+0.004}$ & $2.230_{-0.070}^{+0.070}$ & $0.408_{-0.002}^{+0.002}$ & $4.507_{-0.026}^{+0.026}$ \\ 
& $\num{65536}$ & $6.291_{-0.126}^{+0.126}$ & $0.958_{-0.021}^{+0.021}$ & $1.071_{-0.091}^{+0.091}$ & $0.564_{-0.008}^{+0.008}$ & $4.366_{-0.134}^{+0.134}$ & $0.859_{-0.015}^{+0.015}$ & $9.113_{-0.081}^{+0.081}$ \\ 
& $\num{131072}$ & $12.214_{-0.198}^{+0.198}$ & $2.323_{-0.029}^{+0.029}$ & $1.589_{-0.041}^{+0.041}$ & $1.600_{-0.037}^{+0.037}$ & $8.717_{-0.286}^{+0.286}$ & $2.130_{-0.016}^{+0.016}$ & $18.474_{-0.168}^{+0.168}$ \\ 
& $\num{262144}$ & $23.209_{-0.208}^{+0.208}$ & $5.338_{-0.051}^{+0.051}$ & $3.207_{-0.003}^{+0.003}$ & $3.258_{-0.003}^{+0.003}$ & $16.913_{-0.510}^{+0.510}$ & $4.950_{-0.012}^{+0.012}$ & $36.091_{-0.241}^{+0.241}$ \\ 
& $\num{524288}$ & $44.925_{-0.264}^{+0.264}$ & $10.792_{-0.060}^{+0.060}$ & $6.822_{-0.013}^{+0.013}$ & $6.874_{-0.003}^{+0.003}$ & $34.539_{-0.701}^{+0.701}$ & $10.438_{-0.018}^{+0.018}$ & $69.282_{-0.390}^{+0.390}$ \\ 
			\bottomrule
		\end{widetable}
		\end{scriptsize}
	\end{sidewaystable*}
	\begin{sidewaystable*}
		\begin{scriptsize}
		\caption{Time needed for an \texttt{MPI\_\-Neighbor\_\-alltoall}\xspace exchange, different $k$-neighborhoods and reordering algorithms on \emph{VSC}$4$\xspace with $N=100$ and $p=48$. 
			Experiment performed as described in~\ref{sec:throughput}. We present the mean time in \SI{}{\milli\second} and the~$95\%$ confidence interval range. }
		\begin{widetable}{\textwidth}{l  r r r r r r r r}
			\toprule
			Stencil & Size [\SI{}{\byte}] & \texttt{Blocked}\xspace & \texttt{Hyperplane}\xspace & $k$-$d$~\texttt{tree}\xspace & \texttt{Stencil Strips}\xspace & \texttt{Nodecart}\xspace & \texttt{VieM}\xspace & \texttt{Random}\xspace \\
                        \midrule
\multirow{14}{*}{\shortstack[l]{Nearest\\ neighbor}}
& $\num{64}$ & $0.033_{-0.002}^{+0.002}$ & $0.028_{-0.002}^{+0.002}$ & $0.021_{-0.001}^{+0.001}$ & $0.057_{-0.008}^{+0.008}$ & $0.034_{-0.003}^{+0.003}$ & $0.064_{-0.007}^{+0.007}$ & $0.169_{-0.023}^{+0.023}$ \\ 
& $\num{128}$ & $0.039_{-0.002}^{+0.002}$ & $0.032_{-0.002}^{+0.002}$ & $0.023_{-0.001}^{+0.001}$ & $0.043_{-0.005}^{+0.005}$ & $0.033_{-0.002}^{+0.002}$ & $0.020_{-0.001}^{+0.001}$ & $0.165_{-0.020}^{+0.020}$ \\ 
& $\num{256}$ & $0.047_{-0.002}^{+0.002}$ & $0.043_{-0.004}^{+0.004}$ & $0.024_{-0.001}^{+0.001}$ & $0.048_{-0.005}^{+0.005}$ & $0.040_{-0.003}^{+0.003}$ & $0.035_{-0.002}^{+0.002}$ & $0.207_{-0.024}^{+0.024}$ \\ 
& $\num{512}$ & $0.071_{-0.003}^{+0.003}$ & $0.045_{-0.004}^{+0.004}$ & $0.035_{-0.001}^{+0.001}$ & $0.081_{-0.013}^{+0.013}$ & $0.049_{-0.003}^{+0.003}$ & $0.027_{-0.001}^{+0.001}$ & $0.178_{-0.010}^{+0.010}$ \\ 
& $\num{1024}$ & $0.111_{-0.003}^{+0.003}$ & $0.052_{-0.002}^{+0.002}$ & $0.178_{-0.022}^{+0.022}$ & $0.091_{-0.011}^{+0.011}$ & $0.066_{-0.002}^{+0.002}$ & $0.098_{-0.009}^{+0.009}$ & $0.321_{-0.015}^{+0.015}$ \\ 
& $\num{2048}$ & $0.207_{-0.004}^{+0.004}$ & $0.088_{-0.005}^{+0.005}$ & $0.086_{-0.001}^{+0.001}$ & $0.088_{-0.005}^{+0.005}$ & $0.104_{-0.001}^{+0.001}$ & $0.125_{-0.008}^{+0.008}$ & $0.603_{-0.014}^{+0.014}$ \\ 
& $\num{4096}$ & $17.399_{-3.241}^{+3.241}$ & $0.134_{-0.004}^{+0.004}$ & $0.150_{-0.001}^{+0.001}$ & $0.178_{-0.017}^{+0.017}$ & $0.213_{-0.008}^{+0.008}$ & $0.144_{-0.004}^{+0.004}$ & $1.257_{-0.025}^{+0.025}$ \\ 
& $\num{8192}$ & $16.526_{-3.260}^{+3.260}$ & $17.620_{-3.266}^{+3.266}$ & $0.345_{-0.003}^{+0.003}$ & $16.196_{-3.271}^{+3.271}$ & $19.444_{-3.314}^{+3.314}$ & $0.310_{-0.003}^{+0.003}$ & $2.849_{-0.095}^{+0.095}$ \\ 
& $\num{16384}$ & $49.466_{-0.046}^{+0.046}$ & $0.546_{-0.007}^{+0.007}$ & $0.644_{-0.006}^{+0.006}$ & $17.479_{-3.343}^{+3.343}$ & $19.636_{-3.342}^{+3.342}$ & $0.579_{-0.004}^{+0.004}$ & $5.243_{-0.056}^{+0.056}$ \\ 
& $\num{32768}$ & $49.348_{-0.018}^{+0.018}$ & $49.276_{-0.015}^{+0.015}$ & $1.277_{-0.012}^{+0.012}$ & $49.338_{-0.025}^{+0.025}$ & $49.348_{-0.012}^{+0.012}$ & $1.186_{-0.012}^{+0.012}$ & $11.390_{-0.229}^{+0.229}$ \\ 
& $\num{65536}$ & $5.787_{-0.036}^{+0.036}$ & $2.638_{-0.015}^{+0.015}$ & $3.061_{-0.020}^{+0.020}$ & $49.762_{-0.036}^{+0.036}$ & $3.577_{-0.033}^{+0.033}$ & $2.787_{-0.025}^{+0.025}$ & $23.361_{-0.480}^{+0.480}$ \\ 
& $\num{131072}$ & $11.792_{-0.079}^{+0.079}$ & $5.720_{-0.081}^{+0.081}$ & $6.374_{-0.041}^{+0.041}$ & $50.086_{-0.043}^{+0.043}$ & $7.049_{-0.047}^{+0.047}$ & $5.893_{-0.052}^{+0.052}$ & $42.202_{-0.409}^{+0.409}$ \\ 
& $\num{262144}$ & $22.827_{-0.114}^{+0.114}$ & $11.900_{-0.241}^{+0.241}$ & $12.763_{-0.089}^{+0.089}$ & $99.369_{-0.125}^{+0.125}$ & $13.894_{-0.086}^{+0.086}$ & $12.018_{-0.099}^{+0.099}$ & $84.227_{-0.731}^{+0.731}$ \\ 
& $\num{524288}$ & $44.733_{-0.167}^{+0.167}$ & $23.877_{-0.353}^{+0.353}$ & $25.459_{-0.183}^{+0.183}$ & $25.231_{-0.667}^{+0.667}$ & $27.661_{-0.128}^{+0.128}$ & $23.729_{-0.220}^{+0.220}$ & $165.619_{-1.292}^{+1.292}$ \\ 
                        \midrule
\multirow{14}{*}{\shortstack[l]{Nearest\\ neighbor\\ with hops}}
& $\num{64}$ & $0.030_{-0.001}^{+0.001}$ & $0.026_{-0.001}^{+0.001}$ & $0.029_{-0.002}^{+0.002}$ & $0.030_{-0.002}^{+0.002}$ & $0.029_{-0.001}^{+0.001}$ & $0.425_{-0.040}^{+0.040}$ & $0.054_{-0.002}^{+0.002}$ \\ 
& $\num{128}$ & $0.059_{-0.000}^{+0.000}$ & $0.034_{-0.000}^{+0.000}$ & $0.030_{-0.001}^{+0.001}$ & $0.031_{-0.001}^{+0.001}$ & $0.047_{-0.001}^{+0.001}$ & $0.034_{-0.002}^{+0.002}$ & $0.083_{-0.002}^{+0.002}$ \\ 
& $\num{256}$ & $0.096_{-0.000}^{+0.000}$ & $0.059_{-0.001}^{+0.001}$ & $0.042_{-0.001}^{+0.001}$ & $0.038_{-0.001}^{+0.001}$ & $0.071_{-0.000}^{+0.000}$ & $0.044_{-0.002}^{+0.002}$ & $0.135_{-0.001}^{+0.001}$ \\ 
& $\num{512}$ & $0.140_{-0.002}^{+0.002}$ & $0.085_{-0.001}^{+0.001}$ & $0.067_{-0.001}^{+0.001}$ & $0.061_{-0.001}^{+0.001}$ & $0.117_{-0.001}^{+0.001}$ & $0.062_{-0.001}^{+0.001}$ & $0.251_{-0.001}^{+0.001}$ \\ 
& $\num{1024}$ & $0.275_{-0.002}^{+0.002}$ & $0.157_{-0.001}^{+0.001}$ & $0.111_{-0.001}^{+0.001}$ & $0.108_{-0.001}^{+0.001}$ & $0.227_{-0.002}^{+0.002}$ & $0.103_{-0.001}^{+0.001}$ & $0.516_{-0.003}^{+0.003}$ \\ 
& $\num{2048}$ & $0.558_{-0.005}^{+0.005}$ & $0.293_{-0.001}^{+0.001}$ & $0.262_{-0.009}^{+0.009}$ & $0.177_{-0.002}^{+0.002}$ & $0.425_{-0.002}^{+0.002}$ & $0.176_{-0.003}^{+0.003}$ & $1.117_{-0.008}^{+0.008}$ \\ 
& $\num{4096}$ & $1.093_{-0.006}^{+0.006}$ & $0.593_{-0.002}^{+0.002}$ & $0.419_{-0.004}^{+0.004}$ & $0.335_{-0.002}^{+0.002}$ & $0.854_{-0.003}^{+0.003}$ & $0.397_{-0.013}^{+0.013}$ & $3.304_{-0.373}^{+0.373}$ \\ 
& $\num{8192}$ & $9.779_{-1.322}^{+1.322}$ & $1.233_{-0.005}^{+0.005}$ & $1.547_{-0.128}^{+0.128}$ & $0.860_{-0.003}^{+0.003}$ & $1.871_{-0.009}^{+0.009}$ & $0.936_{-0.011}^{+0.011}$ & $5.015_{-0.051}^{+0.051}$ \\ 
& $\num{16384}$ & $5.007_{-0.040}^{+0.040}$ & $3.893_{-0.431}^{+0.431}$ & $2.935_{-0.193}^{+0.193}$ & $1.778_{-0.090}^{+0.090}$ & $4.845_{-0.361}^{+0.361}$ & $1.906_{-0.016}^{+0.016}$ & $11.202_{-0.186}^{+0.186}$ \\ 
& $\num{32768}$ & $9.360_{-0.069}^{+0.069}$ & $5.151_{-0.018}^{+0.018}$ & $3.960_{-0.019}^{+0.019}$ & $3.481_{-0.009}^{+0.009}$ & $7.212_{-0.035}^{+0.035}$ & $3.835_{-0.026}^{+0.026}$ & $20.874_{-0.182}^{+0.182}$ \\ 
& $\num{65536}$ & $18.461_{-0.099}^{+0.099}$ & $10.410_{-0.049}^{+0.049}$ & $7.939_{-0.040}^{+0.040}$ & $6.866_{-0.016}^{+0.016}$ & $14.351_{-0.044}^{+0.044}$ & $7.796_{-0.154}^{+0.154}$ & $42.163_{-0.295}^{+0.295}$ \\ 
& $\num{131072}$ & $36.220_{-0.140}^{+0.140}$ & $20.761_{-0.072}^{+0.072}$ & $15.928_{-0.071}^{+0.071}$ & $13.605_{-0.035}^{+0.035}$ & $28.637_{-0.091}^{+0.091}$ & $14.838_{-0.145}^{+0.145}$ & $84.438_{-0.599}^{+0.599}$ \\ 
& $\num{262144}$ & $72.336_{-0.252}^{+0.252}$ & $41.625_{-0.138}^{+0.138}$ & $32.105_{-0.105}^{+0.105}$ & $27.084_{-0.092}^{+0.092}$ & $56.833_{-0.092}^{+0.092}$ & $28.781_{-0.097}^{+0.097}$ & $203.593_{-2.473}^{+2.473}$ \\ 
& $\num{524288}$ & $144.648_{-0.691}^{+0.691}$ & $82.193_{-0.091}^{+0.091}$ & $63.598_{-0.146}^{+0.146}$ & $52.156_{-0.126}^{+0.126}$ & $113.202_{-0.162}^{+0.162}$ & $57.651_{-0.144}^{+0.144}$ & $339.577_{-2.344}^{+2.344}$ \\ 
                        \midrule
\multirow{14}{*}{\shortstack[l]{Component}}
& $\num{64}$ & $0.041_{-0.003}^{+0.003}$ & $0.020_{-0.001}^{+0.001}$ & $0.014_{-0.001}^{+0.001}$ & $0.010_{-0.000}^{+0.000}$ & $0.164_{-0.022}^{+0.022}$ & $0.033_{-0.004}^{+0.004}$ & $0.081_{-0.017}^{+0.017}$ \\ 
& $\num{128}$ & $0.165_{-0.030}^{+0.030}$ & $0.020_{-0.001}^{+0.001}$ & $0.013_{-0.001}^{+0.001}$ & $0.010_{-0.000}^{+0.000}$ & $0.277_{-0.023}^{+0.023}$ & $0.056_{-0.011}^{+0.011}$ & $0.297_{-0.043}^{+0.043}$ \\ 
& $\num{256}$ & $0.045_{-0.002}^{+0.002}$ & $0.023_{-0.002}^{+0.002}$ & $0.012_{-0.001}^{+0.001}$ & $0.011_{-0.000}^{+0.000}$ & $0.639_{-0.016}^{+0.016}$ & $0.054_{-0.009}^{+0.009}$ & $0.193_{-0.022}^{+0.022}$ \\ 
& $\num{512}$ & $0.072_{-0.002}^{+0.002}$ & $0.028_{-0.002}^{+0.002}$ & $0.013_{-0.001}^{+0.001}$ & $0.128_{-0.016}^{+0.016}$ & $0.047_{-0.004}^{+0.004}$ & $0.159_{-0.023}^{+0.023}$ & $0.120_{-0.011}^{+0.011}$ \\ 
& $\num{1024}$ & $0.111_{-0.002}^{+0.002}$ & $0.036_{-0.001}^{+0.001}$ & $0.017_{-0.001}^{+0.001}$ & $0.016_{-0.001}^{+0.001}$ & $0.089_{-0.007}^{+0.007}$ & $0.015_{-0.000}^{+0.000}$ & $0.156_{-0.005}^{+0.005}$ \\ 
& $\num{2048}$ & $0.194_{-0.003}^{+0.003}$ & $0.060_{-0.001}^{+0.001}$ & $0.181_{-0.016}^{+0.016}$ & $0.024_{-0.001}^{+0.001}$ & $0.138_{-0.008}^{+0.008}$ & $0.092_{-0.014}^{+0.014}$ & $0.321_{-0.013}^{+0.013}$ \\ 
& $\num{4096}$ & $0.389_{-0.006}^{+0.006}$ & $0.119_{-0.004}^{+0.004}$ & $0.037_{-0.001}^{+0.001}$ & $0.240_{-0.026}^{+0.026}$ & $0.210_{-0.007}^{+0.007}$ & $0.348_{-0.036}^{+0.036}$ & $1.321_{-0.140}^{+0.140}$ \\ 
& $\num{8192}$ & $0.897_{-0.012}^{+0.012}$ & $0.313_{-0.015}^{+0.015}$ & $0.062_{-0.001}^{+0.001}$ & $0.061_{-0.001}^{+0.001}$ & $0.537_{-0.023}^{+0.023}$ & $0.108_{-0.008}^{+0.008}$ & $1.351_{-0.027}^{+0.027}$ \\ 
& $\num{16384}$ & $1.615_{-0.017}^{+0.017}$ & $0.562_{-0.025}^{+0.025}$ & $0.129_{-0.006}^{+0.006}$ & $0.167_{-0.013}^{+0.013}$ & $1.008_{-0.038}^{+0.038}$ & $0.151_{-0.007}^{+0.007}$ & $2.958_{-0.100}^{+0.100}$ \\ 
& $\num{32768}$ & $3.163_{-0.020}^{+0.020}$ & $1.042_{-0.045}^{+0.045}$ & $0.231_{-0.002}^{+0.002}$ & $0.232_{-0.001}^{+0.001}$ & $1.827_{-0.046}^{+0.046}$ & $0.257_{-0.004}^{+0.004}$ & $5.673_{-0.146}^{+0.146}$ \\ 
& $\num{65536}$ & $5.558_{-0.043}^{+0.043}$ & $3.435_{-0.331}^{+0.331}$ & $0.531_{-0.002}^{+0.002}$ & $0.546_{-0.002}^{+0.002}$ & $3.338_{-0.042}^{+0.042}$ & $0.779_{-0.040}^{+0.040}$ & $13.893_{-0.791}^{+0.791}$ \\ 
& $\num{131072}$ & $10.943_{-0.081}^{+0.081}$ & $3.984_{-0.101}^{+0.101}$ & $1.400_{-0.002}^{+0.002}$ & $1.397_{-0.002}^{+0.002}$ & $6.515_{-0.038}^{+0.038}$ & $2.595_{-0.285}^{+0.285}$ & $26.735_{-1.219}^{+1.219}$ \\ 
& $\num{262144}$ & $21.678_{-0.132}^{+0.132}$ & $9.878_{-0.611}^{+0.611}$ & $3.189_{-0.002}^{+0.002}$ & $3.189_{-0.002}^{+0.002}$ & $12.989_{-0.101}^{+0.101}$ & $3.316_{-0.023}^{+0.023}$ & $54.507_{-2.131}^{+2.131}$ \\ 
& $\num{524288}$ & $44.034_{-0.220}^{+0.220}$ & $15.972_{-0.240}^{+0.240}$ & $6.859_{-0.003}^{+0.003}$ & $6.765_{-0.004}^{+0.004}$ & $25.547_{-0.052}^{+0.052}$ & $7.151_{-0.126}^{+0.126}$ & $112.673_{-4.697}^{+4.697}$ \\ 
			\bottomrule
		\end{widetable}
		\end{scriptsize}
	\end{sidewaystable*}
\end{document}